\documentclass[sn-mathphys,Numbered]{sn-jnl}

\usepackage{graphicx}%
\usepackage{multirow}%
\usepackage{amsmath,bm,bbm,amssymb,amsfonts}%
\usepackage{amsthm}%
\usepackage{mathrsfs}%
\usepackage[title]{appendix}%
\usepackage{xcolor}%
\usepackage{textcomp}%
\usepackage{manyfoot}%
\usepackage{booktabs}%
\usepackage{algorithm}%
\usepackage{algorithmicx}%
\usepackage{algpseudocode}%
\usepackage{listings}%

\DeclareMathOperator*{\argmin}{arg\,min}
\DeclareMathOperator*{\argmax}{arg\,max}

\theoremstyle{thmstyleone}%
\newtheorem{theorem}{Theorem}
\theoremstyle{thmstyletwo}%
\newtheorem{remark}{Remark}%

\theoremstyle{thmstylethree}%

\begin{document}

\title[Compressed Sensing: A Discrete Optimization Approach]{Compressed Sensing: A Discrete Optimization Approach}


\author[1]{\fnm{Dimitris} \sur{Bertsimas}}\email{dbertsim@mit.edu}
\equalcont{These authors contributed equally to this work.}

\author*[2]{\fnm{Nicholas} \sur{Johnson}}\email{nagj@mit.edu}
\equalcont{These authors contributed equally to this work.}

\affil[1]{\orgdiv{Sloan School of Management}, \orgname{Massachusetts Institute of Technology}, \orgaddress{\street{100 Main Street}, \city{Cambridge}, \postcode{02142}, \state{MA}, \country{USA}}}

\affil[2]{\orgdiv{Operations Research Center}, \orgname{Massachusetts Institute of Technology}, \orgaddress{\street{1 Amherst Street}, \city{Cambridge}, \postcode{02142}, \state{MA}, \country{USA}}}


\abstract{We study the Compressed Sensing (CS) problem, which is the problem of finding the most sparse vector that satisfies a set of linear measurements up to some numerical tolerance. CS is a central problem in Statistics, Operations Research and Machine Learning which arises in applications such as signal processing, data compression{\color{black}, image reconstruction, and multi-label learning}. We introduce an $\ell_2$ regularized formulation of CS which we reformulate as a mixed integer second order cone program. We derive a second order cone relaxation of this problem and show that under mild conditions on the regularization parameter, the resulting relaxation is equivalent to the well studied basis pursuit denoising problem. We present a semidefinite relaxation that strengthens the second order cone relaxation and develop a custom branch-and-bound algorithm that leverages our second order cone relaxation to solve {\color{black}small-scale }instances of CS to certifiable optimality. 
{\color{black} When compared against solutions produced by three state of the art benchmark methods on synthetic data, our numerical results show that our approach produces solutions that are on average $6.22\%$ {\color{black}more sparse.} When compared only against the experiment-wise best performing benchmark method on synthetic data, our approach produces solutions that are on average $3.10\%$ more sparse.} On real world ECG data, for a given $\ell_2$ reconstruction error our approach produces solutions that are on average $9.95\%$ more sparse than benchmark methods ($3.88\%$ more sparse if only compared against the best performing benchmark), while for a given sparsity level our approach produces solutions that have on average $10.77\%$ {\color{black}lower reconstruction error than benchmark methods} ($1.42\%$ lower error if only compared against the best performing benchmark). When used as a component of a multi-label classification algorithm, our approach achieves greater classification accuracy than benchmark compressed sensing methods. {\color{black}This improved accuracy comes at the cost of an increase in computation time by several orders of magnitude.} Thus, for applications where runtime is not of critical importance, leveraging integer optimization can yield sparser and lower error solutions to CS than existing benchmarks.}




\maketitle

\section{Introduction}\label{sec1}

The \textit{Compressed Sensing} (CS) problem seeks to find a most sparse vector $\bm{x} \in \mathbb{R}^n$ that is consistent with a set of $m$ linear equalities. 
CS is a fundamental problem in Statistics, Operations Research and Machine Learning which arises in numerous applications such as multi-label learning \citep{hsu2009multi}, medical resonance imaging \citep{lustig2007sparse}, holography \citep{brady2009compressive}, climate monitoring \citep{hashemi2016efficient}, natural resource mining \citep{wang2018design} and electrocardiogram signal acquisition \citep{chen2019compressed} among many others. Formally, given a matrix $\bm{A} \in \mathbb{R}^{m \times n}$ and a vector $\bm{b} \in \mathbb{R}^m$, CS is given by \citep{donoho2006compressed}:
\begin{equation}
\begin{aligned}
    \min_{\bm{x} \in \mathbb{R}^n} \Vert\bm{x}\Vert_0\
    \text{s.t.} \ \bm{A}\bm{x} = \bm{b}.
\end{aligned} \label{opt:classic_problem}
\end{equation} In the presence of noisy measurements, it is necessary to relax the equality constraint in \eqref{opt:classic_problem}, leading to the following formulation for $\epsilon > 0$:
\begin{equation}
\begin{aligned}
    \min_{\bm{x} \in \mathbb{R}^n} \Vert\bm{x}\Vert_0\
    \text{s.t.} \ \Vert\bm{A}\bm{x} - \bm{b}\Vert_2^2 \leq \epsilon.
\end{aligned} \label{opt:main_problem}
\end{equation} This problem is sometimes referred to as sparse approximation in the literature \citep{tropp2010computational} and trivially reduces to \eqref{opt:classic_problem} for $\epsilon = 0$. CS allows signals to be reconstructed surprisingly well after sampling at a rate far below the Nyquist sampling rate by leveraging the inherent sparsity of most signals, either in the signal's latent space or in an appropriately defined transform space. For example, natural images tend to have a sparse representation in the wavelet domain, speech can be represented using a small number of coefficients in the Fourier transform domain and medical images can be represented sparsely in the Radon transform domain \citep{rani2018systematic}.

In Section \ref{sec:lit_review}, we will see that the vast majority of existing approaches to CS either rely on $\ell_1$ based convex approximations to \eqref{opt:main_problem} or are greedy heuristics whereas the use of integer optimization techniques has gone relatively unexplored. In this work, we formulate CS as:
\begin{equation}
\begin{aligned}
    \min_{\bm{x} \in \mathbb{R}^n} \Vert\bm{x}\Vert_0 + \frac{1}{\gamma} \Vert\bm{x}\Vert_2^2\
    \text{s.t.} \ \Vert\bm{A}\bm{x} - \bm{b}\Vert_2^2 \leq \epsilon,
\end{aligned} \label{opt:our_formulation}
\end{equation} where $\gamma > 0$ is a regularization parameter that in practice can either take a default value (e.g. $\sqrt{n}$) or be cross-validated by minimizing a validation metric \citep[see, e.g.,][]{validation} to obtain strong out-of-sample performance \citep{bousquet2002stability}. A defining characteristic of the approach we present in this work is that we leverage techniques from integer optimization to exploit the inherent discreteness of formulation \eqref{opt:our_formulation} rather than relying on more commonly studied approximate methods. Note that Problem \eqref{opt:our_formulation} is a special case of the formulation given by:
\begin{equation}
\begin{aligned}
    \min_{\bm{x} \in \mathbb{R}^n}  \Vert\bm{x}\Vert_0 + \frac{1}{\gamma} \Vert\bm{W}\bm{x}\Vert_2^2\
    \text{s.t.} \ \Vert\bm{A}\bm{x} - \bm{b}\Vert_2^2 \leq \epsilon.
\end{aligned} \label{opt:our_formulation_W}
\end{equation} where $\bm{W} \in \mathbb{R}^{n \times n}$ is a diagonal matrix with nonnegative diagonal entries that should be interpreted as coordinate weights on the vector $\bm{x}$. Indeed, \eqref{opt:our_formulation_W} reduces to \eqref{opt:our_formulation} when we take $\bm{W} = \bm{I}$.


\subsection{Contributions and Structure}
In this paper, we approach CS using mixed integer second order cone optimization. We derive a second order cone relaxation of this problem and show that under mild conditions on the regularization parameter, the resulting relaxation is equivalent to the well studied basis pursuit denoising problem. We present a semidefinite relaxation that strengthens the second order cone relaxation and develop a custom branch-and-bound algorithm that leverages our second order cone relaxation to solve instances of CS to certifiable optimality. Our numerical results show that our approach produces solutions that are on average $6.22\%$ more sparse than solutions returned by three state of the art benchmark methods on synthetic data in minutes. If we restrict the comparison {\color{black}to the best performing benchmark method on each problem instance,} our approach produces solutions that are on average $3.10\%$ more sparse. On real world ECG data, for a given $\ell_2$ reconstruction error our approach produces solutions that are on average $9.95\%$ more sparse than benchmark methods {\color{black}($3.88\%$ more sparse if only compared against the best performing benchmark)}, while for a given sparsity level our approach produces solutions that have on average $10.77\%$ lower reconstruction error than benchmark methods in minutes {\color{black}($1.42\%$ lower error if only compared against the best performing benchmark). On a real world multi-label classification task, our approach outperforms existing approaches in terms of accuracy, precision and recall. {\color{black}This increase in accuracy, precision and recall comes at the expense of a significant increase in running time of several orders of magnitude.} Thus, for applications where runtime is not of critical importance, leveraging integer optimization can yield sparser and lower error solutions to CS than existing benchmarks.}

The rest of the paper is structured as follows. In Section \ref{sec:lit_review}, we review existing formulations and solution methods of the CS problem. In Section \ref{sec:formulation_properties}, we study how our regularized formulation of CS \eqref{opt:our_formulation} connects to the commonly used formulation \eqref{opt:main_problem}. We reformulate \eqref{opt:our_formulation} exactly as a mixed integer second order cone problem in Section \ref{sec:MISOC} and present a second order cone relaxation in Section \ref{ssec:SOC} and a stronger but more computationally expensive semidefinite cone relaxation in Section \ref{ssec:SDP}. We show that our second order cone relaxation is equivalent to the Basis Pursuit Denoising problem under mild conditions offering a new interpretation of this well studied method as a convex relaxation of our mixed integer second order cone reformulation of \eqref{opt:our_formulation}. We leverage our second order cone relaxation to develop a custom branch-and-bound algorithm in Section \ref{sec:BnB} that can solve instances of \eqref{opt:our_formulation} to certifiable optimality. In Section \ref{sec:experiments}, we investigate the performance of our branch-and-bound algorithm against state of the art benchmark methods on synthetic {\color{black}data, real world ECG signal acquisition and real world multi-label classification}. 

\paragraph{Notation:} We let nonbold face characters such as $b$ denote scalars, lowercase bold faced characters such as $\bm{x}$ denote vectors, uppercase bold faced characters such as $\bm{X}$ denote matrices, and calligraphic uppercase characters such as $\mathcal{Z}$ denote sets. We let $[n]$ denote the set of running indices $\{1, \ldots,  n\}$. We let $\bm{1}_n$ denote an {\color{black}$n$-dimensional} vector of all $1$'s, $\bm{0}_n$ denote an {\color{black}$n$-dimensional} vector of all $0$'s, and $\bm{I}$ denote the identity matrix. We let $\mathcal{S}^n$ denote the {\color{black}subspace} of $n \times n$ symmetric matrices and {\color{black}$\mathcal{S}^n_+ \in \mathcal{S}^n$} denote the cone of {\color{black}symmetric} $n \times n$ positive semidefinite matrices.

\section{Literature Review} \label{sec:lit_review}

In this section, we review several key approaches from the literature that have been employed to solve the CS problem. As an exhaustive literature review is outside of the scope of this paper, we focus our review on a handful of well studied approaches which will be used as benchmarks in this work. For a more detailed CS literature review, we refer the reader to \cite{tropp2010computational}. 

The majority of existing approaches to the CS problem are heuristic in nature and generally can be classified as either convex approximations or greedy methods as we will see in this section. For these methods, associated performance guarantees require making strong statistical assumptions on the underlying problem data. Integer optimization has been given little attention in the CS literature despite its powerful modelling capabilities. \cite{karahanouglu2013mixed} and \cite{bourguignon2015exact} explore formulating Problem \eqref{opt:main_problem} as a mixed integer linear program for the case when $\epsilon=0$. However this approach relies on using the big-$M$ method which requires estimating reasonable values for $M$ and cannot immediately generalize to the setting where $\epsilon > 0$.

\subsection{Basis Pursuit Denoising}

A common class of CS methods rely on solving convex approximations of \eqref{opt:main_problem} rather than attempting to solve \eqref{opt:main_problem} directly. A popular approach is to use the $\ell_1$ norm as a convex surrogate for the $\ell_0$ norm \citep{chen1994basis, chen2001atomic,donoho2006compressed,candes2006near,gill2011crowd}. This approximation is typically motivated by the observation that the unit $\ell_1$ ball given by $\mathcal{B}_{\ell_1} = \{\bm{x} \in \mathbb{R}^n: \Vert \bm{x} \Vert_1 \leq 1\}$ is the convex hull of the nonconvex set $\mathcal{X} = \{\bm{x} \in \mathbb{R}^n: \Vert \bm{x} \Vert_0 \leq 1, \Vert \bm{x} \Vert_\infty \leq 1\}$. Replacing the $\ell_0$ norm by the $\ell_1$ norm in \eqref{opt:main_problem}, we obtain:
\begin{equation}
\begin{aligned}
    \min_{\bm{x} \in \mathbb{R}^n} \Vert\bm{x}\Vert_1\
    \text{s.t.} \ \Vert\bm{A}\bm{x} - \bm{b}\Vert_2^2 \leq \epsilon.
\end{aligned} \label{opt:BPD}
\end{equation} Problem \eqref{opt:BPD} is referred to as Basis Pursuit Denoising and is a quadratically constrained convex optimization problem which can be solved efficiently using one of several off the shelf optimization packages. Basis Pursuit Denoising produces an approximate solution to Problem \eqref{opt:main_problem} by either directly returning the solution of \eqref{opt:BPD} or by post-processing the solution of \eqref{opt:BPD} to further sparsify the result. One such post-processing technique is a greedy rounding mechanism where columns of the matrix $\bm{A}$ are iteratively selected in the order corresponding to decreasing magnitude of the entries of the optimal solution of \eqref{opt:BPD} until the selected column set of $\bm{A}$ is sufficiently large to produce a feasible solution to \eqref{opt:main_problem}. Basis Pursuit Denoising is very closely related to the Lasso problem which is given by:
\begin{equation}
\begin{aligned}
    &\min_{\bm{x} \in \mathbb{R}^n} & & \Vert\bm{A}\bm{x} - \bm{b}\Vert_2^2 + \lambda \Vert\bm{x}\Vert_1,
\end{aligned} \label{opt:Lasso}
\end{equation} where $\lambda > 0$ is a tunable hyperparameter. Lasso is a statistical estimator commonly used for sparse regression as empirically, the optimal solution of Problem \eqref{opt:Lasso} tends to be sparse \citep{tibshirani1996regression}. More recently, strong connections between Lasso and robust optimization have been established \citep{bertsimas2018characterization}. Basis Pursuit Denoising and Lasso are equivalent in that Lasso is obtained by relaxing the hard constraint in \eqref{opt:BPD} and instead introducing a penalty term in the objective function. It is straightforward to show that for given input data $\bm{A}, \bm{b}$ and $\epsilon$ in \eqref{opt:BPD}, there exists a value $\lambda^\star > 0$ such that there exists a solution $\bm{x}^\star$ that is both optimal for \eqref{opt:BPD} and \eqref{opt:Lasso} when the tunable parameter takes value $\lambda = \lambda^\star$.

Note that by taking $\epsilon = 0$, Problem \eqref{opt:BPD} reduces to the well studied Basis Pursuit problem where the equality constraint $\bm{A}\bm{x}=\bm{b}$ is enforced. A large body of work studies conditions under which the optimal solution of the Basis Pursuit problem is also an optimal solution of \eqref{opt:classic_problem}. For example, see \cite{elad2002generalized}, \cite{donoho2003optimally}, \cite{gribonval2003sparse}, and \cite{tropp2004greed}. One of the most well studied conditions under which this equivalence holds is when the input matrix $\bm{A}$ satisfies the Restricted Isometry Property (RIP). Formally, a matrix $\bm{A} \in \mathbb{R}^{m \times n}$ is said to satisfy RIP of order $s$ and parameter $\delta_s \in (0, 1)$ if for every vector $\bm{x} \in \mathbb{R}^n$ such that $\Vert \bm{x} \Vert_0 \leq s$, we have
\[(1-\delta_s) \Vert \bm{x} \Vert_2^2 \leq \Vert \bm{A}\bm{x} \Vert_2^2 \leq (1+\delta_s) \Vert \bm{x} \Vert_2^2.\]
It has been established that if $\bm{A}$ satisfies RIP {\color{black}of} order $2s$ and parameter $\delta_{2s} < 1/3$, then the optimal solution of the Basis Pursuit problem is also an optimal solution of \eqref{opt:classic_problem} where $s$ denotes the cardinality of this optimal solution \citep{candes2005decoding}. While it has been shown that certain random matrices satisfy this desired RIP property with high probability \citep{baraniuk2008simple,guedon2014restricted}, RIP in general is not tractable to verify on arbitrary real world data.

\subsection{Iterative Reweighted L1} \label{sec:IRWL1}

Iterative Reweighted $\ell_1$ minimization is an iterative method that can generate an approximate solution to \eqref{opt:main_problem} by solving a sequence of convex optimization problems that are very closely related to the Basis Pursuit Denoising problem given by \eqref{opt:BPD} \citep{candes2008enhancing, needell2009noisy, asif2013fast}. This approach falls in the class of convex approximation based methods for solving CS. The approach considers the weighted $\ell_1$ minimization problem given by:
\begin{equation}
\begin{aligned}
    &\min_{\bm{x} \in \mathbb{R}^n} & & \Vert \bm{W}\bm{x}\Vert_1\\
    &\text{s.t.} & & \Vert\bm{A}\bm{x} - \bm{b}\Vert_2^2 \leq \epsilon,
\end{aligned} \label{opt:BPD_W}
\end{equation} where $\bm{W} \in \mathbb{R}^{n \times n}$ is a diagonal matrix with nonnegative diagonal entries. Each diagonal entry $W_{ii} = w_i$ of $\bm{W}$ can be interpreted as a weighting of the $i^{th}$ coordinate of the vector $\bm{x}$. Interpreting the $\ell_1$ norm as a convex surrogate for the $\ell_0$ norm, Problem \eqref{opt:BPD_W} can be viewed as a relaxation of the nonconvex problem given by 
\begin{equation}
\begin{aligned}
    &\min_{\bm{x} \in \mathbb{R}^n} & & \Vert \bm{W}\bm{x}\Vert_0\\
    &\text{s.t.} & & \Vert\bm{A}\bm{x} - \bm{b}\Vert_2^2 \leq \epsilon.
\end{aligned} \label{opt:main_problem_W}
\end{equation} It is trivial to verify that when $\bm{W} = \alpha \bm{I}$, where $\alpha > 0$ and $\bm{I}$ is the $n$-by-$n$ identity matrix, \eqref{opt:main_problem_W} and \eqref{opt:BPD_W} reduce exactly to \eqref{opt:main_problem} and \eqref{opt:BPD} respectively. Assuming the weights never vanish, the nonconvex Problems \eqref{opt:main_problem} and \eqref{opt:main_problem_W} have the same optimal solution, yet their convex relaxations \eqref{opt:BPD} and \eqref{opt:BPD_W} will generally have very different solutions. In this regard, the weights can be regarded as parameters that if chosen correctly can produce a better solution than \eqref{opt:BPD}. Iterative Reweighted $\ell_1$ minimization proceeds as follows \citep{candes2008enhancing}:
\begin{enumerate}
    \item Initialize the iteration count $t \xleftarrow[]{} 0$ and the weights $w_{i}^{(0)} \xleftarrow[]{} 1$.
    \item Solve \eqref{opt:BPD_W} with $\bm{W} = \bm{W}^{(t)}$. Let $\bm{x}^{(t)}$ denote the optimal solution.
    \item Update the weights as $w_i^{(t+1)} \xleftarrow[]{} \frac{1}{\vert x_i^{(t)} \vert + \delta}$ where $\delta > 0$ is a fixed parameter for numerical stability. 
    \item Terminate if $t$ reaches a maximum number of iterations or if the iterates $\bm{x}^{(t)}$ have converged. Otherwise, increment $t$ and return to Step $2$.
\end{enumerate} {\color{black}It has been shown} empirically that in many settings the solution returned by Iterated Reweighted $\ell_1$ minimization outperforms the solution returned by Basis Pursuit Denoising by recovering the true underlying signal while requiring fewer measurements to be taken {\color{black}\cite{candes2008enhancing}}. We note that this approach is an instance of a broader class of sparsifying iterative reweighted methods \citep{chen2010convergence, wang2021nonconvex, wang2022extrapolated}.

\subsection{Orthogonal Matching Pursuit} \label{sec:OMP}

Orthogonal Matching Pursuit (OMP) is a canonical greedy algorithm for obtaining heuristic solutions to \eqref{opt:main_problem} \citep{pati1993orthogonal, mallat1993matching}. Solving Problem \eqref{opt:main_problem} can be interpreted as determining the minimum number of columns from the input matrix $\bm{A}$ that must be selected such that the residual of the projection of the input vector $\bm{b}$ onto the span of the selected columns has $\ell_2$ norm equal to at most $\sqrt{\epsilon}$. The OMP algorithm proceeds by first selecting the column of $\bm{A}$ that is most collinear with $\bm{b}$ and subsequently iteratively adding the column of $\bm{A}$ that is most collinear with the residual of the projection of $\bm{b}$ onto the subspace spanned by the selected columns until the norm of this residual is at most $\sqrt{\epsilon}$. Concretely, OMP proceeds as follows where for an arbitrary collection of indices $\mathcal{I}_t \subseteq [n]$, we let $\bm{A}(\mathcal{I}_t) \in \mathbb{R}^{m \times \vert  \mathcal{I}_t \vert}$ denote the matrix obtained by stacking the $\vert \mathcal{I}_t \vert$ columns of $\bm{A}$ corresponding to the indices in the set $\mathcal{I}_t$:
\begin{enumerate}
    \item Initialize the iteration count $t \xleftarrow[]{} 0$, the residual $\bm{r}_0 \xleftarrow[]{} \bm{b}$ and the index set $\mathcal{I}_0 \xleftarrow[]{} \emptyset$.
    \item Select the column that is most collinear with the residual $i_t \xleftarrow[]{} \argmax_{i \in [n]\setminus \mathcal{I}_t} \vert \bm{a}_i^T\bm{r}_t \vert$ and update the index set $\mathcal{I}_{t+1} \xleftarrow[]{} \mathcal{I}_t \cup i_t$.
    \item Compute the projection of $\bm{b}$ onto the current set of columns \[\bm{x}_{t+1} \xleftarrow[]{} \big{[}\bm{A}(\mathcal{I}_{t+1})^T\bm{A}(\mathcal{I}_{t+1})\big{]}^\dagger\bm{A}(\mathcal{I}_{t+1})^T\bm{b},\] and update the residual $\bm{r}_{t+1} \xleftarrow[]{} \bm{b} - \bm{A}(\mathcal{I}_{t+1})\bm{x}_{t+1}$.
    \item Terminate if $\Vert \bm{r}_{t+1} \Vert_2^2 \leq \epsilon$, otherwise increment $t$ and return to Step 2.
\end{enumerate} Conditions under which the solution returned by OMP is the optimal solution of \eqref{opt:main_problem} (either with high probability or with certainty) have been studied extensively \citep{tropp2004greed, cai2011orthogonal, wang2015support}. Unfortunately, these conditions suffer from the same limitation as RIP in that in general they are not tractable to verify on real world data. A closely related method to OMP is Subspace Pursuit (SP) which is another greedy algorithm for obtaining a heuristic solution to \eqref{opt:main_problem} in the $\epsilon=0$ setting but has the additional requirement that a target sparsity value $K$ must be specified in advance \citep{dai2009subspace}. SP is initialized by selecting the $K$ columns of $\bm{A}$ that are most collinear with the vector $\bm{b}$. At each iteration, SP first computes the residual of the projection of $\bm{b}$ onto the current column set and then greedily updates up to $K$ elements of the column set, repeating this process until doing so no longer decreases the norm of the residual.

\section{Formulation Properties} \label{sec:formulation_properties}

In this section, we rigorously investigate connections between formulations \eqref{opt:our_formulation} and \eqref{opt:main_problem} for the CS problem in the noisy setting. The only difference between formulations \eqref{opt:main_problem} and \eqref{opt:our_formulation} is the inclusion of a $\ell_2$ regularization term in the objective function in \eqref{opt:our_formulation}. We will see in Section \ref{sec:MISOC} that the presence of this regularization term facilitates useful reformulations. Moreover, in the case of regression, \cite{bertsimas2018characterization} show that augmenting the ordinary least squares objective function with a $\ell_2$ regularization penalty produces regression vectors that are robust against data perturbations which suggests the presence of such a regularization term may result in a similar benefit in \eqref{opt:our_formulation}. A natural question to ask is: under what conditions do problems \eqref{opt:main_problem} and \eqref{opt:our_formulation} have the same solution? We answer this question in Theorem \ref{thm:large_gamma}.

\begin{theorem} \label{thm:large_gamma}

There exists a finite value $\gamma_0 < \infty$ such that for all $\Bar{\gamma} \geq \gamma_0$, there exists a vector $\bm{x}^\star$ such that $\bm{x}^\star$ is an optimal solution of \eqref{opt:main_problem} and also an optimal solution of \eqref{opt:our_formulation} where we set $\gamma = \Bar{\gamma}$. Letting $\Tilde{\bm{x}}$ denote a minimum norm solution to \eqref{opt:main_problem}, we can take $\gamma_0 = \Vert \Tilde{\bm{x}} \Vert_2^2$ and $\bm{x}^\star = \Tilde{\bm{x}}$. \end{theorem} Phrased simply, Theorem \ref{thm:large_gamma} establishes that there exists a finite value $\gamma_0$ such that if the regularization parameter $\gamma$ in problem \eqref{opt:our_formulation} is at least as large as $\gamma_0$, then there is a vector $\bm{x}^\star$ that is optimal to both problems \eqref{opt:main_problem} and \eqref{opt:our_formulation}. We note that this finite value $\gamma_0$ depends on the input data $\bm{A}, \bm{b}$ and $\epsilon$.

\begin{proof}
Consider any matrix $\bm{A} \in \mathbb{R}^{m \times n}$, vector $\bm{b} \in \mathbb{R}^m$ and scalar $\epsilon > 0$. Let $\Omega$ denote the set of optimal solutions to \eqref{opt:main_problem} and let $\mathcal{X}$ denote the feasible set of \eqref{opt:main_problem} and \eqref{opt:our_formulation}. We have $\mathcal{X} = \{\bm{x} : \Vert \bm{A}\bm{x}-\bm{b}\Vert_2^2 \leq \epsilon \}$ and $\Omega \subseteq \mathcal{X}$. Let 
$\Tilde{\bm{x}} \in \argmin_{\bm{x} \in \Omega} \Vert \bm{x} \Vert_2^2$ and let $\gamma_0 = \Vert \Tilde{\bm{x}} \Vert_2^2$. Since $\Tilde{\bm{x}} \in \Omega$, $\Tilde{\bm{x}}$ is an optimal solution to \eqref{opt:main_problem}. It remains to show that $\Tilde{\bm{x}}$ is optimal to \eqref{opt:our_formulation} for all $\gamma \geq \gamma_0$.

Fix any $\gamma \geq \gamma_0$. To show that $\Tilde{\bm{x}}$ is an optimal solution of \eqref{opt:our_formulation}, we will show that for all $\Bar{\bm{x}} \in \mathcal{X}$, we have
\[\Vert \Tilde{\bm{x}} \Vert_0 + \frac{1}{\gamma} \Vert \Tilde{\bm{x}} \Vert_2^2 \leq \Vert \Bar{\bm{x}} \Vert_0 + \frac{1}{\gamma} \Vert \Bar{\bm{x}} \Vert_2^2.\] Fix an arbitrary $\Bar{\bm{x}} \in \mathcal{X}$. Either $\Bar{\bm{x}} \in \mathcal{X} \setminus \Omega$ or $\Bar{\bm{x}} \in \Omega$. Suppose $\Bar{\bm{x}} \in \mathcal{X} \setminus \Omega$. The definition of $\Omega$ and the fact that $\Tilde{\bm{x}} \in \Omega$ implies 
\[\Vert \Tilde{\bm{x}} \Vert_0 < \Vert \Bar{\bm{x}} \Vert_0 \implies \Vert \Tilde{\bm{x}} \Vert_0 + 1 \leq \Vert \Bar{\bm{x}} \Vert_0.\] Next, note that since $\gamma \geq \gamma_0 = \Vert \Tilde{\bm{x}} \Vert_2^2$, we have 
\[\Vert \Tilde{\bm{x}} \Vert_0 + \frac{1}{\gamma} \Vert \Tilde{\bm{x}} \Vert_2^2 \leq \Vert \Tilde{\bm{x}} \Vert_0 + 1 \leq \Vert \Bar{\bm{x}} \Vert_0 \leq \Vert \Bar{\bm{x}} \Vert_0 + \frac{1}{\gamma} \Vert \Bar{\bm{x}} \Vert_2^2.\] Suppose instead that $\Bar{\bm{x}} \in \Omega$. The definition of $\Omega$ and $\Tilde{\bm{x}}$ imply $\Vert \Tilde{\bm{x}} \Vert_0 = \Vert \Bar{\bm{x}} \Vert_0$ and $\Vert \Tilde{\bm{x}} \Vert_2^2 \leq \Vert \Bar{\bm{x}} \Vert_2^2$. It then follows immediately that $\Vert \Tilde{\bm{x}} \Vert_0 + \frac{1}{\gamma} \Vert \Tilde{\bm{x}} \Vert_2^2 \leq \Vert \Bar{\bm{x}} \Vert_0 + \frac{1}{\gamma} \Vert \Bar{\bm{x}} \Vert_2^2$. Thus, $\Tilde{\bm{x}}$ is optimal to \eqref{opt:our_formulation}. This completes the proof. \end{proof}
Though Theorem \ref{thm:large_gamma} is useful in establishing conditions for the equivalence of problems \eqref{opt:main_problem} and \eqref{opt:our_formulation}, it is important to note that computing the value of $\gamma_0$ specified in the Theorem requires solving \eqref{opt:main_problem} which is difficult in general. Suppose we are solving problem \eqref{opt:our_formulation} with some regularization parameter $\gamma$ in the regime where $0 < \gamma < \gamma_0$. A natural question to ask is: how well does the solution of \eqref{opt:our_formulation} approximate the solution of \eqref{opt:main_problem}. We answer this question in Theorem \ref{thm:small_gamma}.  

\begin{theorem} \label{thm:small_gamma}

Let $\Tilde{\bm{x}}$ and $\gamma_0$ be as defined in Theorem \ref{thm:large_gamma}, and let $\mathcal{X}$ denote the feasible set of \eqref{opt:main_problem} and \eqref{opt:our_formulation}. Specifically, $\Tilde{\bm{x}}$ denotes a minimum norm solution to \eqref{opt:main_problem}, $\gamma_0 = \Vert \Tilde{\bm{x}} \Vert_2^2$ and $\mathcal{X} = \{\bm{x} : \Vert \bm{A}\bm{x}-\bm{b}\Vert_2^2 \leq \epsilon \}$. Let $\lambda_\epsilon > 0$ be a value such that \[\argmin_{\bm{x} \in \mathcal{X}} \Vert \bm{x} \Vert_2^2 = \argmin_{\bm{x}} \Vert \bm{A}\bm{x}-\bm{b}\Vert_2^2 + \lambda_\epsilon \Vert \bm{x} \Vert_2^2.\] Fix any value $\gamma$ with $0 < \gamma < \gamma_0$. Suppose $\Bar{\bm{x}}$ is an optimal solution to \eqref{opt:our_formulation}. Then we have
\[ \Vert \Tilde{\bm{x}} \Vert_0 \leq \Vert \Bar{\bm{x}} \Vert_0 \leq \Vert \Tilde{\bm{x}} \Vert_0 + \frac{1}{\gamma}\bigg{(} \Vert \Tilde{\bm{x}} \Vert_2^2 - \Big{\Vert} \Big{(}\frac{1}{\lambda_\epsilon}\bm{I}+\bm{A}^T \bm{A}\Big{)}^{-1}\bm{A}^T\bm{b} \Big{\Vert}_2^2 \bigg{)}.\] \end{theorem}

\begin{proof}
Fix any value $\gamma$ with $0 < \gamma < \gamma_0$ and consider any optimal solution $\Bar{\bm{x}}$ to \eqref{opt:our_formulation}. The inequality $\Vert \Tilde{\bm{x}} \Vert_0 \leq \Vert \Bar{\bm{x}} \Vert_0$ follows immediately from the optimality of $\Tilde{\bm{x}}$ in \eqref{opt:main_problem}. By the optimality of $\Bar{\bm{x}}$, we must have \[\Vert \Bar{\bm{x}} \Vert_0 + \frac{1}{\gamma} \Vert \Bar{\bm{x}} \Vert_2^2 \leq \Vert \Tilde{\bm{x}} \Vert_0 + \frac{1}{\gamma} \Vert \Tilde{\bm{x}} \Vert_2^2 \implies \Vert \Bar{\bm{x}} \Vert_0 \leq \Vert \Tilde{\bm{x}} \Vert_0 +\frac{1}{\gamma} (\Vert \Tilde{\bm{x}} \Vert_2^2 - \Vert \Bar{\bm{x}} \Vert_2^2).\] Thus, to establish the result we need only derive an upper bound for the term $(\Vert \Tilde{\bm{x}} \Vert_2^2 - \Vert \Bar{\bm{x}} \Vert_2^2)$, or equivalently to derive a lower bound for the term $\Vert \Bar{\bm{x}} \Vert_2^2$. Since $\Bar{\bm{x}} \in \mathcal{X}$, such a lower bound can be obtained by solving the optimization problem given by 
\begin{equation}
\begin{aligned}
    &\min_{\bm{x} \in \mathbb{R}^n} & & \Vert \bm{x} \Vert_2^2\\
    &\text{s.t.} & & \bm{x} \in \mathcal{X} = \{\bm{x}: \Vert\bm{A}\bm{x} - \bm{b}\Vert_2^2 \leq \epsilon\}.
\end{aligned} \label{opt:thm2_lb}
\end{equation}
This optimization problem has the same optimal solution as the ridge regression problem given by 
\begin{equation}
\begin{aligned}
    &\min_{\bm{x} \in \mathbb{R}^n} \Vert \bm{A}\bm{x}-\bm{b}\Vert_2^2 + \lambda_\epsilon \Vert \bm{x} \Vert_2^2.
\end{aligned} \label{opt:ridge}
\end{equation}
for some value $\lambda_\epsilon > 0$. To see this, we form the Lagrangian for \eqref{opt:thm2_lb} $L(\bm{x}, \mu) = \Vert \bm{x} \Vert_2^2 + \mu (\Vert \bm{A}\bm{x}-\bm{b}\Vert_2^2 - \epsilon)$ and observe that the KKT conditions for $(\bm{x}, \mu) \in \mathbb{R}^n \times \mathbb{R}$ are given by
\begin{enumerate}
    \item $\Vert \bm{A}\bm{x}-\bm{b}\Vert_2^2 \leq \epsilon$;
    \item $\mu \geq 0$;
    \item $\mu (\Vert \bm{A}\bm{x}-\bm{b}\Vert_2^2 - \epsilon) = 0 \implies \mu = 0$ or $\Vert \bm{A}\bm{x}-\bm{b}\Vert_2^2 = \epsilon$;
    \item $\nabla_{\bm{x}}L(\bm{x}, \mu) = \bm{0} \implies \bm{x} = (\frac{1}{\mu}\bm{I}+\bm{A}^T\bm{A})^{-1}\bm{A}^T\bm{b}$ if $\mu \neq 0$ and $\bm{x} = 0$ if $\mu = 0$.
\end{enumerate} We note that if $\bm{0} \in \mathcal{X}$, then $\bm{0}$ is trivially an optimal solution to \eqref{opt:thm2_lb} with optimal value given by $0$. This corresponds to the degenerate case. In the nondegenerate case, we have $\bm{0} \notin \mathcal{X}$. This condition, coupled with the first and fourth KKT conditions implies that at optimality, we have $\mu \neq 0$ and $\bm{x} = (\frac{1}{\mu}\bm{I}+\bm{A}^T\bm{A})^{-1}\bm{A}^T\bm{b}$. Next, we note that the unconstrained quadratic optimization problem given by \eqref{opt:ridge} has an optimal solution $\bm{x}^\star$ given by $\bm{x}^\star = (\lambda_\epsilon\bm{I}+\bm{A}^T\bm{A})^{-1}\bm{A}^T\bm{b}$. Finally, we observe that the two preceding expressions are the same when $\lambda_\epsilon = \frac{1}{\mu} > 0$. Thus, we have \[\Vert \Bar{\bm{x}} \Vert_2^2 \geq \min_{\bm{x} \in \mathcal{X}} \Vert \bm{x} \Vert_2^2 = \Big{\Vert} \Big{(}\frac{1}{\lambda_\epsilon}\bm{I}+\bm{A}^T \bm{A}\Big{)}^{-1}\bm{A}^T\bm{b} \Big{\Vert}_2^2,\] which implies that \[ \Vert \Bar{\bm{x}} \Vert_0 \leq \Vert \Tilde{\bm{x}} \Vert_0 + \frac{1}{\gamma}\bigg{(} \Vert \Tilde{\bm{x}} \Vert_2^2 - \Big{\Vert} \Big{(}\frac{1}{\lambda_\epsilon}\bm{I}+\bm{A}^T \bm{A}\Big{)}^{-1}\bm{A}^T\bm{b} \Big{\Vert}_2^2 \bigg{)}.\] This completes the proof. \end{proof}
\begin{remark}
Though the statement of Theorem \ref{thm:small_gamma} is made for any fixed $\gamma$ satisfying $0 < \gamma < \gamma_0$ with $\gamma_0$ given by Theorem \ref{thm:large_gamma}, we note that the proof of Theorem \ref{thm:small_gamma} in fact generalizes to any $\gamma > 0$. This implies that the result of Theorem \ref{thm:large_gamma} holds for any $\gamma_0'$ satisfying $\gamma_0' > \bigg{(} \Vert \Tilde{\bm{x}} \Vert_2^2 - \Big{\Vert} \Big{(}\frac{1}{\lambda_\epsilon}\bm{I}+\bm{A}^T \bm{A}\Big{)}^{-1}\bm{A}^T\bm{b} \Big{\Vert}_2^2 \bigg{)}$. This is a stronger condition than the one established by Theorem \ref{thm:large_gamma} but has the drawback of depending on the value $\lambda_\epsilon$ which in general cannot be computed easily. 
\end{remark}
\noindent
Theorem \ref{thm:small_gamma} provides a worst case guarantee on the sparsity of the solution of \eqref{opt:our_formulation} when the regularization parameter $\gamma$ satisfies $0 < \gamma < \gamma_0$.

\section{An Exact Reformulation and Convex Relaxations} \label{sec:MISOC}

In this section, we reformulate \eqref{opt:our_formulation_W} as a mixed integer second order cone optimization problem. We then employ the perspective relaxation \citep{gunluk2012perspective} to construct a second order cone relaxation for \eqref{opt:our_formulation_W} and demonstrate that under certain conditions on the regularization parameter $\gamma$, the resulting relaxation is equivalent to the Weighted Basis Pursuit Denoising problem given by \eqref{opt:BPD_W}. As a special case, we obtain a convex relaxation for \eqref{opt:our_formulation} and demonstrate that it is equivalent to \eqref{opt:BPD} under the same conditions on $\gamma$. Finally, we present a family of semidefinite relaxations to \eqref{opt:our_formulation_W} using techniques from polynomial optimization.

To model the sparsity of the vector $\bm{x}$ in \eqref{opt:our_formulation_W}, we introduce binary variables $\bm{z} \in \{0, 1\}^{n}$ and require that $x_i = z_ix_i$. This gives the following reformulation of \eqref{opt:our_formulation_W}:

\begin{equation}
\begin{aligned}
    &\min_{\bm{z}, \bm{x} \in \mathbb{R}^n} & & \sum_{i=1}^n z_i + \frac{1}{\gamma} \sum_{i=1}^nw_i^2x_i^2\\
    &\text{s.t.} & & \Vert\bm{A}\bm{x} - \bm{b}\Vert_2^2 \leq \epsilon, \ x_i = z_ix_i \ \forall \ i,\ z_i \in \{0, 1\} \ \forall \ i.
\end{aligned} \label{opt:nonlinear}
\end{equation} The constraints $x_i = z_ix_i$ in \eqref{opt:nonlinear} are  nonconvex in the decision variables $(\bm{x}, \bm{z})$. To deal with these constraints, we make use of the perspective reformulation \citep{gunluk2012perspective}. Specifically, we introduce non-negative variables $\bm{\theta} \in \mathbb{R}_+^{n}$ where $\theta_i$ models $x_i^2$ and introduce the constraints $\theta_iz_i \geq x_i^2$, which are second order cone representable. Thus, if $z_i=0$, we will have $x_i = 0$. This results in the following reformulation of \eqref{opt:nonlinear}:

\begin{equation}
\begin{aligned}
    &\min_{\bm{z}, \bm{x}, \bm{\theta} \in \mathbb{R}^n} & & \sum_{i=1}^n z_i + \frac{1}{\gamma} \sum_{i=1}^nw_i^2\theta_i\\
    &\text{s.t.} & & \Vert\bm{A}\bm{x} - \bm{b}\Vert_2^2 \leq \epsilon, \ x_i^2 \leq z_i\theta_i \ \forall \ i,\\
    & & & z_i \in \{0, 1\} \ \forall \ i, \ \theta_i \geq 0 \ \forall \ i.
\end{aligned} \label{opt:MISOC_W}
\end{equation}

\begin{theorem}
The mixed integer second order cone problem given by \eqref{opt:MISOC_W} is an exact reformulation of \eqref{opt:our_formulation_W}.
\end{theorem}

\begin{proof}
We show that given a feasible solution to \eqref{opt:our_formulation_W}, we can construct a feasible solution to \eqref{opt:MISOC_W} that achieves the same objective value and vice versa.

Consider an arbitrary solution $\Bar{\bm{x}}$ to \eqref{opt:our_formulation_W}. Let $\Bar{\bm{z}} \in \mathbb{R}^n$ be the binary vector obtained by setting $\Bar{z}_i = \mathbbm{1}\{\Bar{x}_i \neq 0\}$ and let $\Bar{\bm{\theta}} \in \mathbb{R}^n$ be the vector obtained by setting $\Bar{\theta}_i = \Bar{x}_i^2$. We have $\Vert\bm{A}\Bar{\bm{x}} - \bm{b}\Vert_2^2 \leq \epsilon$, $\Bar{z}_i\Bar{\theta}_i = \mathbbm{1}\{\Bar{x}_i \neq 0\} \cdot \Bar{x}_i^2 = \Bar{x}_i^2$, $\Bar{\bm{z}} \in \{0, 1\}^n$ and $\Bar{\theta}_i \geq 0$ so the solution $(\Bar{\bm{x}}, \Bar{\bm{z}}, \Bar{\bm{\theta}})$ is feasible to \eqref{opt:MISOC_W}. Lastly, notice that we have \[\sum_{i=1}^n \Bar{z}_i + \frac{1}{\gamma} \sum_{i=1}^nw_i^2\Bar{\theta}_i = \sum_{i=1}^n \mathbbm{1}\{\Bar{x}_i \neq 0\} + \frac{1}{\gamma} \sum_{i=1}^nw_i^2\Bar{x}_i^2 = \Vert \Bar{\bm{x}} \Vert_0 + \frac{1}{\gamma} \Vert \bm{W} \Bar{\bm{x}} \Vert_2^2{\color{black},}\] {\color{black}where $\bm{W} = \text{diag}(w_1, \ldots, w_n)$.} Thus, the solution $(\Bar{\bm{x}}, \Bar{\bm{z}}, \Bar{\bm{\theta}})$ is a feasible solution to \eqref{opt:MISOC_W} that achieves the same objective value as $\Bar{\bm{x}}$ does in \eqref{opt:our_formulation_W}.

Consider now an arbitrary solution $(\Bar{\bm{x}}, \Bar{\bm{z}}, \Bar{\bm{\theta}})$ to \eqref{opt:MISOC_W}. Since we have $\Vert\bm{A}\Bar{\bm{x}} - \bm{b}\Vert_2^2 \leq \epsilon$, $\Bar{\bm{x}}$ is feasible to \eqref{opt:our_formulation_W}. Next, we note that the constraints $x_i^2 \leq z_i\theta_i$ and $z_i \in \{0, 1\}$ imply that $\Bar{z}_i \geq \mathbbm{1}\{\Bar{x}_i \neq 0\}$ and $\Bar{\theta}_i \geq \Bar{x}_i^2$. Finally, we observe that \[\Vert \Bar{\bm{x}} \Vert_0 + \frac{1}{\gamma} \Vert \bm{W} \Bar{\bm{x}} \Vert_2^2 = \sum_{i=1}^n \mathbbm{1}\{\Bar{x}_i \neq 0\} + \frac{1}{\gamma} \sum_{i=1}^nw_i^2\Bar{x}_i^2 \leq \sum_{i=1}^n \Bar{z}_i + \frac{1}{\gamma} \sum_{i=1}^nw_i^2\Bar{\theta}_i.\] Thus, the solution $\Bar{\bm{x}}$ is a feasible solution to \eqref{opt:our_formulation_W} that achieves an objective value equal to or less than the objective value that $(\Bar{\bm{x}}, \Bar{\bm{z}}, \Bar{\bm{\theta}})$ achieves in \eqref{opt:MISOC_W}. This completes the proof. \end{proof}

\subsection{A Second Order Cone Relaxation} \label{ssec:SOC}

Problem \eqref{opt:MISOC_W} is a reformulation of Problem \eqref{opt:our_formulation_W} where the problem's nonconvexity is entirely captured by the binary variables $\bm{z}$. We now obtain a convex relaxation of \eqref{opt:our_formulation_W} by solving \eqref{opt:MISOC_W} with $\bm{z} \in \text{conv}(\{0, 1\}^n) = [0, 1]^n$.  This gives the following convex optimization problem:

\begin{equation}
\begin{aligned}
    &\min_{\bm{z}, \bm{x}, \bm{\theta} \in \mathbb{R}^n} & & \sum_{i=1}^n z_i + \frac{1}{\gamma} \sum_{i=1}^nw_i^2\theta_i\\
    &\text{s.t.} & & \Vert\bm{A}\bm{x} - \bm{b}\Vert_2^2 \leq \epsilon, \ x_i^2 \leq z_i\theta_i \ \forall \ i,\\
    & & & 0 \leq z_i \leq 1 \ \forall \ i,\ \theta_i \geq 0 \ \forall \ i.
\end{aligned} \label{opt:MISOC_W_relax}
\end{equation} A natural question to ask is how problem \eqref{opt:MISOC_W_relax} compares to the Weighted Basis Pursuit Denoising problem given by \eqref{opt:BPD_W}, a common convex approximation for CS in the noisy setting. Surprisingly, under mild conditions on the regularization parameter $\gamma$, it can be shown that solving \eqref{opt:MISOC_W_relax} is exactly equivalent to solving \eqref{opt:BPD_W}. This implies that though Basis Pursuit Denoising is typically motivated as a convex approximation to CS in the presence of noise, it can alternatively be understood as the natural convex relaxation of the mixed integer second order cone problem given by \eqref{opt:MISOC_W} for appropriately chosen values of $\gamma$. We formalize this statement in Theorem \ref{thm:BPD_equiv}.

\begin{theorem} \label{thm:BPD_equiv}

There exists a finite value $\gamma_0 < \infty$ such that for all $\Bar{\gamma} \geq \gamma_0$, any vector $\bm{x}^\star$ that is an optimal solution of \eqref{opt:BPD_W} is also an optimal solution of \eqref{opt:MISOC_W_relax}. Let $\mathcal{X} = \{\bm{x}: \Vert\bm{A}\bm{x} - \bm{b}\Vert_2^2 \leq \epsilon\}$, the feasible set of $\eqref{opt:BPD_W}$. We can take $\gamma_0 = \max_{x \in \mathcal{X}} \Vert \bm{W} \bm{x} \Vert_\infty^2${\color{black}, where $\bm{W} = \text{diag}(w_1, \ldots, w_n)$}.

\end{theorem}

\begin{proof} Rewrite \eqref{opt:MISOC_W_relax} as the two stage optimization problem given by \eqref{opt:MISOC_W_relax_2}.

\begin{equation}
\begin{aligned}
    \min_{\bm{x} \in \mathcal{X}} \quad & \min_{\bm{z}, \bm{\theta} \in \mathbb{R}^n} & & \sum_{i=1}^n z_i + \frac{1}{\gamma} \sum_{i=1}^nw_i^2\theta_i\\
    &\text{s.t.} & & x_i^2 \leq z_i\theta_i \ \forall \ i, \ 0 \leq z_i \leq 1 \ \forall \ i,\ \theta_i \geq 0 \ \forall \ i.
\end{aligned} \label{opt:MISOC_W_relax_2}
\end{equation} Let $\gamma_0 = \max_{x \in \mathcal{X}} \Vert x \Vert_\infty^2$. To establish the result, we will show that for any $\bm{x} \in \mathcal{X}$ the optimal value of the inner minimization problem in \eqref{opt:MISOC_W_relax_2} is a scalar multiple of the $\ell_1$ norm of $\bm{W}\bm{x}$ provided that $\gamma \geq \gamma_0$.

Fix $\gamma \geq \gamma_0$ and consider any $\Bar{\bm{x}} \in \mathcal{X}$. We make three observations that allow us to reformulate the inner minimization problem in \eqref{opt:MISOC_W_relax_2}:

\begin{enumerate}
    \item The objective function of the inner minimization problem is separable.
    \item For any $i$ such that $\Bar{x}_i = 0$, it is optimal to set $z_i = \theta_i = 0$ which results in no contribution to the objective function.
    \item For any $i$ such that $\Bar{x}_i \neq 0$, we must have $z_i > 0$ and it is optimal to take $\theta_i = \frac{\Bar{x}_i^2}{z_i}$.
\end{enumerate} We can therefore equivalently express the inner minimization problem of \eqref{opt:MISOC_W_relax_2} as:

\begin{equation}
\begin{aligned}
    & \min_{\bm{z} \in \mathbb{R}^n} & & \sum_{i: x_i \neq 0} \bigg{[}z_i + \frac{w_i^2}{\gamma} \cdot \frac{\Bar{x}_i^2}{z_i}\bigg{]}\\
    &\text{s.t.} & & 0 < z_i \leq 1 \ \forall \ i.
\end{aligned} \label{opt:MISOC_W_relax_simplified}
\end{equation} Let $f_i(z) = z + \frac{w_i^2}{\gamma} \cdot \frac{\Bar{x}_i^2}{z}$. We want to minimize the function $f_i(z)$ over the interval $(0, 1]$ for all $i$ such that $\Bar{x}_i \neq 0$. Fix an arbitrary $i$ satisfying $\Bar{x}_i \neq 0$. We have $\frac{d}{dz}f_i(z) = 1 - \frac{w_i^2}{\gamma} \cdot \frac{\Bar{x}_i^2}{z_i^2}$ and $\frac{d}{dz}f_i(z^\star) = 0 \iff z^\star = \pm \frac{w_i}{\sqrt{\gamma}} \vert \Bar{x}_i \vert$. The condition $\gamma \geq \gamma_0 = \max_{x \in \mathcal{X}} \Vert \bm{W} \bm{x} \Vert_\infty^2$ and the fact that $\Bar{\bm{x}} \in \mathcal{X}$ implies that $1 \geq \frac{w_i^2\Bar{x}_i^2}{\gamma}$ for all $i$. Thus, we have $0 < \frac{w_i}{\sqrt{\gamma}}\vert \Bar{x}_i \vert \leq 1$. Let $\Bar{z} = \frac{w_i}{\sqrt{\gamma}}\vert \Bar{x}_i \vert$. Noting that $\lim_{z\xrightarrow{}0^+}f_i(z) = \infty$, the minimum of $f_i(z)$ over the interval $(0, 1]$ must occur either at $1$ or $\Bar{z}$. We have {\color{black}\[\bigg{(}\frac{w_i}{\sqrt{\gamma}}\vert \Bar{x}_i \vert-1\bigg{)}^2 \geq 0 \implies f_i(1) = \frac{w_i^2\Bar{x}_i^2}{\gamma}+1 \geq \frac{2w_i}{\sqrt{\gamma}}\vert \Bar{x}_i \vert = f_i(\Bar{z}).\]} Therefore, the minimum of $f_i(z)$ on $(0, 1]$ occurs at $\Bar{z} = \frac{w_i}{\sqrt{\gamma}}\vert \Bar{x}_i \vert$ and is equal to $f_i(\Bar{z})=\frac{2}{\sqrt{\gamma}}\vert \Bar{x}_i \vert$. This allows us to conclude that the optimal value of \eqref{opt:MISOC_W_relax_2} is given by: \[\sum_{i: x_i \neq 0} \frac{2w_i}{\sqrt{\gamma}}\vert \Bar{x}_i \vert = \sum_{i=1}^n \frac{2w_i}{\sqrt{\gamma}}\vert \Bar{x}_i \vert = \frac{2}{\sqrt{\gamma}} \Vert \bm{W}\Bar{\bm{x}} \Vert_1.\] We have shown that for fixed $\bm{x} \in \mathcal{X}$, the optimal value of the inner minimization problem of \eqref{opt:MISOC_W_relax_2} is a scalar multiple of the $\ell_1$ norm of $\bm{W}\bm{x}$. We can rewrite \eqref{opt:MISOC_W_relax_2} as 
\begin{equation}
\begin{aligned}
    \min_{\bm{x} \in \mathcal{X}} \frac{2}{\sqrt{\gamma}} \Vert \bm{W}\Bar{\bm{x}} \Vert_1,
\end{aligned}
\end{equation} which has the same set of optimal solutions as \eqref{opt:BPD_W} because this set is invariant under scaling of the objective function. This completes the proof. \end{proof}

\begin{remark}

Note that by taking $\bm{W} = \bm{I}$, it immediately follows from Theorem \ref{thm:BPD_equiv} that any vector $\bm{x}^\star$ that is an optimal solution of \eqref{opt:BPD} is also an optimal solution of \eqref{opt:MISOC_W_relax} when we set $\gamma \geq \gamma_0 = \max_{\bm{x} \in \mathcal{X}} \Vert \bm{x} \Vert_\infty^2$.
\end{remark}
Convex relaxations of nonconvex optimization problems are helpful for two reasons. Firstly, a convex relaxation provides a lower (upper) bound to a minimization (maximization) problem which given a feasible solution to the nonconvex optimization problem provides a certificate of worst case suboptimality. Secondly, convex relaxations can often be used as building blocks in the construction of global optimization algorithms or heuristics for nonconvex optimization problems. Strong convex relaxations are desirable because they produce tighter bounds on the optimal value of the problem of interest (stronger certificates of worst case suboptimality) and generally lead to more performant global optimization algorithms and heuristics. Let $\mathcal{X}_1 = \{(\bm{z}, \bm{x}, \bm{\theta}) \in \mathbb{R}^n \times \mathbb{R}^n \times \mathbb{R}^n: \Vert \bm{A}\bm{x}-\bm{b} \Vert_2^2 \leq \epsilon, x_i^2 \leq z_i\theta_i \,\, \forall \ i, \theta_i \geq 0 \,\, \forall \ i\}$ and $\mathcal{X}_1 = \{(\bm{z}, \bm{x}, \bm{\theta}) \in \mathbb{R}^n \times \mathbb{R}^n \times \mathbb{R}^n: \bm{z} \in \{0, 1\}^n\}$. We can equivalently write \eqref{opt:MISOC_W} as:
\[\min_{(\bm{z}, \bm{x}, \bm{\theta}) \in \mathcal{X}_1 \cap \mathcal{X}_2} \sum_{i=1}^n z_i + \frac{1}{\gamma} \sum_{i=1}^nw_i^2\theta_i.\] The strongest possible convex relaxation to \eqref{opt:MISOC_W} would be obtained by minimizing the objective function in \eqref{opt:MISOC_W} subject to the constraint that $(\bm{z}, \bm{x}, \bm{\theta}) \in \text{conv}(\mathcal{X}_1 \cap \mathcal{X}_2)$. Since the objective function is linear in the decision variables, solving over $\text{conv}(\mathcal{X}_1 \cap \mathcal{X}_2)$ would produce an optimal solution to \eqref{opt:MISOC_W} since the objective would be minimized at an extreme point of $\text{conv}(\mathcal{X}_1 \cap \mathcal{X}_2)$ which by definition must be an element of $\mathcal{X}_1 \cap \mathcal{X}_2$. Unfortunately, in general it is hard to represent $\text{conv}(\mathcal{X}_1 \cap \mathcal{X}_2)$ explicitly. The relaxation given by \eqref{opt:MISOC_W_relax} consists of minimizing the objective function of \eqref{opt:MISOC_W} subject to the constraint that $(\bm{z}, \bm{x}, \bm{\theta}) \in (\text{conv}(\mathcal{X}_1) \cap \text{conv}(\mathcal{X}_2)) = \mathcal{X}_1 \cap \text{conv}(\mathcal{X}_2) \supseteq \text{conv}(\mathcal{X}_1 \cap \mathcal{X}_2)$.

Stronger convex relaxations to \eqref{opt:MISOC_W} can be obtained by introducing additional valid inequalities to \eqref{opt:MISOC_W} and then relaxing the integrality constraint on $\bm{z}$. For example, suppose we know a value $M \geq \gamma_0 = \max_{\bm{x} \in \mathcal{X}}\Vert \bm{W}\bm{x}\Vert_\infty^2$. We can use this value to introduce Big-M constraints {\color{black}similar in} flavour to the formulation proposed by \cite{karahanouglu2013mixed}. Under this assumption, it follows immediately that any feasible solution to \eqref{opt:MISOC_W} satisfies $-Mz_i \leq w_ix_i \leq Mz_i \,\, \forall \ i$. Thus, we can obtain another convex relaxtion of \eqref{opt:MISOC_W} by minimizing its objective function subject to the constraint $(\bm{z}, \bm{x}, \bm{\theta}) \in \Bar{\mathcal{X}}_1 \cap \text{conv}(\mathcal{X}_2) \supseteq \text{conv}(\mathcal{X}_1 \cap \mathcal{X}_2)$ where we define $\Bar{\mathcal{X}}_1 = \mathcal{X}_1 \cap \{(\bm{z}, \bm{x}, \bm{\theta}) \in \mathbb{R}^n \times \mathbb{R}^n \times \mathbb{R}^n: -Mz_i \leq w_ix_i \leq Mz_i \,\, \forall \ i\}$. Explicitly, with knowledge of such a value $M$ we can solve

\begin{equation}
\begin{aligned}
    &\min_{\bm{z}, \bm{x}, \bm{\theta} \in \mathbb{R}^n} & & \sum_{i=1}^n z_i + \frac{1}{\gamma} \sum_{i=1}^nw_i^2\theta_i\\
    &\text{s.t.} & & \Vert\bm{A}\bm{x} - \bm{b}\Vert_2^2 \leq \epsilon, \ x_i^2 \leq z_i\theta_i \ \forall \ i,\\
    & & & -Mz_i \leq w_ix_i \leq Mz_i \ \forall \ i,\ 0 \leq z_i \leq 1 \ \forall \ i,\ \theta_i \geq 0 \ \forall \ i.
\end{aligned} \label{opt:MISOC_W_relax_M}
\end{equation}

\begin{remark}
Given any input data $\bm{A}, \bm{b}, \epsilon$, if $M$ satisfies $M \geq \gamma_0 = \max_{\bm{x}\in \mathcal{X}}\Vert\bm{W}\bm{x}\Vert_\infty^2$, then the optimal value of \eqref{opt:MISOC_W_relax_M} is no less than the optimal value of \eqref{opt:MISOC_W_relax}. This follows immediately by noting that under the condition on $M$, the feasible set of \eqref{opt:MISOC_W_relax_M} is contained in the feasible set of \eqref{opt:MISOC_W_relax}.
\end{remark}

The mixed integer second order cone reformulation and convex relaxation introduced in this section lead to two approaches for solving \eqref{opt:our_formulation_W} to certifiable optimality. On the one hand, solvers like \verb|Gurobi| contain direct support for solving mixed integer second order cone problems so problem \eqref{opt:our_formulation_W} can be solved directly. On the other hand, it is possible to develop a custom branch-and-bound routine that leverages a modification of \eqref{opt:BPD_W} to compute lower bounds. We illustrate this in Section \ref{sec:BnB}. This custom, problem specific approach outperforms \verb|Gurobi| because \eqref{opt:BPD} is a more tractable problem than \eqref{opt:MISOC_W_relax} due in part to the presence of fewer second order cone constraints which decreases the computational time spent computing lower bounds.

\subsection{A Positive Semidefinite Cone Relaxation} \label{ssec:SDP}

In this section, we formulate \eqref{opt:our_formulation_W} as a polynomial optimization problem and present a semidefinite relaxation using the sum of squares (SOS) hierarchy \citep{lasserre2001explicit}. We show that this semidefinite relaxation is tighter than the second order cone relaxation presented previously.

Let $f(\bm{z}, \bm{x})=\sum_{i=1}^n z_i + \frac{1}{\gamma} \sum_{i=1}^nw_i^2\theta_i$ denote the objective function of \eqref{opt:MISOC_W}. Notice that the constraint $\bm{z} \in \{0, 1\}^n$ in \eqref{opt:MISOC_W} is equivalent to the constraint $\bm{z} \circ \bm{z} = \bm{z}$ (where $\circ$ denotes the element wise product). With this observation, we can express the feasible set of \eqref{opt:MISOC_W} as the semialgebraic set given by: \[\Omega = \{(\bm{z}, \bm{x}) \in \mathbb{R}^n \times \mathbb{R}^n: \epsilon - \Vert\bm{A}\bm{x} - \bm{b}\Vert_2^2 \geq 0, x_iz_i-x_i=0 \,\, \forall \ i, z_i^2-z_i=0 \,\, \forall \ i\}.\] Thus, we can equivalently write \eqref{opt:MISOC_W} as $\min_{(\bm{z}, \bm{x}) \in \Omega} f(\bm{z}, \bm{x})$. It is not difficult to see that the preceding optimization problem has the same optimal value as the problem given by
\begin{equation}
\begin{aligned}
    \max_{\lambda \in \mathbb{R}} & & \lambda\ \text{s.t.} \ f(\bm{z}, \bm{x}) - \lambda \geq 0 \ \forall \ (\bm{z}, \bm{x}) \in \Omega.
\end{aligned} \label{opt:poly_reform}
\end{equation} Problem \eqref{opt:poly_reform} is a polynomial optimization problem that has the same optimal value as \eqref{opt:our_formulation_W}.

We can obtain tractable lower bounds for \eqref{opt:poly_reform} by leveraging techniques from sum of squares (SOS) optimization \citep{lasserre2001global, lasserre2009moments}. A polynomial $g \in \mathbb{R}[x]$ is said to be sum of squares (SOS) if for some $K \in \mathbb{N}$ there exists polynomials $\{g_k\}_{k=1}^K \subset \mathbb{R}[x]$ such that $g = \sum_{k=1}^K g_k^2$. We denote the set of all SOS polynomials as $\Sigma^2[x]$. Moreover, we denote the set of polynomials of degree at most $d$ as $\mathbb{R}_d[x] \subset \mathbb{R}[x]$ and we denote the set of SOS polynomials of degree at most $2d$ as $\Sigma^2_d[x] \subset \Sigma^2[x]$. It is trivial to see that any polynomial that is SOS is globally non-negative. More generally, SOS polynomials can be utilized to model polynomial non-negativity over arbitrary semialgebraic sets. The quadratic module associated with the semialgebraic set $\Omega$ is defined as:
\begin{equation}
\begin{aligned}
    & QM(\Omega) = & &\bigg{\{}s_0(\bm{z}, \bm{x}) + s_1(\bm{z}, \bm{x})(\epsilon - \Vert\bm{A}\bm{x} - \bm{b}\Vert_2^2) + \sum_{i=1}^n t_i(\bm{z}, \bm{x})(x_iz_i-x_i) \\
    & & & + \sum_{i=1}^n r_i(\bm{z}, \bm{x})(z_i^2-z_i): s_0, s_1 \in \Sigma^2[\bm{z}, \bm{x}], t_i, r_i \in \mathbb{R}[\bm{z}, \bm{x}] \,\, \forall \ i \bigg{\}}.
\end{aligned} \label{def:quad_module}
\end{equation} It is straightforward to see that if a function $h(\bm{z}, \bm{x})$ is an element of $QM(\Omega)$, then $h(\bm{z}, \bm{x})$ is non-negative on $\Omega$ (since for points in $\Omega$, $h(\bm{z}, \bm{x})$ takes the form of the sum of two SOS polynomials). Thus, membership in $QM(\Omega)$ is a sufficient condition for non-negativity on $\Omega$. We further define the restriction of $QM(\Omega)$ to polynomials of degree at most $2d$ as:
\begin{equation}
\begin{aligned}
    QM_d&(\Omega) = \bigg{\{}s_0(\bm{z}, \bm{x}) + s_1(\bm{z}, \bm{x})(\epsilon - \Vert\bm{A}\bm{x} - \bm{b}\Vert_2^2) + \sum_{i=1}^n t_i(\bm{z}, \bm{x})(x_iz_i-x_i) \\
    &+ \sum_{i=1}^n r_i(\bm{z}, \bm{x})(z_i^2-z_i): s_0 \in \Sigma^2_d[\bm{z}, \bm{x}], s_1 \in \Sigma^2_{d-1}[\bm{z}, \bm{x}], t_i, r_i \in \mathbb{R}_{2d-2}[\bm{z}, \bm{x}] \,\, \forall \ i \bigg{\}}.
\end{aligned} \label{def:quad_module_d}
\end{equation} It is immediate that $QM_d(\Omega) \subset QM(\Omega)$ and membership in $QM_d(\Omega)$ provides a certificate of non-negativity on $\Omega$. Importantly, given an arbitrary polynomial $h(\bm{z}, \bm{x})$ it is possible to verify membership in $QM_d(\Omega)$ by checking feasibility of a semidefinite program. Thus, for any $d \in \mathbb{N}$, we obtain a semidefinite relaxation of \eqref{opt:our_formulation_W} by solving:
\begin{equation}
\begin{aligned}
    \max_{\lambda \in \mathbb{R}} & & \lambda\ \text{s.t.} \ f(\bm{z}, \bm{x}) - \lambda \in QM_d(\Omega).
\end{aligned} \label{opt:poly_relax}
\end{equation} Since $QM_d(\Omega) \subset QM_{d+1}(\Omega)$, \eqref{opt:poly_relax} produces an increasingly strong lower bound with increasing values of $d$. A natural question to ask is how the relaxation given by \eqref{opt:poly_relax} compares to that given by \eqref{opt:MISOC_W_relax}. We answer this question in Theorem \ref{thm:SOS_better}.
\\
\begin{theorem} \label{thm:SOS_better}

For every $d \geq 1$, the optimal value of \eqref{opt:poly_relax} is no less than the optimal value of \eqref{opt:MISOC_W_relax}.

\end{theorem}

\begin{proof} Without loss of generality, we take $\bm{W} = \bm{I}$. We prove the result for $\gamma \geq \gamma_0 = \max_{\bm{x} \in \mathcal{X}} \Vert \bm{x} \Vert_\infty^2$ though the result extends naturally to the case of arbitrary $\gamma$. Fix any $\epsilon > 0$, $\bm{A} \in \mathbb{R}^{m \times n}$ and $\bm{b} \in \mathbb{R}^m$. By Theorem \ref{thm:BPD_equiv}, \eqref{opt:MISOC_W_relax} has the same optimal value as \eqref{opt:BPD_W}. Consider the dual of \eqref{opt:BPD_W} which for $\bm{W} = \bm{I}$ is given by
\begin{equation}
\begin{aligned}
    \max_{\bm{\nu} \in \mathbb{R}^m} & & \bm{b}^T\bm{\nu} - \sqrt{\epsilon} \Vert \bm{\nu} \Vert_2 \ \text{s.t.} \ \vert \bm{\nu}^TA_i \vert \leq \frac{2}{\sqrt{\gamma}} \quad \forall \ i{\color{black},}
\end{aligned} \label{opt:BPD_dual}
\end{equation} {\color{black}where $A_i$ denotes the $i^{th}$ column of $\bm{A}$.} Strong duality holds between \eqref{opt:BPD_dual} and \eqref{opt:BPD_W} since $\bm{\nu} = 0$ is always a strictly feasible point in \eqref{opt:BPD_dual} \citep{boyd2004convex}. Fix $d=1$. We will show that for any feasible solution to \eqref{opt:BPD_dual}, we can construct a feasible solution to \eqref{opt:poly_relax} that achieves the same objective value. Let $\bm{\Bar{\nu}} \in \mathbb{R}^m$ denote an arbitrary feasible solution to \eqref{opt:BPD_dual}. {\color{black}Define $\Bar{r}_i(\bm{z}, \bm{x}) = -1, \Bar{t}_i(\bm{z}, \bm{x}) = A_i^T\Bar{\bm{\nu}}$ for all $i$, $\Bar{s}_1(\bm{z}, \bm{x}) = \frac{\Vert \bm{\Bar{\nu}} \Vert_2}{2\sqrt{\epsilon}}$ and define $\Bar{s}_0(\bm{z}, \bm{x}) = \text{monomial}(\bm{z}, \bm{x}, 1)^T \Bar{\bm{S}}\text{monomial}(\bm{z}, \bm{x}, 1)$ where $\text{monomial}(\bm{x}, \bm{z}, 1) \in \mathbb{R}[\bm{z}, \bm{x}]^{2n+1}$ is the vector of monomials in $\mathbb{R}[\bm{z}, \bm{x}]$ of degree at most $1$ and $\Bar{\bm{S}} \in \mathbb{R}^{2n+1 \times 2n+1}$ is given by
\[
\Bar{\bm{S}} = \left[ 
\begin{array}{c | c | c} 
  \frac{1}{\gamma}\bm{I}_n+\frac{\Vert \bm{\Bar{\nu}} \Vert_2}{2\sqrt{\epsilon}}\bm{A}^T\bm{A} & \text{diag}\Big{(}\frac{-\bm{A}^T\Bar{\bm{\nu}}}{2}\Big{)} & \frac{1}{2}\bm{A}^T\Big{(}\Bar{\bm{\nu}}-\frac{\Vert \bm{\Bar{\nu}} \Vert_2}{\sqrt{\epsilon}}\bm{b}\Big{)} \\ 
  \hline 
  \text{diag}\Big{(}\frac{-\bm{A}^T\Bar{\bm{\nu}}}{2}\Big{)} & \bm{I}_n & \bm{0}_n \\
  \hline
  \frac{1}{2}\Big{(}\Bar{\bm{\nu}}^T-\frac{\Vert \bm{\Bar{\nu}} \Vert_2}{\sqrt{\epsilon}}\bm{b}^T\Big{)}\bm{A} & \bm{0}_n^T & \Big{(}\frac{\bm{b}^T\bm{b}}{2\sqrt{\epsilon}}+\frac{\sqrt{\epsilon}}{2}\Big{)}\Vert \bm{\Bar{\nu}}\Vert_2 - \Bar{\bm{\nu}}^T\bm{b}
 \end{array} 
\right].
\]
Clearly, we have $\Bar{t}_i, \Bar{r}_i \in \mathbb{R}_0[\bm{z}, \bm{x}]$ for all $i$ and $\Bar{s}_1 \in \Sigma^2_{0}[\bm{z}, \bm{x}]$ because $\frac{\Vert \bm{\Bar{\nu}} \Vert_2}{2\sqrt{\epsilon}} \geq 0$. We claim that $\Bar{s}_0 \in \Sigma^2_1[\bm{z}, \bm{x}]$. To see this, note that by the generalized Schur complement lemma (see Boyd et al. 1994, Equation 2.41), $\Bar{\bm{S}} \succeq 0$ if and only if $\begin{pmatrix}\bm{I}_n & \bm{0}_n\\ \bm{0}_n^T & \sigma\end{pmatrix} \succeq 0$ and $\frac{1}{\gamma}\bm{I}_n+\frac{\Vert \bm{\Bar{\nu}} \Vert_2}{2\sqrt{\epsilon}}\bm{A}^T\bm{A}-\text{diag}\Big{(}\frac{-\bm{A}^T\Bar{\bm{\nu}}}{2}\Big{)}^2-\frac{1}{4\sigma}\bm{A}^T\Big{(}\Bar{\bm{\nu}}-\frac{\Vert \bm{\Bar{\nu}} \Vert_2}{\sqrt{\epsilon}}\bm{b}\Big{)}\Big{(}\Bar{\bm{\nu}}-\frac{\Vert \bm{\Bar{\nu}} \Vert_2}{\sqrt{\epsilon}}\bm{b}\Big{)}^T\bm{A} \succeq 0$ where we let $\sigma = \Big{(}\frac{\bm{b}^T\bm{b}}{2\sqrt{\epsilon}}+\frac{\sqrt{\epsilon}}{2}\Big{)}\Vert \bm{\Bar{\nu}}\Vert_2 - \Bar{\bm{\nu}}^T\bm{b}$. The first condition is satisfied if $\sigma \geq 0$. To see that this is always the case, notice that we can equivalently express $\sigma$ as \[\sigma = \bigg{(}\frac{\bm{b}^T\bm{b} + \epsilon}{2\sqrt{\epsilon} \Vert \bm{b} \Vert_2}\bigg{)}\Vert \bm{b} \Vert_2 \Vert \bm{\Bar{\nu}}\Vert_2 - \Bar{\bm{\nu}}^T\bm{b}.\] By Cauchy-Schwarz, we have $\vert \Bar{\bm{\nu}}^T\bm{b} \vert \leq \Vert \bm{b} \Vert_2 \Vert \bm{\Bar{\nu}}\Vert_2$. Moreover, we have $0 \leq (\Vert \bm{b} \Vert_2 - \sqrt{\epsilon})^2 \implies \frac{\bm{b}^T\bm{b} + \epsilon}{2\sqrt{\epsilon} \Vert \bm{b} \Vert_2} \geq 1$. It follows immediately that $\sigma \geq 0$. 

To establish the second condition, we first rewrite the Schur complement of $\bm{\Bar{S}}$ as the sum of two matrices:
\[\Bigg{[}\frac{1}{\gamma}\bm{I}_n-\text{diag}\Big{(}\frac{-\bm{A}^T\Bar{\bm{\nu}}}{2}\Big{)}^2 \Bigg{]} + \Bigg{[} \frac{\Vert \bm{\Bar{\nu}} \Vert_2}{2\sqrt{\epsilon}}\bm{A}^T\bm{A} -\frac{1}{4\sigma}\bm{A}^T\Big{(}\Bar{\bm{\nu}}-\frac{\Vert \bm{\Bar{\nu}} \Vert_2}{\sqrt{\epsilon}}\bm{b}\Big{)}\Big{(}\Bar{\bm{\nu}}-\frac{\Vert \bm{\Bar{\nu}} \Vert_2}{\sqrt{\epsilon}}\bm{b}\Big{)}^T\bm{A} \Bigg{]},\] where we let $\bm{\Phi}$ denote the first matrix in the sum and we let $\bm{\Psi}$ denote the second matrix. It suffices to show that $\bm{\Phi}$ and $\bm{\Psi}$ are both positive semidefinite. The eigenvalues of $\bm{\Phi}$ are given by $\{\frac{1}{\gamma}- \frac{(\Bar{\bm{\nu}}^TA_i)^2}{4}\}_{i=1}^n$. Thus, $\bm{\Phi}$ is positive semidefinite as long as $\vert \Bar{\bm{\nu}}^TA_i \vert \leq \frac{2}{\sqrt{\gamma}}$ for all $i$ which is guaranteed by the feasibility of $\Bar{\bm{\nu}}$ in \eqref{opt:BPD_dual}. To see that $\bm{\Psi} \succeq 0$, note that we can write $\bm{\Psi}$ as \[\bm{\Psi} =  \bm{A}^T\Bigg{(} \frac{\Vert \bm{\Bar{\nu}} \Vert_2}{2\sqrt{\epsilon}}\bm{I}_m -\frac{1}{4\sigma}\Big{(}\Bar{\bm{\nu}}-\frac{\Vert \bm{\Bar{\nu}} \Vert_2}{\sqrt{\epsilon}}\bm{b}\Big{)}\Big{(}\Bar{\bm{\nu}}-\frac{\Vert \bm{\Bar{\nu}} \Vert_2}{\sqrt{\epsilon}}\bm{b}\Big{)}^T\Bigg{)}\bm{A}.\] Since any matrix of the form $\bm{X}^T\bm{Y}\bm{X}$ is positive semidefinite provided that $\bm{Y}$ is positive semidefinite, it suffices to show that $\frac{\Vert \bm{\Bar{\nu}} \Vert_2}{2\sqrt{\epsilon}}\bm{I}_m \succeq \frac{1}{4\sigma}\Big{(}\Bar{\bm{\nu}}-\frac{\Vert \bm{\Bar{\nu}} \Vert_2}{\sqrt{\epsilon}}\bm{b}\Big{)}\Big{(}\Bar{\bm{\nu}}-\frac{\Vert \bm{\Bar{\nu}} \Vert_2}{\sqrt{\epsilon}}\bm{b}\Big{)}^T$.
Notice that the left hand side term of the matrix inequality has $m$ eigenvalues of the form $\frac{\Vert \bm{\Bar{\nu}} \Vert_2}{2\sqrt{\epsilon}}$ while the right hand side term of the inequality is a rank $1$ matrix with $\frac{1}{4\sigma}\Big{(}\Bar{\bm{\nu}}-\frac{\Vert \bm{\Bar{\nu}} \Vert_2}{\sqrt{\epsilon}}\bm{b}\Big{)}^T\Big{(}\Bar{\bm{\nu}}-\frac{\Vert \bm{\Bar{\nu}} \Vert_2}{\sqrt{\epsilon}}\bm{b}\Big{)}$  as its only non-zero eigenvalue. Recalling the definition of $\sigma$, it can be easily verified that $\Big{(}\Bar{\bm{\nu}}-\frac{\Vert \bm{\Bar{\nu}} \Vert_2}{\sqrt{\epsilon}}\bm{b}\Big{)}^T\Big{(}\Bar{\bm{\nu}}-\frac{\Vert \bm{\Bar{\nu}} \Vert_2}{\sqrt{\epsilon}}\bm{b}\Big{)} = \frac{2\sigma\Vert \bm{\Bar{\nu}} \Vert_2}{\sqrt{\epsilon}}$ by simple algebraic manipulation. Accordingly, we have $\frac{\Vert \bm{\Bar{\nu}} \Vert_2}{2\sqrt{\epsilon}}\bm{I}_m \succeq \frac{1}{4\sigma}\Big{(}\Bar{\bm{\nu}}-\frac{\Vert \bm{\Bar{\nu}} \Vert_2}{\sqrt{\epsilon}}\bm{b}\Big{)}\Big{(}\Bar{\bm{\nu}}-\frac{\Vert \bm{\Bar{\nu}} \Vert_2}{\sqrt{\epsilon}}\bm{b}\Big{)}^T \implies \bm{\Psi} \succeq 0$. Thus, we have established that $\Bar{\bm{S}} \succeq 0 \implies \Bar{s}_0 \in \Sigma^2_1[\bm{z}, \bm{x}]$. Finally, we note that 
\[\Bar{s}_0 + \Bar{s}_1(\epsilon - \Vert\bm{A}\bm{x} - \bm{b}\Vert_2^2) + \sum_{i=1}^n \Bar{t}_i(x_iz_i-x_i) + \sum_{i=1}^n \Bar{r}_i(z_i^2-z_i) = f(\bm{z}, \bm{x}) - \bm{b}^T\Bar{\bm{\nu}}+\sqrt{\epsilon}\Vert \Bar{\bm{\nu}}\Vert_2.\]}
We have shown that given an arbitrary feasible solution to \eqref{opt:BPD_dual}, we can construct a solution that is feasible to \eqref{opt:poly_relax} that achieves the same objective value. Note that this construction holds for any $d \geq 1$. Thus, for any $d \in \mathbb{N}$ the optimal value of \eqref{opt:poly_relax} is at least as high as the optimal value of \eqref{opt:MISOC_W_relax}.
\end{proof}

We have shown that for any value of $d$, \eqref{opt:poly_relax} produces a lower bound on the optimal value of \eqref{opt:our_formulation_W} at least as strong as the bound given by \eqref{opt:MISOC_W_relax}. Unfortunately, \eqref{opt:poly_relax} suffers from scalability challenges as it requires solving a positive semidefinite program with PSD constraints on matrices with dimension ${2n+d \choose d} \times {2n+d \choose d}$. We further discuss the scalability of \eqref{opt:poly_relax} in Section \ref{sec:experiments}. Note that since \eqref{opt:poly_relax} is a maximization problem, any feasible solution (in particular a nearly optimal one) still consists of a valid lower bound on the optimal value of \eqref{opt:our_formulation_W}.

\section{Branch-and-Bound} \label{sec:BnB}

In this section, we propose a branch-and-bound algorithm in the sense of \citep{Land2010, little1966branch} that computes certifiably optimal solutions to Problem \eqref{opt:our_formulation} by solving the mixed integer second order cone reformulation given by \eqref{opt:MISOC_W}. We state explicitly our subproblem strategy in Section \ref{ssec:subprroblems}, before stating our overall algorithmic approach in  Section \ref{ssec:overallbnb}.

\subsection{Subproblems} \label{ssec:subprroblems}

Henceforth, for simplicity we will assume the weights $w_i$ take value $1$ for all $i$. What follows generalizes immediately to the setting where this assumption does not hold. Notice that \eqref{opt:MISOC_W} can be equivalently written as the two stage optimization problem given by $\min_{\bm{z} \in \{0,1\}^n}h(\bm{z})$ where we define $h(\bm{z})$ as:
\begin{equation}
\begin{aligned}
    h(\bm{z}) = &\min_{\bm{x}, \bm{\theta} \in \mathbb{R}^n} & & \sum_{i=1}^n z_i + \frac{1}{\gamma} \sum_{i=1}^n\theta_i\\
    &\text{s.t.} & & \Vert\bm{A}\bm{x} - \bm{b}\Vert_2^2 \leq \epsilon, \ x_i^2 \leq z_i\theta_i \ \forall \ i, \ \theta_i \geq 0 \ \forall i.
\end{aligned} \label{opt:h(z)}
\end{equation} Note that in general, there exist binary vectors $\Bar{\bm{z}} \in \{0,1\}^n$ such that the optimization problem in \eqref{opt:h(z)} is infeasible. For any such $\Bar{\bm{z}}$, we define $h(\Bar{\bm{z}}) = \infty$. We construct an enumeration tree that branches on the entries of the binary vector $\bm{z}$ which models the support of $\bm{x}$. A (partial or complete) sparsity pattern is associated with each node in the tree and is defined by disjoint collections $\mathcal{I}_0, \mathcal{I}_1 \subseteq [n]$. For indices $i \in \mathcal{I}_0$, we constrain $z_i = 0$ and for indices $j \in \mathcal{I}_1$, we constrain $z_j = 1$. We say that $\mathcal{I}_0$ and $\mathcal{I}_1$ define a complete sparsity pattern if $|\mathcal{I}_0| + |\mathcal{I}_1| = n$, otherwise we say that $\mathcal{I}_0$ and $\mathcal{I}_1$ define a partial sparsity pattern. A node in the tree is said to be terminal if its associated sparsity pattern is complete.

Each node in the enumeration tree has an associated subproblem, defined by the collections $\mathcal{I}_0$ and $\mathcal{I}_1$, which is given by:

\begin{equation}
\begin{aligned}
    \min_{\bm{z}\in \{0,1\}^n} \quad & h(\bm{z}), \quad \text{s.t.} \,\, & z_i = 0 \,\, \forall \, i \in \mathcal{I}_0, z_j = 1 \,\, \forall \, j \in \mathcal{I}_1.
\end{aligned} \label{opt:partial_pattern}
\end{equation} Note that if $\mathcal{I}_0 = \mathcal{I}_1 = \emptyset$, \eqref{opt:partial_pattern} is equivalent to \eqref{opt:MISOC_W} (under the assumetion that $w_i=1$ for all $i$).

\subsubsection{Subproblem Lower Bound}

Let $\mathcal{I} = \mathcal{I}_0 \cup \mathcal{I}_1$. We obtain a lower bound for \eqref{opt:partial_pattern} by relaxing the binary variables that are not fixed ($z_i$ such that $i \notin \mathcal{I}$) to take values within the interval $[0, 1]$. The resulting lower bound is given by

\begin{equation}
\begin{aligned}
    &\min_{\bm{z}, \bm{x}, \bm{\theta} \in \mathbb{R}^n} & & \sum_{i=1}^n z_i + \frac{1}{\gamma} \sum_{i=1}^n\theta_i\\
    &\text{s.t.} & & \Vert\bm{A}\bm{x} - \bm{b}\Vert_2^2 \leq \epsilon, \ x_i^2 \leq z_i\theta_i \ \forall \ i,\ 0 \leq z_i \leq 1 \ \forall \ i \notin \mathcal{I},\\
    & & & z_i = 0 \ \forall \ i \in \mathcal{I}_0, \ z_i = 1 \ \forall \ i \in \mathcal{I}_1, \ \theta_i \geq 0 \ \forall \ i.
\end{aligned} \label{opt:subproblem_relax}
\end{equation} Notice that for an arbitrary set $\Bar{\mathcal{I}}_0 \subseteq [n]$, problems \eqref{opt:partial_pattern} and \eqref{opt:subproblem_relax} are infeasible if and only if the set $\{\bm{x} : \Vert \bm{A}\bm{x}-\bm{b} \Vert_2^2 \leq \epsilon, x_i = 0 \,\, \forall \, i \in \Bar{\mathcal{I}}_0\}$ is empty. Moreover, observe it immediately follows that if \eqref{opt:partial_pattern} and \eqref{opt:subproblem_relax} are infeasible for $\Bar{\mathcal{I}}_0$, then they are also infeasible for any set $\hat{\mathcal{I}}_0 \subseteq [n]$ satisfying $\Bar{\mathcal{I}}_0 \subseteq \hat{\mathcal{I}}_0$. We use this observation in section \ref{ssec:overallbnb} to generate feasibility cuts whenever an infeasible subproblem is encountered in the branch-and-bound tree. Using a similar argument as in the proof of Theorem \ref{thm:BPD_equiv}, it can be shown that when $\gamma \geq \gamma_0 = \max_{x \in \mathcal{X}} \Vert \bm{x} \Vert_\infty^2$, \eqref{opt:subproblem_relax} is equivalent to the convex optimization problem given by \eqref{opt:subproblem_primal}:
\begin{equation}
\begin{aligned}
    &\min_{\bm{x} \in \mathbb{R}^n} & & \vert \mathcal{I}_1\vert + \frac{1}{\gamma} \sum_{i \in \mathcal{I}_1} x_i^2 + \frac{2}{\sqrt{\gamma}} \sum_{i \notin \mathcal{I}}\vert x_i \vert\\
    &\text{s.t.} & & \Vert\bm{A}\bm{x} - \bm{b}\Vert_2^2 \leq \epsilon, \ x_i = 0 \ \forall \ i \in \mathcal{I}_0,
\end{aligned} \label{opt:subproblem_primal}
\end{equation} where if $\bm{x}^\star$ is optimal to \eqref{opt:subproblem_primal}, then $(\bm{z}^\star, \bm{x}^\star, \bm{\theta}^\star)$ is optimal to \eqref{opt:subproblem_relax} taking $z_i^\star = \frac{\vert x_i \vert ^\star}{\sqrt{\gamma}}$ and $\theta_i^\star = x_i^{\star 2}$. Problem \eqref{opt:subproblem_primal} is a second order cone problem {\color{black}that} emits the following dual:
\begin{equation}
\begin{aligned}
    &\max_{\bm{\nu} \in \mathbb{R}^m} & & \vert \mathcal{I}_1\vert + \bm{b}^T\bm{\nu} - \sqrt{\epsilon} \Vert \bm{\nu} \Vert_2 - \frac{\gamma}{4} \bm{\nu}^T \sum_{i \in \mathcal{I}_1} (A_iA_i^T)\bm{\nu}\ \text{s.t.} \ \vert \bm{\nu}^TA_i \vert \leq \frac{2}{\sqrt{\gamma}} \ \forall \ i \notin \mathcal{I}.
\end{aligned} \label{opt:subproblem_dual}
\end{equation} Strong duality holds between \eqref{opt:subproblem_primal} and \eqref{opt:subproblem_dual} since $\bm{\nu} = 0$ is always a strictly feasible point in \eqref{opt:subproblem_dual} for any collections $\mathcal{I}_0, \mathcal{I}_1$ \citep{boyd2004convex}. In our branch-and-bound implementation described in \ref{ssec:overallbnb}, we compute lower bounds by solving \eqref{opt:subproblem_primal} using \verb|Gurobi|. We note that depending on the solver employed, it may be beneficial to compute lower bounds using \eqref{opt:subproblem_dual} in place of \eqref{opt:subproblem_primal}.


\subsubsection{Subproblem Upper Bound}

Recall that solving Problem \eqref{opt:main_problem} can be interpreted as determining the minimum number of columns from the input matrix $\bm{A}$ that must be selected such that the residual of the projection of the input vector $\bm{b}$ onto the span of the selected columns has $\ell_2$ norm equal to at most $\sqrt{\epsilon}$. The same interpretation holds for Problem \eqref{opt:our_formulation} under the assumption that the $\ell_2$ regularization term in the objective is negligible. 

Consider an arbitrary node in the branch-and-bound algorithm and let $\bm{x}^\star$ denote an optimal solution to \eqref{opt:subproblem_primal}. To obtain an {\color{black}upper bound} to \eqref{opt:partial_pattern}, we define an ordering on the columns of $\bm{A}$ and iteratively select columns from this ordering from largest to smallest until the $\ell_2$ norm of the residual of the projection of $\bm{b}$ onto the selected columns is less than $\sqrt{\epsilon}$. The ordering of the columns of $\bm{A}$ corresponds to sorting the entries of $\bm{x}^\star$ in decreasing absolute value. Specifically, we have $A_i \succeq A_j \iff \vert x^\star_i \vert \geq \vert x^\star_j \vert$. Algorithm \ref{alg:upper_bound} outlines this approach. For an arbitrary collection of indices $\mathcal{I}_t \subseteq [n]$, we let $\bm{A}(\mathcal{I}_t) \in \mathbb{R}^{m \times \vert  \mathcal{I}_t \vert}$ denote the matrix obtained by stacking the $\vert \mathcal{I}_t \vert$ columns of $\bm{A}$ corresponding to the indices in the set $\mathcal{I}_t$. Specifically, if $i_k$ denotes the $k^{th}$ entry of $\mathcal{I}_t$, then the $k^{th}$ column of $\bm{A}(\mathcal{I}_t)$ is $A_i$. Let $\bm{x}^{ub}$ denote the output of Algorithm \ref{alg:upper_bound}. The objective value achieved by $\bm{x}^{ub}$ in \eqref{opt:our_formulation} is the upper bound.

\begin{algorithm}
\caption{Branch-and-Bound Upper Bound}\label{alg:upper_bound}
\begin{algorithmic}[1]
\Require $\bm{A} \in \mathbb{R}^{m \times n}, \bm{b} \in \mathbb{R}^m, \epsilon > 0$. An optimal solution $\bm{x}^\star$ of \eqref{opt:subproblem_primal}.
\Ensure $\Bar{\bm{x}}$ is feasible to \eqref{opt:our_formulation}. 
\State $\mathcal{I}_0 \xleftarrow[]{} \emptyset$;
\State $\bm{r}_0 \xleftarrow[]{} \bm{b}$;
\State $t \xleftarrow[]{} 0$;
\State $\delta_0 \xleftarrow[]{} \Vert \bm{r}_0 \Vert_2^2$\;
\While{$\delta_t > \epsilon$}
    \State $i_t \xleftarrow[]{} \argmax_{i \in [n]\setminus \mathcal{I}_t} \vert x^\star_i \vert$;
    \State $\mathcal{I}_{t+1} \xleftarrow[]{} \mathcal{I}_t \cup i_t$;
    \State $\bm{x}_{t+1} \xleftarrow[]{} \big{[}\bm{A}(\mathcal{I}_{t+1})^T\bm{A}(\mathcal{I}_{t+1})\big{]}^\dagger\bm{A}(\mathcal{I}_{t+1})^T\bm{b}$;
    \State $\bm{r}_{t+1} \xleftarrow[]{} \bm{b} - \bm{A}(\mathcal{I}_{t+1})\bm{x}_{t+1}$;
    \State $\delta_{t+1} \xleftarrow[]{} \Vert \bm{r}_{t+1} \Vert_2^2$;
    \State $t \xleftarrow[]{} t + 1$;
\EndWhile
\State Define $\Bar{\bm{x}} \in \mathbb{R}^n$ as $\Bar{x}(i_k) = x_t(k)$ for $i_k \in \mathcal{I}_t$ and $\Bar{x}(i_k) = 0$ otherwise;
\State Return $\Bar{\bm{x}}$.
\end{algorithmic}
\end{algorithm}

The computational bottleneck of Algorithm \ref{alg:upper_bound} is computing the matrix inverse of \\  $\bm{A}(\mathcal{I}_{t})^T\bm{A}(\mathcal{I}_{t}) \in \mathbb{R}^{\vert  \mathcal{I}_t \vert \times \vert  \mathcal{I}_t \vert}$ at each iteration. Doing so explicitly at each iteration $t$ would require $O(\vert  \mathcal{I}_t \vert^3)$ operations. Letting $k^\star = \Vert \bm{x}^{ub} \Vert_0$ where $\bm{x}^{ub}$ is the output of Algorithm \ref{alg:upper_bound}, the total cost of executing these matrix inversions is \[\sum_{t=1}^{k^\star}\vert  \mathcal{I}_t \vert^3 = \sum_{t=1}^{k^\star} t^3 = \bigg{[}\frac{k^\star(k^\star+1)}{2}\bigg{]}^2 = O(k^{\star 4})\] However, it is possible to accelerate the computation of these matrix inverses by leveraging the fact that $\bm{A}(\mathcal{I}_{t})$ and $\bm{A}(\mathcal{I}_{t+1})$ differ only by the addition of one column and leveraging block matrix inversion which states that for matrices $\bm{C} \in \mathbb{R}^{n_1 \times n_1}, \bm{D} \in \mathbb{R}^{n_2 \times n_2}$ and $\bm{U}, \bm{V} \in \mathbb{R}^{n_1 \times n_2}$, we have:

\[\left[ 
\begin{array}{c|c} 
  \bm{C} & \bm{U} \\ 
  \hline 
  \bm{V}^T & \bm{D} 
\end{array} 
\right]^\dagger = \left[ 
\begin{array}{c|c} 
  \bm{C}^\dagger + \bm{C}^\dagger\bm{U}(\bm{D}-\bm{V}^T\bm{C}^\dagger\bm{U})^{-1}\bm{V}^T\bm{C}^\dagger & -\bm{C}^\dagger\bm{U}(\bm{D}-\bm{V}^T\bm{C}^\dagger\bm{U})^{-1} \\ 
  \hline 
  -(\bm{D}-\bm{V}^T\bm{C}^\dagger\bm{U})^{-1}\bm{V}^T\bm{C}^\dagger & (\bm{D}-\bm{V}^T\bm{C}^\dagger\bm{U})^{-1}
\end{array} 
\right] \] where it is assumed that the matrix $(\bm{D}-\bm{V}^T\bm{C}^\dagger\bm{U})$ is invertible \citep{petersen2008matrix}. Letting $n_1 = \vert \mathcal{I}_t \vert, n_2 = 1, \bm{C} = \bm{A}(\mathcal{I}_{t})^T\bm{A}(\mathcal{I}_{t}), \bm{U}=\bm{V} = \bm{A}(\mathcal{I}_{t})^Ta_{i_t}$, and $\bm{D} = a_{i_t}^Ta_{i_t}$, we can compute the matrix inverse of $\bm{A}(\mathcal{I}_{t+1})^T\bm{A}(\mathcal{I}_{t+1})$ using $O(\vert  \mathcal{I}_t \vert^2 + m\vert \mathcal{I}_t \vert)$ operations. With this implementation, the total cost of executing matrix inversions in Algorithm \ref{alg:upper_bound} becomes 
\[{\color{black}\sum_{t=1}^{k^\star} \vert \mathcal{I}_t \vert^2 + m\vert \mathcal{I}_t \vert = \sum_{t=1}^{k^\star} t^2 + m\sum_{t=1}^{k^\star} t = O(k^{\star 3} + mk^{\star 2})}\]which is a significant improvement over the naive $O(k^{\star 4})$ approach.



\subsection{Branch-and-Bound Algorithm} \label{ssec:overallbnb}

Having stated how we can compute upper and lower bounds to \eqref{opt:h(z)} at each node in the enumeration tree, we are now ready to present the branch-and-bound algorithm in its entirety. Algorithm \ref{alg:BNB} describes our approach which is based on the implementation by \citep{bertsimas2023SLR}. Though branching rules and node selection rules for branch-and-bound algorithms form a rich literature \citep{MORRISON201679}, we follow the design of \cite{bertsimas2023SLR} and employ the most fractional branching rule and least lower bound node selection rule.

Explicitly, for an arbitrary non-terminal node $p$, let $\bm{z}^*$ be the optimal vector $\bm{z}$ of the node relaxation given by \eqref{opt:subproblem_relax}. We branch on entry $i^\star = \argmin_{i \notin \mathcal{I}_0 \cup \mathcal{I}_1} |z_i - 0.5|$. When selecting a node to investigate, we select the node whose lower bound is equal to the global lower bound. If multiple such nodes exist, we choose arbitrarily from the collection of nodes satisfying this condition. Suppose that a given node produces a subproblem \eqref{opt:subproblem_primal} that is infeasible where we let $\Bar{\mathcal{I}}_0$ correspond to the zero index set of this node. Note that this implies that all child nodes of this node will also produce infeasible subproblems. Accordingly, to prune this region of the parameter space entirely, we introduce the feasibility cut $\sum_{i \in \Bar{\mathcal{I}}_0} z_i \geq 1$. Let $f(\bm{x}) = \Vert \bm{x} \Vert_0 +\frac{1}{\gamma} \Vert \bm{x} \Vert_2^2$, the objective function of \eqref{opt:our_formulation} and let $g(\mathcal{I}_0, \mathcal{I}_1)$ denote the optimal value of \eqref{opt:subproblem_primal} for any collections $\mathcal{I}_0, \mathcal{I}_1 \subseteq [n], \mathcal{I}_0 \cap \mathcal{I}_1 = \emptyset$. The final objective value returned by Algorithm \ref{alg:BNB} is given by $\min_i f(\bm{x}_i)$ where $\{\bm{x}_i\}_i$ denotes the collection of feasible solutions produced by Algorithm \ref{alg:upper_bound} at any point during the execution of Algorithm \ref{alg:BNB}. The output lower bound of Algorithm \ref{alg:BNB} is given by $\min_{(\mathcal{I}_0, \mathcal{I}_1) \in \mathcal{N}} g(\mathcal{I}_0, \mathcal{I}_1)$ where $\mathcal{N}$ denotes the set of non-discarded nodes upon the termination of Algorithm \ref{alg:BNB}.

{\color{black}Let $lb$ denote a lower bound to a given arbitrary minimization problem and let $ub$ denote the objective value achieved by a feasible solution $\Bar{\bm{x}}$ to the minimization problem. We call the solution $\Bar{\bm{x}}$ a \textit{$\delta$ globally optimal solution} to the given minimization problem if we have $lb \leq ub \leq (1+\delta) lb$.}

\begin{theorem}

Algorithm \ref{alg:BNB} terminates in a finite number of iterations and returns a $\delta$ globally optimal solution to \eqref{opt:main_problem}.

\end{theorem}

\begin{proof} The proof follows the proof of Theorem 21 in \cite{bertsimas2023SLR}. Note that Algorithm \ref{alg:BNB} can never visit a node more than once and that there is a finite number of partial and complete sparsity patterns (each corresponding to a possible node in the tree) because the set $\{0,1\}^n$ is discrete. Thus, Algorithm \ref{alg:BNB} terminates in a finite number of iterations. Moreover, upon termination we must have $\frac{ub-lb}{ub} \leq \delta$, therefore the output solution $\bm{\bar{x}}$ is $\delta$ globally optimal to problem \eqref{opt:our_formulation} by definition since $lb$ consists of a global lower bound and $\bm{\bar{x}}$ is feasible to \eqref{opt:our_formulation}. \end{proof}

\begin{algorithm}
\caption{Optimal Compressed Sensing}\label{alg:BNB}
\begin{algorithmic}[1]
\Require $\bm{A} \in \mathbb{R}^{m \times n}, \bm{b} \in \mathbb{R}^m, \epsilon,\gamma \in \mathbb{R}^+$. Tolerance parameter $\delta \geq 0$.
\Ensure $\bm{\bar{x}}$ solves \eqref{opt:our_formulation} within the optimality tolerance $\delta$.
\If{$\Vert (\bm{I} - \bm{A}\big{[}\bm{A}^T\bm{A}\big{]}^\dagger\bm{A}^T) \bm{b}  \Vert_2^2 > \epsilon$}
    \State Return $\emptyset$;
\EndIf
\If{$\Vert\bm{b}\Vert_2^2 \leq \epsilon$}
    \State Return $\bm{0}$;
\EndIf
\State $p_0 \xleftarrow[]{} (\mathcal{I}_0, \mathcal{I}_1) = (\emptyset, \emptyset)$;
\State $\mathcal{N} \xleftarrow[]{} \{p_0\}$;
\State $lb \xleftarrow[]{}$ optimal value of \eqref{opt:subproblem_primal};
\State $\bm{\bar{x}} \xleftarrow[]{}$ solution returned by Algorithm \ref{alg:upper_bound};
\State $ub \xleftarrow[]{} f(\bm{\bar{x}})$;
\While{$\frac{ub - lb}{ub} > \epsilon$}
    \State select $(\mathcal{I}_0, \mathcal{I}_1) \in \mathcal{N}$ according to the node selection rule;
    \State select an index $i \notin \mathcal{I}_0 \cup \mathcal{I}_1$ according to the branching rule;
    \For{$k = 0, 1$}{}
        \State $l \xleftarrow[]{} (k + 1) \text{ mod } 2$;
        \State newnode $\xleftarrow[]{} \Big{(}\big{(}\mathcal{I}_k \cup i\big{)}, \mathcal{I}_l \Big{)}$;
        \If{newnode violates an existing feasibility cut}
            \State continue;
        \EndIf
        \If{newnode is infeasible}
            \State Add the feasibility cut $\sum_{i \in \mathcal{I}_0} z_i \geq 1$;
        \EndIf
        \State \emph{lower} $\xleftarrow[]{}$ lowerBound(newnode);
        \State \emph{upper} $\xleftarrow[]{}$ upperBound(newnode) with feasible point $\bm{x}^\star$;
        \If{\emph{upper} $< ub$}
            \State $ub \xleftarrow[]{}$ \emph{upper};
            \State $\bm{\bar{x}} \xleftarrow[]{} \bm{x}^\star$;
            \State remove any node in $\mathcal{N}$ with \emph{lower} $\geq ub$;
        \EndIf
        \If{\emph{lower} $< ub$}
            \State add newnode to $\mathcal{N}$;
        \EndIf
    \EndFor
    \State remove $(\mathcal{I}_0, \mathcal{I}_1)$ from $\mathcal{N}$;
    \State update $lb$ to be the lowest value of \emph{lower} over $\mathcal{N}$;
\EndWhile
\State Return $\bm{\bar{x}}$, $lb$.
\end{algorithmic}
\end{algorithm}

We conclude the discussion of Algorithm \ref{alg:BNB} by describing two modifications that accelerate its execution time (or equivalently, improve its scalability) at the expense of sacrificing the universal optimality guarantee by drawing on techniques from the high dimensional sparse machine learning literature \citep{bertsimas2022backbone} and the deep learning literature \citep{GoodBengCour16}.

\subsubsection{Backbone Optimization}

Note that the total number of terminal nodes in the branch-and-bound tree is at most $\sum_{k=1}^n {n \choose k} = 2^n$ in the worst case so the total number of nodes can be upper bounded by $2^{n+1}-1$. Since the runtime of Algorithm \ref{alg:BNB} (and the feasible space) is proportional to the number of nodes explored which grows exponentially in $n$, reducing $n$ leads to reduced run time. Observe that if we knew in advance that the support of the optimal solution to \eqref{opt:our_formulation} was contained within a set of cardinality less than $n$, then we could run Algorithm \eqref{alg:BNB} on the corresponding reduced feature set which would result in improving the runtime of \eqref{alg:BNB} while preserving its optimality guarantee. Formally, let $\bm{x}^\star$ denote an optimal solution to \eqref{opt:our_formulation}. If we know a priori that $\text{support}(\bm{x}^\star) \subseteq \mathcal{I} \subset [n]$, then we can pass $\bm{A}(\mathcal{I})$ to Algorithm \ref{alg:BNB} in place of $\bm{A}$ without discarding $\bm{x}^\star$ from the feasible set. The speed up can be quite significant when $\vert \mathcal{I}\vert \ll n$. 

Knowing with certainty that $\text{support}(\bm{x}^\star) \subseteq \mathcal{I} \subset [n]$ a priori is too strong an assumption, however a more reasonable assumption is knowing a priori that with high probability there exists a good solution $\Bar{\bm{x}}$ with $\text{support}(\Bar{\bm{x}}) \subseteq \mathcal{I} \subset [n]$. In this setting, we can still pass $\bm{A}(\mathcal{I})$ to Algorithm \ref{alg:BNB} and benefit from an improved runtime at the expense of sacrificing optimality guarantees. In this setting, the columns of $\bm{A}(\mathcal{I})$ can be interpreted as a backbone for \eqref{opt:our_formulation} \citep{bertsimas2022backbone}. In practice, $\mathcal{I}$ can be taken to be the set of features selected by some heuristic method. In Section \ref{sec:experiments}, we take $\mathcal{I} = \{i: \vert \Bar{x}_i \vert \geq 10^{-6}\}$ where $\Bar{\bm{x}}$ is an optimal solution to \eqref{opt:BPD}.

\subsubsection{Early Stopping}

A common property of branch-and-bound algorithms is that the algorithm quickly arrives at an optimal (or near-optimal) solution early during the optimization procedure and spends the majority of its execution time improving the lower bound to obtain a certificate of optimality. 
Accordingly, this motivates halting Algorithm \ref{alg:BNB} before it terminates and taking its upper bound at the time of termination to be its output. Doing so is likely to still yield a high quality solution while reducing the Algorithm's runtime. In Section \ref{sec:experiments}, we place an explicit time limit on Algorithm \ref{alg:BNB} and return the current upper bound if the Algorithm has not already terminated before reaching the time limit. Note that this approach shares strong connections with early stopping in the training of neural networks \citep{GoodBengCour16}. A well studied property of over-parameterized neural networks is that as the optimization procedure {\color{black}progresses}, the error on the training data continues to decrease though the validation error plateaus and sometimes even increases. Given that the validation error is the metric of greater import, a common network training technique is to stop the optimization procedure after the validation error has not decreased for a prespecified number of iterations. To illustrate the connection in the case of Algorithm \ref{alg:BNB}, the upper bound loosely plays the role of the validation error while the lower bound loosely plays the role of the training error. Note that the neural network literature suggests an alternate approach to early stopping Algorithm \ref{alg:BNB} (instead of an explicit time limit) by terminating the algorithm after the upper bound has remained unchanged after visiting some prespecified number of nodes in the enumeration tree.

\section{Computational Results} \label{sec:experiments}

We evaluate the performance of our branch-and-bound algorithm (Algorithm \ref{alg:BNB}, with $\gamma = \sqrt{n}$), our second order cone lower bound \eqref{opt:MISOC_W_relax} (with $\gamma = \sqrt{n}$) and our semidefinite lower bound \eqref{opt:poly_relax} (with $\gamma = \sqrt{n}$ and $d=1$) implemented in Julia 1.5.2 using the \verb|JuMP.jl| package version 0.21.7, using \verb|Gurobi| version 9.0.3 to solve all second order cone optimization (sub)problems and using \verb|Mosek| version 9.3 to solve all semidefinite optimization problems. We compare our methods against Basis Pursuit Denoising (BPD) given by \eqref{opt:BPD}, Iterative Reweighted $\ell_1$ Minimizaton (IRWL1) described in Section \ref{sec:IRWL1} and Orthogonal Matching Pursuit (OMP) described in Section \ref{sec:OMP}. We perform experiments {\color{black}using synthetic data and two real world data sets.} We conduct our experiments on MIT’s Supercloud Cluster \citep{reuther2018interactive}, which hosts Intel Xeon Platinum 8260 processors {\color{black}and cores with $4$GB RAM}. To bridge the gap between theory and practice, we have made our code freely available on \url{GitHub} at \url{github.com/NicholasJohnson2020/DiscreteCompressedSensing.jl}.

\subsection{Synthetic Data Experiments}

To evaluate the performance of Algorithm \ref{alg:BNB}, BPD, IRWL1 and OMP on synthetic data, we consider the sparsity of the solution returned by each method, its accuracy (ACC), true positive rate (TPR) and true negative rate (TNR). Let $\bm{x}^{true} \in \mathbb{R}^n$ denote the ground truth and consider an arbitrary vector $\hat{\bm{x}} \in \mathbb{R}^n$. Let $\mathcal{I}^{true} = \{i : \vert x^{true}_i \vert > 10^{-4}\}$, $\hat{\mathcal{I}} = \{i : \vert \hat{x}_i \vert > 10^{-4}\}$. The sparsity of $\hat{\bm{x}}$ is given by $\vert \hat{\mathcal{I}} \vert$. We define the accuracy of $\hat{\bm{x}}$ as \[ACC(\hat{\bm{x}}) = \frac{\sum_{i \in \mathcal{I}^{true}} \mathbbm{1}\{\vert \hat{x}_i \vert > 10^{-4}\} + \sum_{i \notin \mathcal{I}^{true}} \mathbbm{1}\{\vert \hat{x}_i \vert \leq 10^{-4}\}}{n}.\] Similarly, we define the true positive rate of $\hat{\bm{x}}$ as \[TPR(\hat{\bm{x}}) = \frac{\sum_{i \in \mathcal{I}^{true}} \mathbbm{1}\{\vert \hat{x}_i \vert > 10^{-4}\}}{\vert \hat{\mathcal{I}} \vert},\] and we define the true negative rate of $\hat{\bm{x}}$ as \[TNR(\hat{\bm{x}}) = \frac{\sum_{i \notin \mathcal{I}^{true}} \mathbbm{1}\{\vert \hat{x}_i \vert \leq 10^{-4}\}}{n-\vert \hat{\mathcal{I}} \vert}.\] To evaluate the performance of \eqref{opt:MISOC_W_relax} and \eqref{opt:poly_relax}, we consider the strength of the lower bound and execution time of each method. We seek to answer the following questions:

\begin{enumerate}
    \item How does the performance of Algorithm \ref{alg:BNB} compare to state-of-the-art methods such as BPD, IRWL1 and OMP on synthetic data?
    \item How is the performance of Algorithm \ref{alg:BNB} affected by the number of features $n$, the underlying sparsity $k$ of the ground truth, and the tolerance parameter $\epsilon$?
    \item How does the strength of the lower bound produced by \eqref{opt:poly_relax} compare to that produced by \eqref{opt:MISOC_W_relax}?
\end{enumerate}

\subsubsection{Synthetic Data Generation}

To generate synthetic data $\bm{x} \in \mathbb{R}^n, \bm{A} \in \mathbb{R}^{m \times n}$ and $\bm{b} \in \mathbb{R}^m$, we first select a random subset of indices $\mathcal{I}^{true} \subset [n]$ that has cardinality $k$ ($\vert \mathcal{I}^{true} \vert = k$) and sample $x_i \sim N(0, \frac{\sigma^2}{n})$ for $i \in \mathcal{I}^{true}$ (for $i \notin \mathcal{I}^{true}$, we fix $x_i=0$). Next, we sample $A_{ij} \sim N(0, \frac{\sigma^2}{n})$ where $\sigma > 0$ is a parameter that controls the signal to noise ratio. We fix $\sigma = 10$ and $m=100$ throughout all experiments unless stated otherwise. Next, we set $\bm{b} = \bm{A}\bm{x} + \bm{n}$ where $n_j \sim N(0, \sigma^2)$. Finally, we set $\epsilon = \alpha \Vert \bm{b} \Vert_2^2$. $\alpha \in [0,1]$ is a parameter that can be thought of as controlling the proportion of observations that are allowed to go unexplained by a solution to \eqref{opt:our_formulation}.

\subsubsection{Sensitivity to \texorpdfstring{$n$}{n}} \label{sssec:n_exp}

We present a comparison of Algorithm \ref{alg:BNB} with BPD, IRWL1 and OMP as we vary the number of features $n$. In these experiments, we fixed $k=10$, and $\alpha=0.2$ across all trials. We varied $n \in \{100, 200, 300, 400, 500, 600, 700, 800\}$ and we performed $100$ trials for each value of $n$. We give Algorithm \ref{alg:BNB} a cutoff time of $10$ minutes. For IRWL1, we terminate the algorithm after the $50^{th}$ iteration or after two subsequent iterates are equal up to numerical tolerance. Formally, letting $\Bar{\bm{x}}_t$ denote the iterate after iteration $t$ of IRWL1, we terminate the algorithm if either $t > 50$ or if $\Vert \Bar{\bm{x}}_t - \Bar{\bm{x}}_{t-1} \Vert_2 \leq 10^{-6}$. Additionally, we further sparsify the solutions returned by BPD (respectively IRWL1) by performing a greedy rounding following the procedure defined by Algorithm \ref{alg:upper_bound} where we pass the solution returned by BPD (respectively IRWL1) as input to the algorithm in place of an optimal solution to \eqref{opt:subproblem_primal}.

We report the sparsity, accuracy (ACC), true positive rate (TPR) and true negative rate (TNR) for each method in Figure \ref{fig:synthetic_N}. We additionally report the sparsity accuracy and execution time for each method in {\color{black} Tables \ref{tbl:N_sparsity}, \ref{tbl:N_acc} and \ref{tbl:N_time}}. The performance metric of greatest interest is the sparsity. Our main findings from this set of experiments are:
\begin{enumerate}
    \item Algorithm \ref{alg:BNB} systematically produces sparser solutions than OMP, IRWL1 and BPD. This trend holds in all but one trial (see Table \ref{tbl:N_sparsity}). Algorithm \ref{alg:BNB} on average produces solutions that are $2.71\%$ more sparse than OMP, $16.62\%$ more sparse than BPD and $6.04\%$ more sparse than IRWL1. BPD is the poorest performing method in terms of sparsity of the fitted solutions. We remind the reader that sparsity is computed only after a greedy rounding of the BPD (respectively IRWL1) solution. The sparsity of the BPD (respectively IRWL1) solution prior to rounding is much greater. Indeed, before further sparsifying the BPD (respectively IRWL1) solution, the solution returned by Algorithm \ref{alg:BNB} is on average $66.33\%$ (respectively $6.21\%$) more sparse than the BPD (respectively IRWL1) solution. The sparsity of solutions returned by all methods increases as the number of features $n$ increases.
    
    \item Algorithm \ref{alg:BNB} marginally outperforms the benchmark methods on accuracy with the exception of the first two parameter configurations ($n=100$ and $n=200$, see Table \ref{tbl:N_acc}). The accuracy of all methods tends to trend upwards with increasing $n$.
    
    \item The TPR and TNR of all methods are roughly comparable across these experiments. The TPR of all methods decreases while the TNR increases as the number of features $n$ is increased. {\color{black} The sharp drop off in the TPR of all methods as $n$ increases from $100$ to $400$ is consistent with the number of features selected by each method increasing sharply as $n$ increases from $100$ to $400$. The latter results in the denominator used in the TPR computation to grow sharply which produces the observed behavior.}
    
\end{enumerate}

\begin{figure*}[h]\centering
  \includegraphics[width=0.9\textwidth]{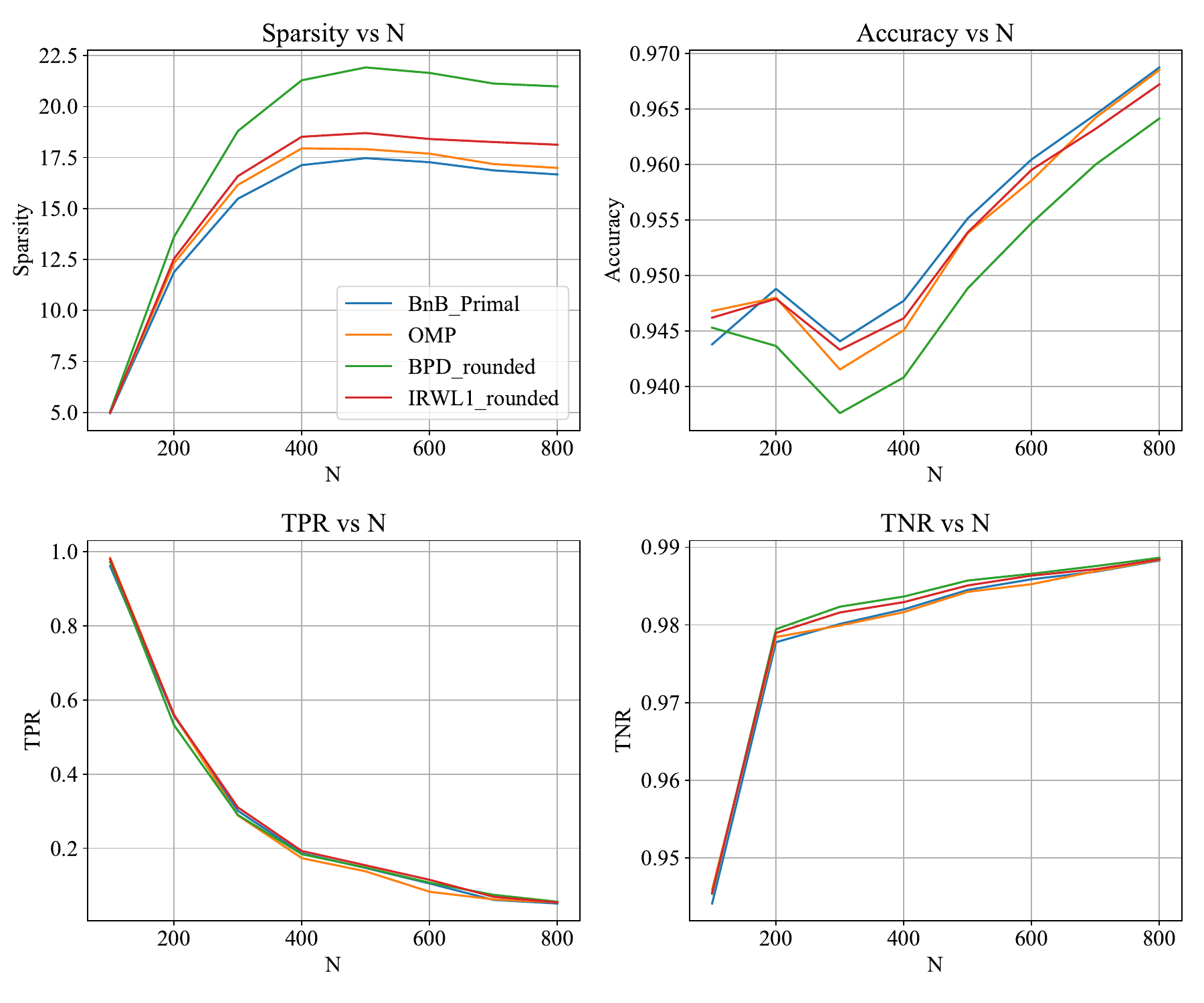}
  \caption{Sparsity (top left), accuracy (top right), true positive rate (bottom left) and true negative rate (bottom right) versus $n$ with $k=10$, and $\alpha=0.2$. Averaged over $100$ trials for each parameter configuration.}
  \label{fig:synthetic_N}
\end{figure*}

\begin{table}[h]
  \centering
  \caption{\color{black}Comparison of the sparsity of solutions returned by \eqref{alg:BNB}, OMP, IRWL1 and BPD for different values of $n$. Averaged over $100$ trials for each parameter configuration.}\label{tbl:N_sparsity}
  \begin{tabular}{c || cccc}
\toprule
    \multicolumn{1}{c}{} & \multicolumn{4}{c}{Sparsity} \\
    \cmidrule(l){1-1} \cmidrule(l){2-5}
    N & Algorithm 2 & OMP & IRWL1 & BPD \\
\midrule
100 &  \textbf{5.0} &          \textbf{5.0} &            \textbf{5.0} &          \textbf{5.0} \\
200 & \textbf{11.9} &         12.3 &           12.5 &         13.6 \\
300 & \textbf{15.5} &         16.2 &           16.6 &         18.8 \\
400 & \textbf{17.1} &         17.9 &           18.5 &         21.3 \\
500 & \textbf{17.5} &         17.9 &           18.7 &         21.9 \\
600 & \textbf{17.3} &         17.7 &           18.4 &         21.6 \\
700 & \textbf{16.9} &         17.2 &           18.3 &         21.1 \\
800 & \textbf{16.7} &         17.0 &           18.1 &         21.0 \\
\bottomrule
\end{tabular}

\end{table}

\begin{table}[h]
  \centering
  \caption{\color{black}Comparison of the accuracy of solutions returned by \eqref{alg:BNB}, OMP, IRWL1 and BPD for different values of $n$. Averaged over $100$ trials for each parameter configuration.}\label{tbl:N_acc}
  \begin{tabular}{c || cccc}
\toprule
    \multicolumn{1}{c}{} & \multicolumn{4}{c}{Accuracy} \\
    \cmidrule(l){1-1} \cmidrule(l){2-5}
    N & Algorithm 2 & OMP & IRWL1 & BPD \\
\midrule
100 &          0.944 & \textbf{0.947} &     0.946 &   0.945 \\
200 & \textbf{0.949} &          0.948 &     0.948 &   0.944 \\
300 & \textbf{0.944} &          0.942 &     0.943 &   0.938 \\
400 & \textbf{0.948} &          0.945 &     0.946 &   0.941 \\
500 & \textbf{0.955} &          0.954 &     0.954 &   0.949 \\
600 & \textbf{0.960} &          0.959 &     \textbf{0.960} &   0.955 \\
700 & \textbf{0.965} &          0.964 &     0.963 &   0.960 \\
800 & \textbf{0.969} &          \textbf{0.969} &     0.967 &   0.964 \\
\bottomrule
\end{tabular}

\end{table}

\begin{table}[h]
  \centering
  \caption{\color{black}Comparison of the execution time of solutions returned by \eqref{alg:BNB}, OMP, IRWL1 and BPD for different values of $n$. Averaged over $100$ trials for each parameter configuration.}\label{tbl:N_time}
  \begin{tabular}{c || cccc}
\toprule
    \multicolumn{1}{c}{} & \multicolumn{4}{c}{Execution Time (milliseconds)} \\
    \cmidrule(l){1-1} \cmidrule(l){2-5}
    $N$ & Algorithm 2 & OMP & IRWL1 & BPD \\
\midrule
100 &      2048.646 &   \textbf{5.717} &         463.636 &       146.111 \\
200 &    334804.020 &  \textbf{13.212} &        1109.263 &       234.263 \\
300 &    574501.859 &  \textbf{25.141} &        1630.212 &       297.849 \\
400 &    601792.939 &  \textbf{42.919} &        2181.636 &       351.717 \\
500 &    601424.020 &  \textbf{72.535} &        2435.141 &       405.131 \\
600 &    601451.838 & \textbf{110.364} &        3118.465 &       433.626 \\
700 &    601572.848 & \textbf{166.525} &        3674.980 &       504.626 \\
800 &    601716.929 & \textbf{231.980} &        3865.788 &       540.859 \\
\bottomrule
\end{tabular}

\end{table}

\subsubsection{Sensitivity to \texorpdfstring{$k$}{k}} \label{sssec:k_exp}

We present a comparison of Algorithm \ref{alg:BNB} with BPD, IRWL1 and OMP as we vary $k$ the sparsity of the underlying ground truth signal. In these experiments, we fixed $n=200$ and $\alpha=0.2$ across all trials. We varied $k \in \{10, 15, 20, 25, 30, 35, 40, 45, 50, 55\}$ and we performed $100$ trials for each value of $k$. We give Algorithm \ref{alg:BNB} a cutoff time of $10$ minutes.

We report the sparsity, accuracy (ACC), true positive rate (TPR) and true negative rate (TNR) for each method in Figure \ref{fig:synthetic_K}. We additionally report the sparsity, accuracy and execution time for each method in {\color{black} Tables \ref{tbl:K_sparsity}, \ref{tbl:K_acc} and \ref{tbl:K_time}.} Our main findings from this set of experiments are:
\begin{enumerate}
    \item Consistent with the results in the previous section, Algorithm \ref{alg:BNB} systematically produces sparser solutions than OMP, IRWL1 and BPD. This trend holds across trials (see Table \ref{tbl:K_sparsity}.  Algorithm \ref{alg:BNB} on average produces solutions that are $4.78\%$ more sparse than OMP, $10.73\%$ more sparse than BPD and $4.20\%$ more sparse than IRWL1. Before further sparsifying the BPD (respectively IRWL1) solution, the solution returned by Algorithm \ref{alg:BNB} is on average $62.97\%$ (respectively $4.29\%$) more sparse than the BPD (respectively IRWL1) solution. BPD is again the poorest performing method in terms of sparsity of the fitted solutions. IRWL1 and OMP produce comparably sparse solutions. The sparsity of solutions returned by all methods initially decreases than subsequently increases as the sparsity level $k$ of the ground truth signal increases.
    
    \item Algorithm \ref{alg:BNB} is competitive with OMP and IRWL1 on accuracy and slightly outperforms BPD on accuracy for larger values of $k$ The accuracy of all methods trends downwards with increasing $k${\color{black}, suggesting that the feature identification problem becomes more challenging for larger values of $k$ in this regime}.
    
    \item The TPR and TNR of Algorithm \ref{alg:BNB}, OMP, and IRWL1 are comparable across these experiments. The TPR and TNR of BPD is competitive with the other methods for small values of $k$, but slightly deteriorates for larger values of $k$.
\end{enumerate}

\begin{figure*}[h]\centering
  \includegraphics[width=0.9\textwidth]{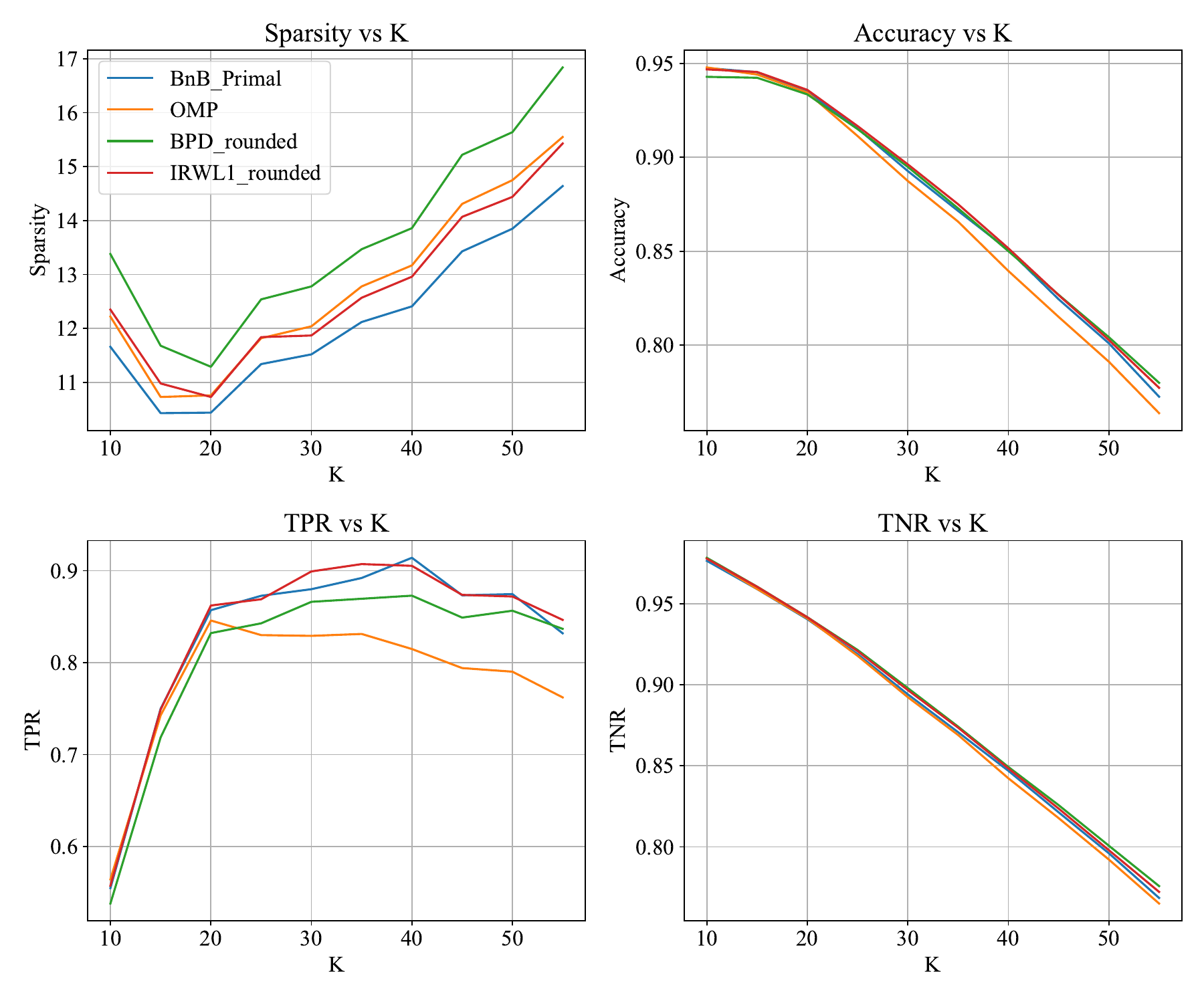}
  \caption{Sparsity (top left), accuracy (top right), true positive rate (bottom left) and true negative rate (bottom right) versus $k$ with $N=200$ and $\alpha=0.2$. Averaged over $100$ trials for each parameter configuration.
  }
  \label{fig:synthetic_K}
\end{figure*}

\begin{table}[h]
  \centering
  \caption{\color{black}Comparison of the sparsity of solutions returned by \eqref{alg:BNB}, OMP, IRWL1 and BPD for different values of $k$. Averaged over $100$ trials for each parameter configuration.}\label{tbl:K_sparsity}
  \begin{tabular}{c || cccc}
\toprule
    \multicolumn{1}{c}{} & \multicolumn{4}{c}{Sparsity} \\
    \cmidrule(l){1-1} \cmidrule(l){2-5}
    K & Algorithm 2 & OMP & IRWL1 & BPD \\
\midrule
10 & \textbf{11.7} &         12.2 &           12.3 &         13.4 \\
15 & \textbf{10.4} &         10.7 &           11.0 &         11.7 \\
20 & \textbf{10.4} &         10.8 &           10.7 &         11.3 \\
25 & \textbf{11.3} &         11.8 &           11.8 &         12.5 \\
30 & \textbf{11.5} &         12.0 &           11.9 &         12.8 \\
35 & \textbf{12.1} &         12.8 &           12.6 &         13.5 \\
40 & \textbf{12.4} &         13.2 &           13.0 &         13.9 \\
45 & \textbf{13.4} &         14.3 &           14.1 &         15.2 \\
50 & \textbf{13.8} &         14.8 &           14.4 &         15.6 \\
55 & \textbf{14.6} &         15.6 &           15.4 &         16.8 \\
\bottomrule
\end{tabular}

\end{table}

\begin{table}[h]
  \centering
  \caption{\color{black}Comparison of the accuracy of solutions returned by \eqref{alg:BNB}, OMP, IRWL1 and BPD for different values of $k$. Averaged over $100$ trials for each parameter configuration.}\label{tbl:K_acc}
  \begin{tabular}{c || cccc}
\toprule
    \multicolumn{1}{c}{} & \multicolumn{4}{c}{Accuracy} \\
    \cmidrule(l){1-1} \cmidrule(l){2-5}
    K & Algorithm 2 & OMP & IRWL1 & BPD \\
\midrule
10 &   \textbf{0.948} & \textbf{0.948} &          0.947 &          0.943 \\
15 &   \textbf{0.945} &          0.944 & \textbf{0.945} &          0.942 \\
20 &   0.935 &          0.934 & \textbf{0.936} &          0.933 \\
25 &   0.915 &          0.911 & \textbf{0.917} &          0.915 \\
30 &   0.893 &          0.887 & \textbf{0.896} &          0.895 \\
35 &   0.872 &          0.866 & \textbf{0.875} &          0.873 \\
40 &   0.851 &          0.840 & \textbf{0.852} &          0.850 \\
45 &   0.825 &          0.815 &          0.827 & \textbf{0.827} \\
50 &   0.801 &          0.791 &          0.803 & \textbf{0.804} \\
55 &   0.772 &          0.764 &          0.777 & \textbf{0.780} \\
\bottomrule
\end{tabular}

\end{table}

\begin{table}[h]
  \centering
  \caption{\color{black}Comparison of the execution time of solutions returned by \eqref{alg:BNB}, OMP, IRWL1 and BPD for different values of $k$. Averaged over $100$ trials for each parameter configuration.}\label{tbl:K_time}
  \begin{tabular}{c || cccc}
\toprule
    \multicolumn{1}{c}{} & \multicolumn{4}{c}{Execution Time (milliseconds)} \\
    \cmidrule(l){1-1} \cmidrule(l){2-5}
    $K$ & Algorithm 2 & OMP & IRWL1 & BPD \\
\midrule
10 &    305993.475 & \textbf{13.182} &        1270.000 &       341.454 \\
15 &    199128.374 & \textbf{12.556} &        1144.818 &       284.071 \\
20 &    119282.667 & \textbf{12.646} &        1080.535 &       278.596 \\
25 &    139224.525 & \textbf{13.263} &        1081.202 &       327.151 \\
30 &    171844.485 & \textbf{12.909} &        1169.798 &       314.192 \\
35 &    193257.535 & \textbf{12.798} &        1163.121 &       361.485 \\
40 &    231721.737 & \textbf{13.404} &        1151.455 &       277.647 \\
45 &    314269.394 & \textbf{13.495} &        1142.374 &       308.919 \\
50 &    351790.071 & \textbf{13.727} &        1219.707 &       315.081 \\
55 &    412429.717 & \textbf{14.010} &        1260.899 &       289.616 \\
\bottomrule
\end{tabular}

\end{table}

\subsubsection{Sensitivity to \texorpdfstring{$\epsilon$}{epsilon}} \label{sssec:e_exp}

We present a comparison of Algorithm \ref{alg:BNB} with BPD, IRWL1 and OMP as we vary $\alpha$ which controls the value of the parameter $\epsilon$. Recall we have $\epsilon = \alpha \Vert \bm{b} \Vert_2^2$, so $\alpha$ can loosely be interpreted as the fraction of the measurements $\bm{b}$ that can be unexplained by the returned solution to \eqref{opt:our_formulation}. In these experiments, we fixed $n=200$ and $k=10$ across all trials. We varied $\alpha \in \{0.05, 0.1, 0.15, \ldots, 0.9\}$ and we performed $100$ trials for each value of $\alpha$. We give Algorithm \ref{alg:BNB} a cutoff time of $10$ minutes.

We report the sparsity, accuracy (ACC), true positive rate (TPR) and true negative rate (TNR) for each method in Figure \ref{fig:synthetic_alpha}, and we report the sparsity, accuracy and execution time for each method in {\color{black} Tables \ref{tbl:alpha_sparsity}, \ref{tbl:alpha_acc} and \ref{tbl:alpha_time}.} Consistent with previous experiments, Algorithm \ref{alg:BNB} outperforms the benchmark methods in terms of sparsity of the returned solution while having comparable performance on accuracy, TPR and TNR. Here, Algorithm \ref{alg:BNB} on average produces solutions that are $2.40\%$ more sparse than OMP, $5.92\%$ more sparse than BPD and $2.54\%$ more sparse than IRWL1. Before further sparsifying the BPD (respectively IRWL1) solution, the solution returned by Algorithm \ref{alg:BNB} is on average $59.23\%$ (respectively $2.62\%$) more sparse than the BPD (respectively IRWL1) solution.

\begin{figure*}[h]\centering
  \includegraphics[width=0.9\textwidth]{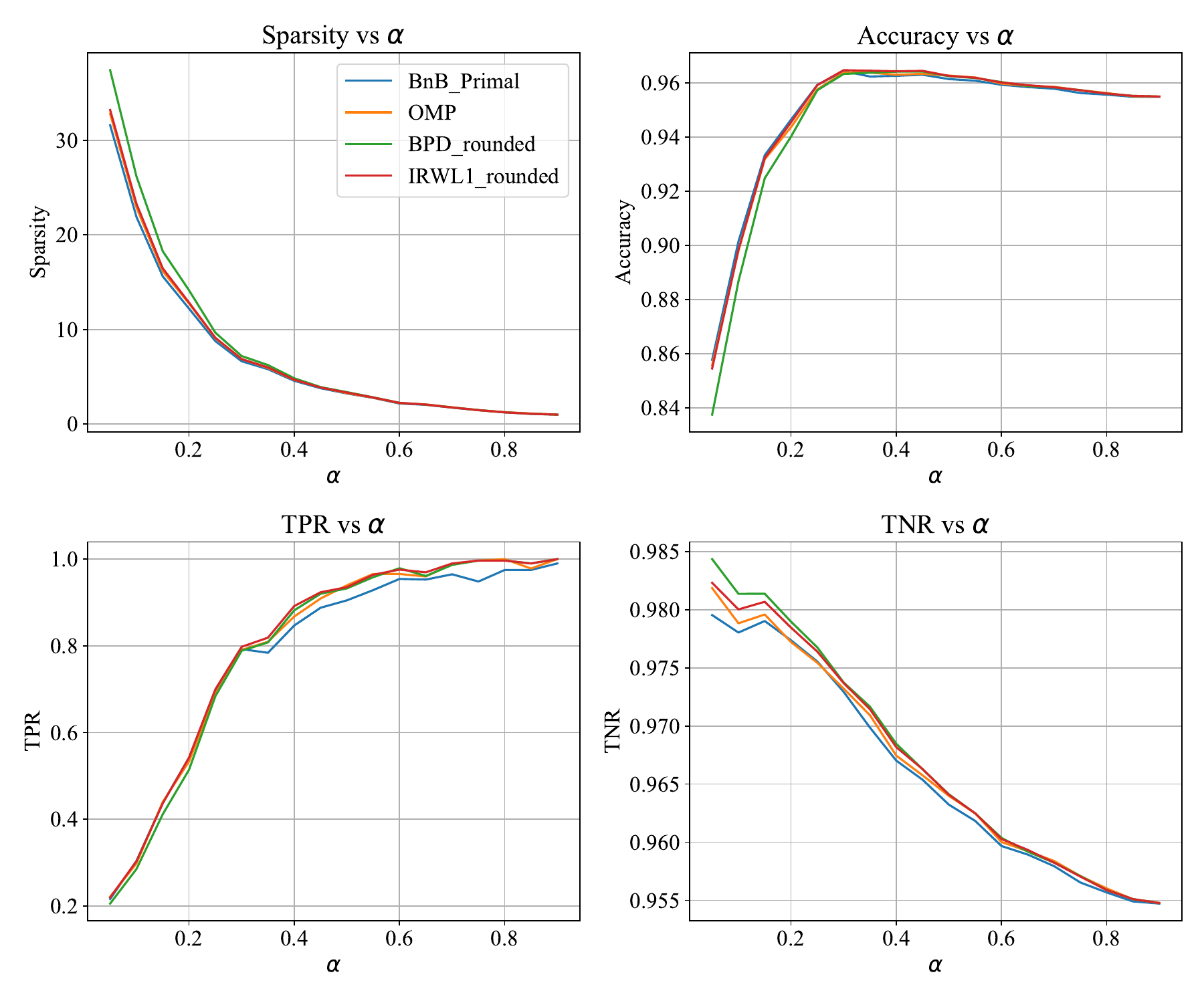}
  \caption{Sparsity (top left), accuracy (top right), true positive rate (bottom left) and true negative rate (bottom right) versus $\alpha$ with $n=200$ and $k=10$. Averaged over $100$ trials for each parameter configuration.
  }
  \label{fig:synthetic_alpha}
\end{figure*}

\begin{table}[h]
  \centering
  \caption{\color{black}Comparison of the sparsity of solutions returned by \eqref{alg:BNB}, OMP, IRWL1 and BPD for different values of $\alpha$. Averaged over $100$ trials for each parameter configuration.}\label{tbl:alpha_sparsity}
  \begin{tabular}{c || cccc}
\toprule
    \multicolumn{1}{c}{} & \multicolumn{4}{c}{Sparsity} \\
    \cmidrule(l){1-1} \cmidrule(l){2-5}
    $\alpha$ & Algorithm 2 & OMP & IRWL1 & BPD \\
\midrule
    0.05 & \textbf{31.6} &         32.8 &           33.2 &         37.4 \\
    0.10 & \textbf{21.9} &         22.9 &           23.3 &         26.2 \\
    0.15 & \textbf{15.6} &         16.1 &           16.5 &         18.3 \\
    0.20 & \textbf{12.2} &         12.7 &           12.8 &         14.1 \\
    0.25 &  \textbf{8.8} &          9.1 &            9.1 &          9.7 \\
    0.30 &  \textbf{6.6} &          6.9 &            6.9 &          7.2 \\
    0.35 &  \textbf{5.8} &          5.9 &            6.0 &          6.2 \\
    0.40 &  \textbf{4.6} &          4.7 &            4.7 &          4.8 \\
    0.45 &  \textbf{3.8} &          3.9 &            3.9 &          3.9 \\
    0.50 &  \textbf{3.2} &          3.3 &            3.3 &          3.4 \\
    0.55 &  \textbf{2.8} &          2.8 &            2.8 &          2.8 \\
    0.60 &  \textbf{2.2} &          2.2 &            2.2 &          2.2 \\
    0.65 &  \textbf{2.0} &          2.1 &            2.1 &          2.1 \\
    0.70 &  \textbf{1.7} &          1.8 &            1.8 &          1.8 \\
    0.75 &  \textbf{1.5} &          1.5 &   \textbf{1.5} &          1.5 \\
    0.80 &  \textbf{1.2} &          1.3 &   \textbf{1.2} & \textbf{1.2} \\
    0.85 &  \textbf{1.1} &          1.1 &            1.1 &          1.1 \\
    0.90 &  \textbf{1.0} & \textbf{1.0} &   \textbf{1.0} & \textbf{1.0} \\
\bottomrule
\end{tabular}

\end{table}

\begin{table}[h]
  \centering
  \caption{\color{black}Comparison of the accuracy of solutions returned by \eqref{alg:BNB}, OMP, IRWL1 and BPD for different values of $\alpha$. Averaged over $100$ trials for each parameter configuration.}\label{tbl:alpha_acc}
  \begin{tabular}{c || cccc}
\toprule
    \multicolumn{1}{c}{} & \multicolumn{4}{c}{Accuracy} \\
    \cmidrule(l){1-1} \cmidrule(l){2-5}
    $\alpha$ & Algorithm 2 & OMP & IRWL1 & BPD \\
\midrule
    0.05 & \textbf{0.858} &          0.856 &          0.855 &          0.838 \\
    0.10 & \textbf{0.901} &          0.898 &          0.898 &          0.887 \\
    0.15 & \textbf{0.933} &          0.932 &          0.932 &          0.925 \\
    0.20 & \textbf{0.946} &          0.944 &          0.946 &          0.940 \\
    0.25 & \textbf{0.959} &          0.958 &          0.959 &          0.957 \\
    0.30 &          0.964 &          0.964 & \textbf{0.965} &          0.963 \\
    0.35 &          0.962 &          0.964 & \textbf{0.965} &          0.964 \\
    0.40 &          0.963 &          0.963 & \textbf{0.964} &          0.964 \\
    0.45 &          0.963 &          0.963 & \textbf{0.965} &          0.964 \\
    0.50 &          0.962 & \textbf{0.963} &          0.963 &          0.963 \\
    0.55 &          0.961 & \textbf{0.962} &          0.962 &          0.962 \\
    0.60 &          0.959 &          0.960 &          0.960 & \textbf{0.960} \\
    0.65 &          0.959 &          0.959 & \textbf{0.959} &          0.959 \\
    0.70 &          0.958 & \textbf{0.959} &          0.958 &          0.958 \\
    0.75 &          0.956 & \textbf{0.957} &          0.957 &          0.957 \\
    0.80 &          0.956 & \textbf{0.956} &          0.956 &          0.956 \\
    0.85 &          0.955 &          0.955 & \textbf{0.955} & \textbf{0.955} \\
    0.90 &          0.955 & \textbf{0.955} & \textbf{0.955} & \textbf{0.955} \\
\bottomrule
\end{tabular}

\end{table}

\begin{table}[h]
  \centering
  \caption{\color{black}Comparison of the execution time of solutions returned by \eqref{alg:BNB}, OMP, IRWL1 and BPD for different values of $\alpha$. Averaged over $100$ trials for each parameter configuration.}\label{tbl:alpha_time}
  \begin{tabular}{c || cccc}
\toprule
    \multicolumn{1}{c}{} & \multicolumn{4}{c}{Execution Time (milliseconds)} \\
    \cmidrule(l){1-1} \cmidrule(l){2-5}
    $\alpha$ & Algorithm 2 & OMP & IRWL1 & BPD \\
\midrule
    0.05 &    603205.808 & \textbf{22.798} &        5164.919 &       934.717 \\
    0.10 &    577124.980 & \textbf{18.061} &        2158.465 &       522.212 \\
    0.15 &    454366.242 & \textbf{15.535} &        1801.596 &       636.172 \\
    0.20 &    343487.778 & \textbf{14.636} &        2013.323 &       967.657 \\
    0.25 &    181672.232 & \textbf{13.970} &        1477.374 &       654.091 \\
    0.30 &     81074.212 & \textbf{14.253} &        1087.737 &       510.162 \\
    0.35 &     61929.838 & \textbf{13.475} &        1503.990 &       904.788 \\
    0.40 &     41775.950 & \textbf{13.838} &        1213.343 &       684.101 \\
    0.45 &     15927.495 & \textbf{18.939} &        1056.717 &       531.091 \\
    0.50 &     10387.101 & \textbf{13.687} &        1384.343 &       872.040 \\
    0.55 &      7032.818 & \textbf{18.091} &        1001.253 &       530.556 \\
    0.60 &       858.253 & \textbf{13.212} &        1118.929 &       640.242 \\
    0.65 &      1012.121 & \textbf{18.404} &        1280.869 &       926.727 \\
    0.70 &       494.808 & \textbf{18.333} &         783.697 &       514.081 \\
    0.75 &       499.677 & \textbf{19.899} &         696.495 &       511.717 \\
    0.80 &       749.626 & \textbf{20.283} &         965.010 &       910.485 \\
    0.85 &       479.091 & \textbf{20.606} &         539.768 &       537.869 \\
    0.90 &       462.808 & \textbf{13.394} &         485.869 &       514.566 \\
\bottomrule
\end{tabular}

\end{table}

\subsubsection{Lower Bound Performance} \label{sssec:lb_exp}

In Section \ref{sec:MISOC}, we reformulated \eqref{opt:our_formulation} exactly as a mixed integer second order cone problem and illustrated multiple approaches to obtain lower bounds on the optimal value of the reformulation. In this Section, we compare the strength of the second order cone relaxation given by \eqref{opt:MISOC_W_relax} and the semidefinite cone relaxation given by \eqref{opt:poly_relax}. We fixed $k=10$ and we varied $\alpha \in \{0.05, 0.1, 0.15, \ldots, 0.9\}$. We report results for $(n,m)=(25,100)$ in Figure \ref{fig:synthetic_lb_n25} and $(n,m)=(50,25)$ in Figure \ref{fig:synthetic_lb_n50}. We performed $100$ trials for each value of $\alpha$. Letting $lb^{SOC}$ denote the optimal value of \eqref{opt:MISOC_W_relax} and $lb^{SOS}$ denote the optimal value of \eqref{opt:poly_relax}, we define the SOS lower bound improvement to be $\frac{lb^{SOS}-lb^{SOC}}{lb^{SOC}}$.

Consistent with the Theorem \ref{thm:SOS_better}, Problem \eqref{opt:poly_relax} produces a stronger lower bound than Problem \eqref{opt:MISOC_W_relax} at the expense of being more computationally intensive to compute due to the presence of positive semidefinite constraints. On average, the bound produced by \eqref{opt:poly_relax} is $8.92\%$ greater than the bound produced by \eqref{opt:MISOC_W_relax}. These results suggests that if Problem \eqref{opt:poly_relax} can be solved to optimality or near optimality efficiently at scale, it could potentially be used to accelerate Algorithm \ref{alg:BNB} by producing stronger lower bounds than the current approach, thereby allowing for a more aggressive pruning of the feasible space. Off the shelf interior point methods suffer from scalability challenges for semidefinite optimization problems. 

\begin{figure*}[h]\centering
  \includegraphics[width=0.9\textwidth]{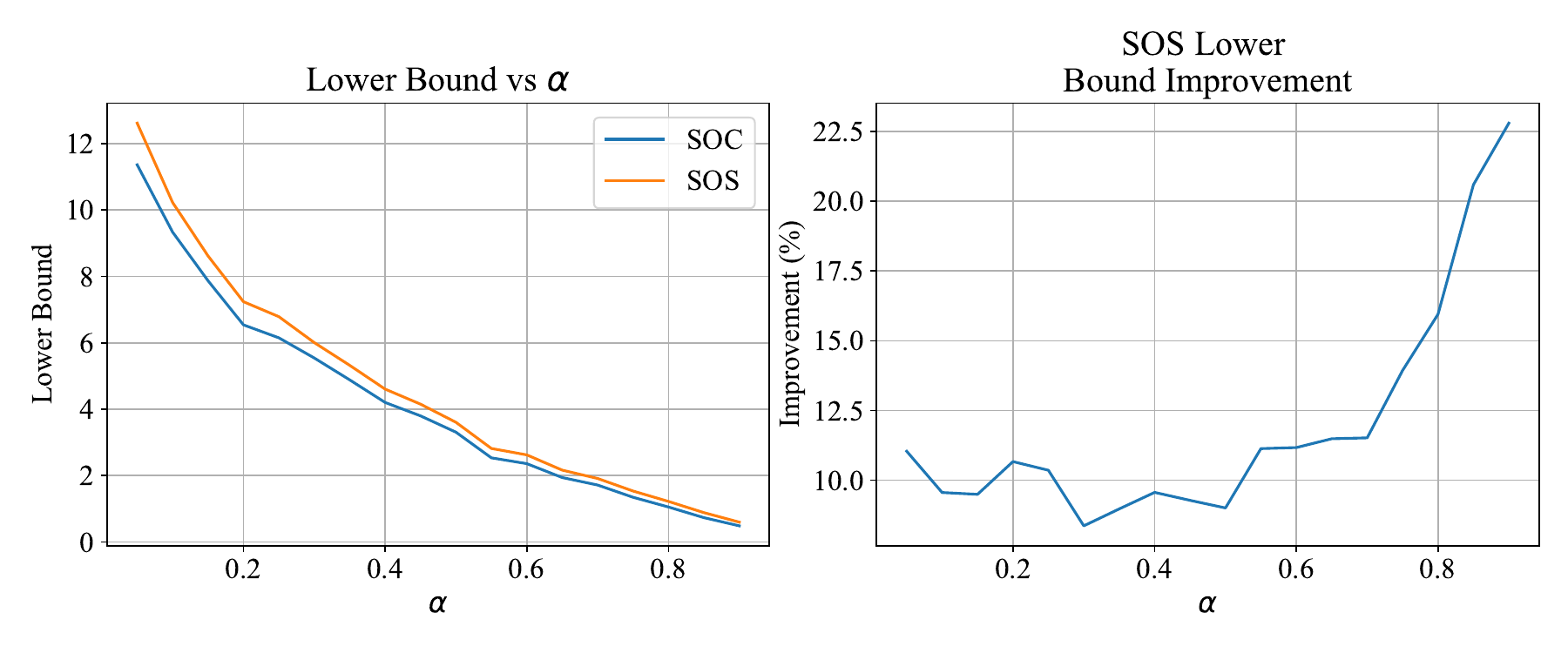}
  \caption{Problem \eqref{opt:our_formulation} lower bound (left) produced by Problem \eqref{opt:MISOC_W_relax} (SOC) and Problem \eqref{opt:poly_relax} (SOS) with $d=1$. Percent improvement of Problem \eqref{opt:our_formulation} lower bound of compared to (right). $n=25, m=100$ and $k=10$.}
  \label{fig:synthetic_lb_n25}
\end{figure*}

\begin{figure*}[h]\centering
  \includegraphics[width=0.9\textwidth]{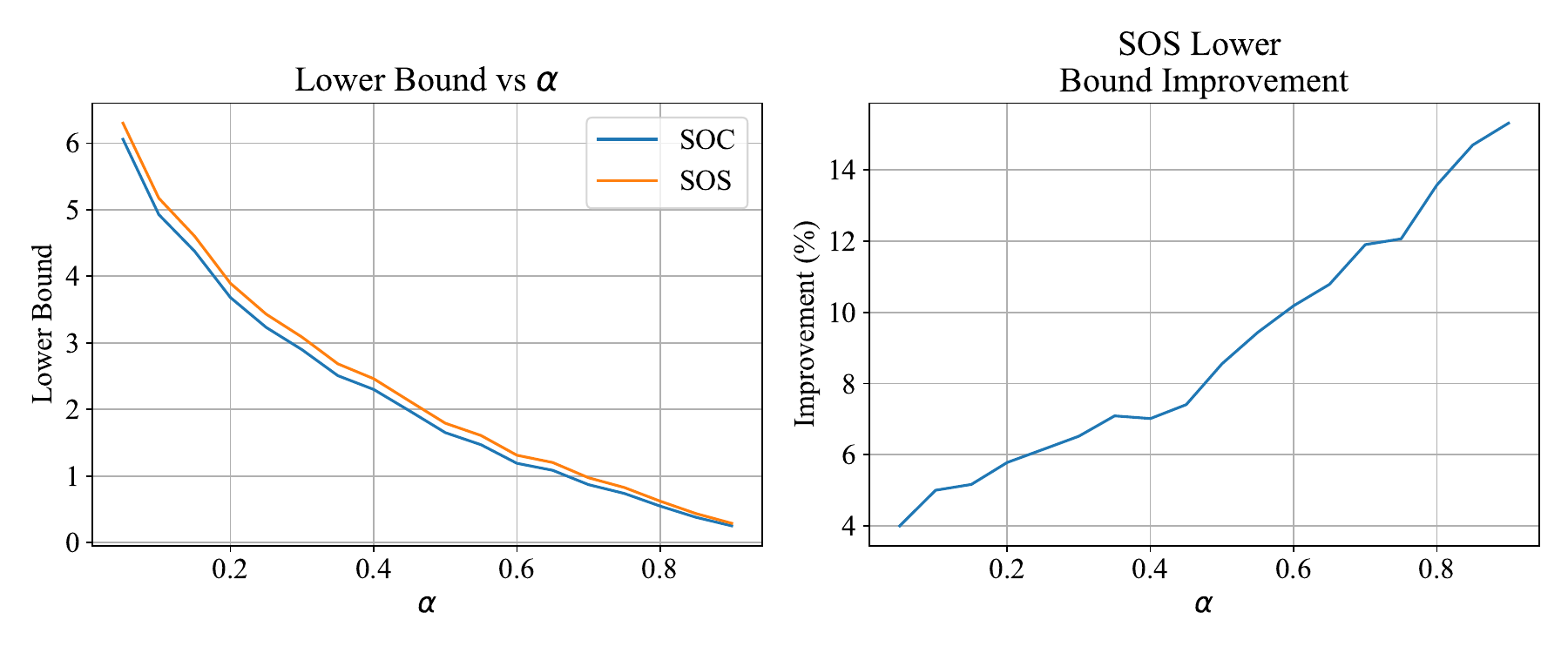}
  \caption{Problem \eqref{opt:our_formulation} lower bound (left) produced by Problem \eqref{opt:MISOC_W_relax} (SOC) and Problem \eqref{opt:poly_relax} (SOS) with $d=1$. Percent improvement of Problem \eqref{opt:our_formulation} lower bound of compared to (right). $n=50, m=25$ and $k=10$.}
  \label{fig:synthetic_lb_n50}
\end{figure*}

{\color{black} \subsection{Electrocardiogram Signal Acquisition} \label{ssec:ECG_exp}}

We seek to answer the following question: how does the performance of Algorithm \ref{alg:BNB} compare to state-of-the-art methods such as BPD, IRWL1 and OMP on {\color{black}signal processing using} real world data? {\color{black}To evaluate performance, we consider the problem of compressed sensing for electrocardiogram (ECG) acquisition \citep{chen2019compressed}.} We obtain real ECG recording samples from the MIT-BIH Arrhythmia Database (\url{https://www.physionet.org/content/mitdb/1.0.0/}) and consider the performance of the methods in terms {\color{black}of} sparsity of the returned signal and reconstruction error between the returned signal and the true signal.

\subsubsection{ECG Experiment Setup}

We employ the same $100$ ECG recordings sampled at $360$ Hz from the MIT-BIH Arrhythmia Database that are used in \citep{chen2019compressed}. These recordings collectively originate from $10$ distinct patients (each contributing $10$ independent recordings) and the recording length of an individual record is $1024$. In keeping with \cite{chen2019compressed}, we use $30$ ECG recordings as a training set to fit an overcomplete dictionary $\bm{D}$ via the K-SVD method \citep{aharon2006k}. We fit a dictionary with $2000$ atoms, meaning that $\bm{D} \in \mathbb{R}^{1024 \times 2000}$ and $\bm{X}^{train} \approx \bm{D} \bm{\Theta}$ where $\bm{X}^{train} \in \mathbb{R}^{1024 \times 30}$ is a matrix whose columns are the training ECG signals and $\bm{\Theta} \in \mathbb{R}^{2000 \times 30}$ is a sparse matrix. Each column of $\bm{\Theta}$ should be thought of as a (sparse) representation of the corresponding column of $\bm{X}^{train}$ in the dictionary given by $\bm{D}$ ($\Vert \bm{\Theta} \Vert_0 \ll \Vert \bm{X}^{train} \Vert_0$). We employ the Bernouilli sensing matrix $\bm{B} \in \mathbb{R}^{40 \times 1024}$ considered by \cite{chen2019compressed}. Given an ECG signal $\bm{x}^{test} \in \mathbb{R}^{1024}$, we consider the perturbed observations $\bm{s} = \bm{B}(\bm{x}^{test}+\bm{\eta})$ where $\bm{\eta} \in \mathbb{R}^{1024}$ is a vector of mean $0$ normal perturbations with variance $\Big{(}\frac{\Vert \bm{x}^{test} \Vert_1}{4 \cdot 1024}\Big{)}^2 \bm{I}$. Figure \ref{fig:ECG_signal} illustrates the ECG signal and perturbed ECG signal for record 31 of the dataset. With these preliminaries, we consider the reconstruction problem given by
\begin{equation}
\begin{aligned}
    &\min_{\bm{\theta} \in \mathbb{R}^{2000}} & & \Vert\bm{\theta}\Vert_0 + \frac{1}{\gamma} \Vert\bm{\theta}\Vert_2^2\\
    &\text{s.t.} & & \Vert\bm{B}\bm{D}\bm{\theta} - \bm{s}\Vert_2^2 \leq \epsilon.
\end{aligned} \label{opt:ecg_formulation}
\end{equation} where we set $\epsilon = 1.05 \cdot \Vert \bm{s} - \bm{B}\bm{x}^{test} \Vert_2^2$. Note that \eqref{opt:ecg_formulation} is equivalent to \eqref{opt:our_formulation} where $(\bm{\theta}, \bm{BD}, \bm{s})$ play the role of $(\bm{x}, \bm{A}, \bm{b})$ and we have $(n, m)=(2000, 40)$. Letting $\hat{\bm{\theta}}$ denote a feasible solution to \eqref{opt:ecg_formulation} returned by one of the solution methods, we employ $10^{-4}$ as the numerical threshold to compute the sparsity $\Vert \hat{\bm{\theta}} \Vert_0$ of $\hat{\bm{\theta}}$ and we define the $\ell_q$ reconstruction error of $\hat{\bm{\theta}}$ as $\frac{\Vert \bm{D}\hat{\bm{\theta}} - \bm{x}^{test}\Vert_q^q}{\Vert \bm{x}^{test}\Vert_q^q}$ for $q \in \{1,2\}$.

\begin{figure*}[h]\centering
  \includegraphics[width=0.9\textwidth]{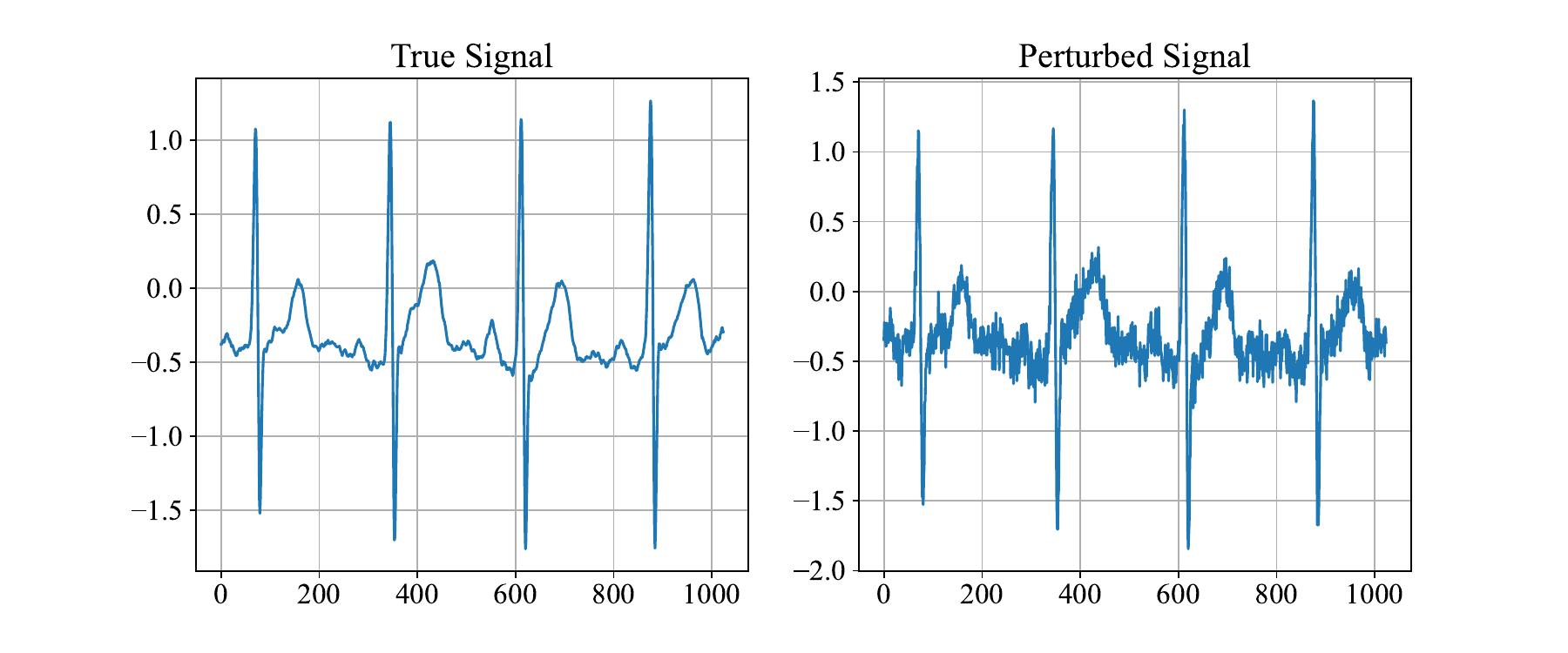}
  \caption{Ground truth ECG signal (left) and perturbed signal (right) for ECG record 31.}
  \label{fig:ECG_signal}
\end{figure*}

\subsubsection{ECG Computational Results}

We present a comparison of BPD, IRWL1, OMP and Algorithm \ref{alg:BNB} as we vary the regularization parameter $\gamma$ in \eqref{opt:ecg_formulation}. We considered values of $\gamma$ in the set \[ \Gamma =\{(8a+0.01)f(n): a \in [14], f(n) \in \{\sqrt{n}, n, n^2\}\},\] and we evaluate performance on the $70$ ECG recordings that are not part of the training set used to fit the overcomplete dictionary $\bm{D}$. We give Algorithm \ref{alg:BNB} a cutoff time of $5$ minutes. As in the synthetic experiments, we terminate IRWL1 after the $50^{th}$ iteration or after two subsequent iterates are equal up to numerical tolerance. Moreover, we sparsify the solutions returned by BPD and IRWL1 using the procedure given by Algorithm \ref{alg:upper_bound} as done in the synthetic experiments.

Figure \ref{fig:ECG_frontier} illustrates the average $\ell_1$ error (left) and average $\ell_2$ error (right) versus the average sparsity of solutions returned by each method. Each red dot corresponds to the performance of Algorithm \ref{alg:BNB} for a fixed value of $\gamma \in \Gamma$. Given that more sparse solutions and solutions with lesser $\ell_1$ (respectively $\ell_2$) error are desirable, Figure \ref{fig:ECG_frontier} demonstrates that as we vary $\gamma$, the solutions returned by Algorithm \ref{alg:BNB} trace out an efficient frontier that dominates the solutions returned by BPD, IRWL1 and OMP. Indeed, for all benchmark methods (BPD, IRWL1 and OMP), there is a value of $\gamma$ such that Algorithm \ref{alg:BNB} finds solutions that achieve lower sparsity and lower reconstruction error than the solution returned by the benchmark method. For the same $\ell_2$ reconstruction error, Algorithm \ref{alg:BNB} can produce solutions that are on average $3.88\%$ more sparse than IRWL1, $6.29\%$ more sparse than BPD and $19.70\%$ more sparse than OMP. For the same sparsity level, Algorithm \ref{alg:BNB} can produce solutions that have on average $1.42\%$ lower $\ell_2$ error than IRWL1, $2.66\%$ lower $\ell_2$ error than BPD and $28.23\%$ lower $\ell_2$ error than OMP. Thus, Algorithm \ref{alg:BNB} outperforms BPD, IRWL1 and OMP on this real world dataset.

\begin{figure*}[h]\centering
  \includegraphics[width=0.9\textwidth]{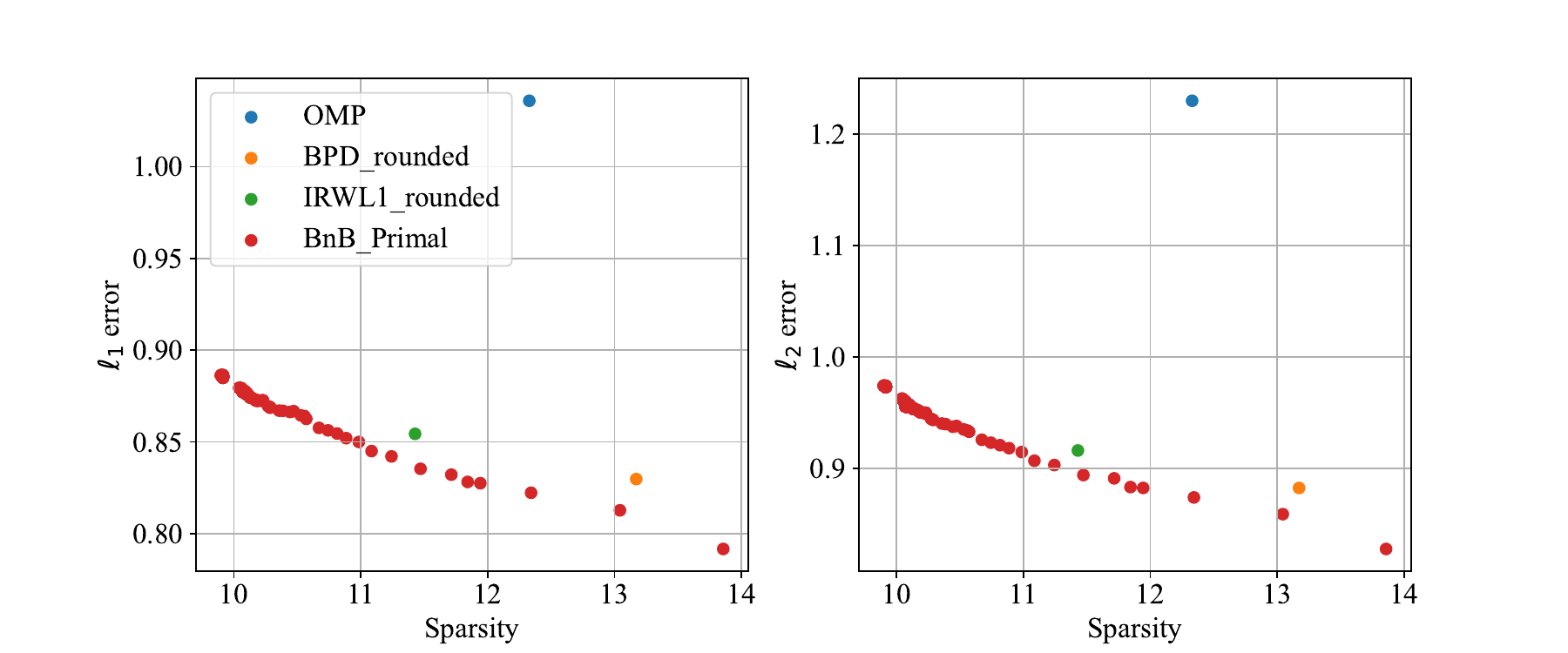}
  \caption{$\ell_1$ reconstruction error (left) and $\ell_2$ reconstruction error (right) versus $\ell_0$ norm (Sparsity) for ECG reconstructions obtained using OMP, BPD, IRWL1 and Algorithm \ref{alg:BNB} for varying values of $\gamma$. $n=2000$ and $m=40$.
  }
  \label{fig:ECG_frontier}
\end{figure*}

\subsection{Multi-Label Classification} \label{ssec:ML_exp}

We seek to answer the following question: how does the performance of Algorithm \ref{alg:BNB} compare to state-of-the-art methods such as BPD, IRWL1 and OMP when used as part of a learning task on real world data? To evaluate the performance of the aforementioned algorithms in this setting, we consider the problem of multi-label classification (MLC) \citep{hsu2009multi}. In MLC, given a dataset $\{(\bm{x}_i, \bm{y}_i)\}_{i=1}^n$ where $\bm{x}_i \in \mathbb{R}^m$ denotes a feature vector and $\bm{y}_i \in \{0, 1\}^d$ denotes a label vector, we seek to learn a function $f:\mathbb{R}^m \rightarrow \{0, 1\}^d$ that can correctly classify unseen examples. Note that this is a more general problem than a typical multi-class classification problem with $d$ classes since in this setting each example can be a member of multiple classes simultaneously. When the number of classes $d$ is very large, training $d$ individual classifiers in a one-against-all approach becomes prohibitively expensive. In this regime, if the label vectors tend to be sparse relative to the size of the latent dimension $d$, one approach is to project the labels into a smaller dimension $k \ll d$ using a linear compression function and then to learn $k$ predictors mapping the feature space to the projected label space. Predictions can then be made by applying the $k$ learned predictors on a given example and subsequently using a compressed sensing based reconstruction algorithm to transform the prediction from projected label space to the latent label space \citep{hsu2009multi, kapoor2012multilabel, som2016learning, kai2017compressed}. Explicitly, given a linear map $\bm{P} \in \mathbb{R}^{k \times d}$ where $k \ll d$, we form a transformed dataset $\{(\bm{x}_i, \bm{P}\bm{y}_i)\}_{i=1}^n$ and for each $j \in \{1, \ldots, k\}$ learn a regression function $h_j: \mathbb{R}^m \rightarrow \mathbb{R}$ using the data given by $\{(\bm{x}_i, [\bm{P}\bm{y}_i]_j)\}_{i=1}^n$. Given a reconstruction function $g: \mathbb{R}^{k \times d} \times \mathbb{R}^k \rightarrow \mathbb{R}^d$ that outputs a sparse vector such that $\bm{P} [g(\bm{P}, \bm{l})] \approx \bm{l}$, we make predictions on a test point $\bm{x}_{test}$ by returning $g(\bm{P}, [h_1(\bm{x}_{test}), \ldots, h_k(\bm{x}_{test})])$. In Sections \ref{sssec:ML_setup} and \ref{sssec:ML_results}, we evaluate the performance of this approach when Algorithm \ref{alg:BNB}, BPD, IRWL1 and OMP are used as the reconstruction function $g$.

\subsubsection{MLC Experiment Setup} \label{sssec:ML_setup}

We obtain a text data set consisting of web pages and descriptive textual tags from the former social bookmarking service \textit{del.icio.us} that was originally collected by \cite{tsoumakas2008effective}. The dataset consists of roughly $16000$ web pages having $d = 983$ unique labels. Each web page on average is associated with $19$ labels. The web pages are each represented as a bag-of-words feature vector. Further details can be found in \cite{tsoumakas2008effective}. We use $80\%$ of the data for training and withhold $20\%$ for testing. We fix $k = 100$ during our experiments. We fix the linear map $\bm{P} = \frac{1}{\sqrt{k}} \bm{H}$ where $\bm{H} \in \mathbb{R}^{k \times d}$ is a matrix obtained by choosing $k$ random rows from the $d \times d$ Hadamard matrix in keeping with the approach taken by \citep{hsu2009multi}. We use ridge regression as our base learning algorithm for the regression functions $\{h_j\}_{j=1}^k$ and tune the ridge parameter via leave one out cross validation over the training set. Given a test data point $\bm{x}_{test} \in \mathbb{R}^m$, we output a predicted label $\bm{\hat{y}} \in \mathbb{R}^d$ by solving

\begin{equation}
\begin{aligned}
    &\min_{\bm{y} \in \mathbb{R}^{d}} & & \Vert\bm{y}\Vert_0 + \frac{5}{\sqrt{d}} \Vert\bm{y}\Vert_2^2\\
    &\text{s.t.} & & \Vert\bm{P}\bm{y} - h(\bm{x}_{test})\Vert_2^2 \leq \epsilon.
\end{aligned} \label{opt:mlc_formulation}
\end{equation} where $[h(\bm{x}_{test})]_j = h_j(\bm{x}_{test})$ and we set $\epsilon = 0.25 \cdot \Vert h(\bm{x}_{test}) \Vert_2^2$. Letting $\hat{\bm{y}}$ denote a feasible solution to \eqref{opt:mlc_formulation} returned by one of the solution methods, we are interested in measuring the accuracy, precision and recall of the support of $\hat{\bm{y}}$ relative to the true label $\bm{y}_{test}$. Let $\mathcal{I}^{true} = \{i : [\bm{y}_{test}]_i = 1\}$ and $\hat{\mathcal{I}} = \{i : \vert \hat{y}_i \vert > 10^{-4}\}$. We define the accuracy of $\hat{\bm{y}}$ as \[ACC(\hat{\bm{y}}) = \frac{\sum_{i \in \mathcal{I}^{true}} \mathbbm{1}\{\vert \hat{y}_i \vert > 10^{-4}\} + \sum_{i \notin \mathcal{I}^{true}} \mathbbm{1}\{\vert \hat{y}_i \vert \leq 10^{-4}\}}{d}.\] Similarly, we define the precision of $\hat{\bm{y}}$ as \[Precision(\hat{\bm{y}}) = \frac{\sum_{i \in \mathcal{I}^{true}} \mathbbm{1}\{\vert \hat{y}_i \vert > 10^{-4}\}}{\vert \mathcal{I}^{true} \vert},\] and we define the recall (true positive rate) of $\hat{\bm{y}}$ as \[TPR(\hat{\bm{y}}) = \frac{\sum_{i \in \mathcal{I}^{true}} \mathbbm{1}\{\vert \hat{y}_i \vert > 10^{-4}\}}{\vert \hat{\mathcal{I}} \vert}.\]

\subsubsection{MLC Computational Results} \label{sssec:ML_results}

We present a comparison of test set classification performance when each of BPD, IRWL1, OMP and Algorithm \ref{alg:BNB} are used as the reconstruction algorithm $g$. We give Algorithm \ref{alg:BNB} a cutoff time of $10$ minutes. As in the synthetic and ECG experiments, we terminate IRWL1 after the $50^{th}$ iteration or after two subsequent iterates are equal up to numerical tolerance. Moreover, we sparsify the solutions returned by BPD and IRWL1 using the procedure given by Algorithm \ref{alg:upper_bound} as done in the synthetic and ECG experiments. Figure \ref{fig:mlc_plot} illustrates the average accuracy achieved by each method on the test set. We see that that Algorithm $2$ achieves slightly greater average accuracy than the benchmark methods. Moreover, Table \ref{tbl:mlc_metrics} illustrates the percentage of test examples on which each method has the best performance for accuracy, precision and recall. We see that Algorithm $2$ exhibits superior performance than the baseline methods across all metrics of interest.

\begin{figure*}[h]\centering
  \includegraphics[width=0.9\textwidth]{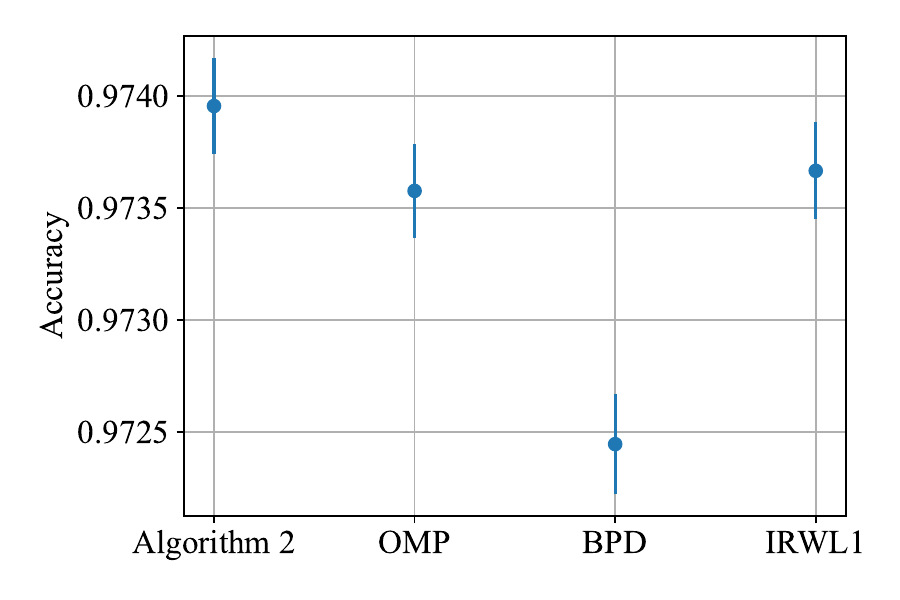}
  \caption{\color{black}Test set prediction accuracy obtained using OMP, BPD, IRWL1 and Algorithm \ref{alg:BNB} as the label reconstruction algorithm.
  }
  \label{fig:mlc_plot}
\end{figure*}

\begin{table}[h]
  \centering
  \caption{\color{black}Comparison of the accuracy of solutions returned by \eqref{alg:BNB}, OMP, IRWL1 and BPD for different values of $\alpha$. Averaged over $100$ trials for each parameter configuration.}\label{tbl:mlc_metrics}
  \begin{tabular}{c || cccc}
\toprule
    \multicolumn{1}{c}{} & \multicolumn{4}{c}{Top Performing Algorithm Frequency (\%)} \\
    \cmidrule(l){1-1} \cmidrule(l){2-5}
    Metric & Algorithm 2 & OMP & IRWL1 & BPD \\
\midrule
    Accuracy & \textbf{53.47} & 28.78 & 11.66 & 6.09 \\
    Precision & \textbf{49.39} & 21.42 & 15.26 & 13.93 \\
    Recall & \textbf{54.90} & 14.14 & 3.20 & 27.76 \\
\bottomrule
\end{tabular}

\end{table}

\subsection{Summary of Findings}

We now summarize our findings from our numerical experiments. In Sections \ref{sssec:n_exp}-\ref{sssec:e_exp}, we see that across all experiments using synthetic data, Algorithm \ref{alg:BNB} produces solutions that are on average $6.22\%$ more sparse than the solutions returned by state of the art benchmark methods after they are further sparsified by greedy rounding. If we omit greedy rounding, Algorithm \ref{alg:BNB} produces solutions that are on average $17.17\%$ more sparse in our synthetic experiments. In Section \ref{sssec:lb_exp}, we find that the bound produced by \eqref{opt:poly_relax} is on average $8.92\%$ greater than the bound produced by \eqref{opt:MISOC_W_relax}. {\color{black}In Section \ref{ssec:ECG_exp}}, we see that for a given level of $\ell_2$ reconstruction error, Algorithm \ref{alg:BNB} produces solutions that are on average $9.95\%$ more sparse than the solutions returned by state of the art benchmark methods after they are further sparsified by greedy rounding on the real world {\color{black}ECG} dataset we experiment with. Furthermore, for a given sparsity level, Algorithm \ref{alg:BNB} produces solutions that have on average $10.77\%$ lower $\ell_2$ reconstruction error than benchmark methods. {\color{black}Finally, in Section \ref{ssec:ML_exp}, we see that Algorithm \ref{alg:BNB} outperforms benchmark methods when used as the reconstruction algorithm for compressed sensing multi-label classification in terms of accuracy, precision and recall.}

\section{Conclusion}

In this paper, we introduced an $\ell_2$ regularized formulation \eqref{opt:our_formulation} for CS which emits a natural reformulation as a mixed integer second order cone program \eqref{opt:MISOC_W}. We presented a second order cone relaxation \eqref{opt:MISOC_W_relax} and a stronger but more expensive semidefinite cone relaxation \eqref{opt:poly_relax} to \eqref{opt:MISOC_W}. We presented Algorithm \ref{alg:BNB}, a custom branch-and-bound algorithm that can compute globally optimal solution for \eqref{opt:our_formulation}. We find that our approach produces solutions that are on average $6.22\%$ more sparse on synthetic data and $9.95\%$ more sparse on real world ECG data when compared to three state of the art benchmark approaches. {\color{black}This improvement in accuracy comes at the cost of an increase in computation time by several orders of magnitude.} When comparing only against the experiment-wise best performing benchmark method, our approach produces solutions that are on average $3.10\%$ more sparse on synthetic data and $3.88\%$ more sparse on real world ECG data. Moreover, our approach outperforms benchmark methods when used as part of a multi-label learning algorithm. Further work might focus on strengthening our convex relaxations by deriving additional valid inequalities for \eqref{opt:MISOC_W} or increasing the scalability of our branch-and-bound method. Algorithm \ref{alg:BNB} currently uses our second order cone relaxation to compute lower bounds. If fast problem specific solution methods could be derived for our positive semidefinite cone relaxation, employing the latter for lower bounds in Algorithm \ref{alg:BNB} could potentially lead to important scalability gains.

\section*{Declarations}

\paragraph{Funding:} The authors did not receive support from any organization for the submitted work.
\paragraph{Conflict of interest/Competing interests:} The authors have no relevant financial or non-financial interests to disclose.
\paragraph{Ethics approval:} Not applicable.
\paragraph{Consent to participate:} Not applicable.
\paragraph{Consent for publication:} Not applicable.
\paragraph{Availability of data and materials:} We obtained the ECG recording samples employed in Section \ref{ssec:ECG_exp} from the MIT-BIH Arrhythmia Database (\url{https://www.physionet.org/content/mitdb/1.0.0/}).
\paragraph{Code availability:} To bridge the gap between theory and practice, we have made our code freely available on \url{GitHub} at \url{github.com/NicholasJohnson2020/DiscreteCompressedSensing.jl}
\paragraph{Authors' contributions:} Both authors contributed to the algorithmic ideation and design. Algorithm implementation, data collection, simulation and data analysis was performed by Nicholas Johnson. The first draft of the manuscript was written by Nicholas Johnson and both authors commented and edited subsequent versions of the manuscript. Both authors read and approved the final manuscript.




\clearpage


\bibliography{sn-bibliography}


\begin{thebibliography}{55}
\ifx \bisbn   \undefined \def \bisbn  #1{ISBN #1}\fi
\ifx \binits  \undefined \def \binits#1{#1}\fi
\ifx \bauthor  \undefined \def \bauthor#1{#1}\fi
\ifx \batitle  \undefined \def \batitle#1{#1}\fi
\ifx \bjtitle  \undefined \def \bjtitle#1{#1}\fi
\ifx \bvolume  \undefined \def \bvolume#1{\textbf{#1}}\fi
\ifx \byear  \undefined \def \byear#1{#1}\fi
\ifx \bissue  \undefined \def \bissue#1{#1}\fi
\ifx \bfpage  \undefined \def \bfpage#1{#1}\fi
\ifx \blpage  \undefined \def \blpage #1{#1}\fi
\ifx \burl  \undefined \def \burl#1{\textsf{#1}}\fi
\ifx \doiurl  \undefined \def \doiurl#1{\url{https://doi.org/#1}}\fi
\ifx \betal  \undefined \def \betal{\textit{et al.}}\fi
\ifx \binstitute  \undefined \def \binstitute#1{#1}\fi
\ifx \binstitutionaled  \undefined \def \binstitutionaled#1{#1}\fi
\ifx \bctitle  \undefined \def \bctitle#1{#1}\fi
\ifx \beditor  \undefined \def \beditor#1{#1}\fi
\ifx \bpublisher  \undefined \def \bpublisher#1{#1}\fi
\ifx \bbtitle  \undefined \def \bbtitle#1{#1}\fi
\ifx \bedition  \undefined \def \bedition#1{#1}\fi
\ifx \bseriesno  \undefined \def \bseriesno#1{#1}\fi
\ifx \blocation  \undefined \def \blocation#1{#1}\fi
\ifx \bsertitle  \undefined \def \bsertitle#1{#1}\fi
\ifx \bsnm \undefined \def \bsnm#1{#1}\fi
\ifx \bsuffix \undefined \def \bsuffix#1{#1}\fi
\ifx \bparticle \undefined \def \bparticle#1{#1}\fi
\ifx \barticle \undefined \def \barticle#1{#1}\fi
\bibcommenthead
\ifx \bconfdate \undefined \def \bconfdate #1{#1}\fi
\ifx \botherref \undefined \def \botherref #1{#1}\fi
\ifx \url \undefined \def \url#1{\textsf{#1}}\fi
\ifx \bchapter \undefined \def \bchapter#1{#1}\fi
\ifx \bbook \undefined \def \bbook#1{#1}\fi
\ifx \bcomment \undefined \def \bcomment#1{#1}\fi
\ifx \oauthor \undefined \def \oauthor#1{#1}\fi
\ifx \citeauthoryear \undefined \def \citeauthoryear#1{#1}\fi
\ifx \endbibitem  \undefined \def \endbibitem {}\fi
\ifx \bconflocation  \undefined \def \bconflocation#1{#1}\fi
\ifx \arxivurl  \undefined \def \arxivurl#1{\textsf{#1}}\fi
\csname PreBibitemsHook\endcsname

\bibitem[\protect\citeauthoryear{Hsu et~al.}{2009}]{hsu2009multi}
\begin{botherref}
\oauthor{\bsnm{Hsu}, \binits{D.J.}},
\oauthor{\bsnm{Kakade}, \binits{S.M.}},
\oauthor{\bsnm{Langford}, \binits{J.}},
\oauthor{\bsnm{Zhang}, \binits{T.}}:
Multi-label prediction via compressed sensing.
Advances in neural information processing systems
\textbf{22}
(2009)
\end{botherref}
\endbibitem

\bibitem[\protect\citeauthoryear{Lustig et~al.}{2007}]{lustig2007sparse}
\begin{barticle}
\bauthor{\bsnm{Lustig}, \binits{M.}},
\bauthor{\bsnm{Donoho}, \binits{D.}},
\bauthor{\bsnm{Pauly}, \binits{J.M.}}:
\batitle{Sparse mri: The application of compressed sensing for rapid mr
  imaging}.
\bjtitle{Magnetic Resonance in Medicine: An Official Journal of the
  International Society for Magnetic Resonance in Medicine}
\bvolume{58}(\bissue{6}),
\bfpage{1182}--\blpage{1195}
(\byear{2007})
\end{barticle}
\endbibitem

\bibitem[\protect\citeauthoryear{Brady et~al.}{2009}]{brady2009compressive}
\begin{barticle}
\bauthor{\bsnm{Brady}, \binits{D.J.}},
\bauthor{\bsnm{Choi}, \binits{K.}},
\bauthor{\bsnm{Marks}, \binits{D.L.}},
\bauthor{\bsnm{Horisaki}, \binits{R.}},
\bauthor{\bsnm{Lim}, \binits{S.}}:
\batitle{Compressive holography}.
\bjtitle{Optics express}
\bvolume{17}(\bissue{15}),
\bfpage{13040}--\blpage{13049}
(\byear{2009})
\end{barticle}
\endbibitem

\bibitem[\protect\citeauthoryear{Hashemi et~al.}{2016}]{hashemi2016efficient}
\begin{bchapter}
\bauthor{\bsnm{Hashemi}, \binits{A.}},
\bauthor{\bsnm{Rostami}, \binits{M.}},
\bauthor{\bsnm{Cheung}, \binits{N.-M.}}:
\bctitle{Efficient environmental temperature monitoring using compressed
  sensing}.
In: \bbtitle{2016 Data Compression Conference (DCC)},
pp. \bfpage{602}--\blpage{602}
(\byear{2016}).
\bcomment{IEEE Computer Society}
\end{bchapter}
\endbibitem

\bibitem[\protect\citeauthoryear{Wang et~al.}{2018}]{wang2018design}
\begin{barticle}
\bauthor{\bsnm{Wang}, \binits{G.}},
\bauthor{\bsnm{Zhao}, \binits{Z.}},
\bauthor{\bsnm{Ning}, \binits{Y.}}:
\batitle{Design of compressed sensing algorithm for coal mine iot moving
  measurement data based on a multi-hop network and total variation}.
\bjtitle{Sensors}
\bvolume{18}(\bissue{6}),
\bfpage{1732}
(\byear{2018})
\end{barticle}
\endbibitem

\bibitem[\protect\citeauthoryear{Chen et~al.}{2019}]{chen2019compressed}
\begin{barticle}
\bauthor{\bsnm{Chen}, \binits{J.}},
\bauthor{\bsnm{Xing}, \binits{J.}},
\bauthor{\bsnm{Zhang}, \binits{L.Y.}},
\bauthor{\bsnm{Qi}, \binits{L.}}:
\batitle{Compressed sensing for electrocardiogram acquisition in wireless body
  sensor network: A comparative analysis}.
\bjtitle{International Journal of Distributed Sensor Networks}
\bvolume{15}(\bissue{7}),
\bfpage{1550147719864884}
(\byear{2019})
\end{barticle}
\endbibitem

\bibitem[\protect\citeauthoryear{Donoho}{2006}]{donoho2006compressed}
\begin{barticle}
\bauthor{\bsnm{Donoho}, \binits{D.L.}}:
\batitle{Compressed sensing}.
\bjtitle{IEEE Transactions on information theory}
\bvolume{52}(\bissue{4}),
\bfpage{1289}--\blpage{1306}
(\byear{2006})
\end{barticle}
\endbibitem

\bibitem[\protect\citeauthoryear{Tropp and
  Wright}{2010}]{tropp2010computational}
\begin{barticle}
\bauthor{\bsnm{Tropp}, \binits{J.A.}},
\bauthor{\bsnm{Wright}, \binits{S.J.}}:
\batitle{Computational methods for sparse solution of linear inverse problems}.
\bjtitle{Proceedings of the IEEE}
\bvolume{98}(\bissue{6}),
\bfpage{948}--\blpage{958}
(\byear{2010})
\end{barticle}
\endbibitem

\bibitem[\protect\citeauthoryear{Rani et~al.}{2018}]{rani2018systematic}
\begin{barticle}
\bauthor{\bsnm{Rani}, \binits{M.}},
\bauthor{\bsnm{Dhok}, \binits{S.B.}},
\bauthor{\bsnm{Deshmukh}, \binits{R.B.}}:
\batitle{A systematic review of compressive sensing: Concepts, implementations
  and applications}.
\bjtitle{IEEE access}
\bvolume{6},
\bfpage{4875}--\blpage{4894}
(\byear{2018})
\end{barticle}
\endbibitem

\bibitem[\protect\citeauthoryear{Owen and Perry}{2009}]{validation}
\begin{barticle}
\bauthor{\bsnm{Owen}, \binits{A.B.}},
\bauthor{\bsnm{Perry}, \binits{P.O.}}:
\batitle{{Bi-cross-validation of the SVD and the nonnegative matrix
  factorization}}.
\bjtitle{The Annals of Applied Statistics}
\bvolume{3}(\bissue{2}),
\bfpage{564}--\blpage{594}
(\byear{2009})
\end{barticle}
\endbibitem

\bibitem[\protect\citeauthoryear{Bousquet and
  Elisseeff}{2002}]{bousquet2002stability}
\begin{barticle}
\bauthor{\bsnm{Bousquet}, \binits{O.}},
\bauthor{\bsnm{Elisseeff}, \binits{A.}}:
\batitle{Stability and generalization}.
\bjtitle{Journal of Machine Learning Research}
\bvolume{2},
\bfpage{499}--\blpage{526}
(\byear{2002})
\end{barticle}
\endbibitem

\bibitem[\protect\citeauthoryear{Karahano{\u{g}}lu
  et~al.}{2013}]{karahanouglu2013mixed}
\begin{bchapter}
\bauthor{\bsnm{Karahano{\u{g}}lu}, \binits{N.B.}},
\bauthor{\bsnm{Erdo{\u{g}}an}, \binits{H.}},
\bauthor{\bsnm{Birbil}, \binits{{\c{S}}.{\. I}.}}:
\bctitle{A mixed integer linear programming formulation for the sparse recovery
  problem in compressed sensing}.
In: \bbtitle{2013 IEEE International Conference on Acoustics, Speech and Signal
  Processing},
pp. \bfpage{5870}--\blpage{5874}
(\byear{2013}).
\bcomment{IEEE}
\end{bchapter}
\endbibitem

\bibitem[\protect\citeauthoryear{Bourguignon
  et~al.}{2015}]{bourguignon2015exact}
\begin{barticle}
\bauthor{\bsnm{Bourguignon}, \binits{S.}},
\bauthor{\bsnm{Ninin}, \binits{J.}},
\bauthor{\bsnm{Carfantan}, \binits{H.}},
\bauthor{\bsnm{Mongeau}, \binits{M.}}:
\batitle{Exact sparse approximation problems via mixed-integer programming:
  Formulations and computational performance}.
\bjtitle{IEEE Transactions on Signal Processing}
\bvolume{64}(\bissue{6}),
\bfpage{1405}--\blpage{1419}
(\byear{2015})
\end{barticle}
\endbibitem

\bibitem[\protect\citeauthoryear{Chen and Donoho}{1994}]{chen1994basis}
\begin{bchapter}
\bauthor{\bsnm{Chen}, \binits{S.}},
\bauthor{\bsnm{Donoho}, \binits{D.}}:
\bctitle{Basis pursuit}.
In: \bbtitle{Proceedings of 1994 28th Asilomar Conference on Signals, Systems
  and Computers},
vol. \bseriesno{1},
pp. \bfpage{41}--\blpage{44}
(\byear{1994}).
\bcomment{IEEE}
\end{bchapter}
\endbibitem

\bibitem[\protect\citeauthoryear{Chen et~al.}{2001}]{chen2001atomic}
\begin{barticle}
\bauthor{\bsnm{Chen}, \binits{S.S.}},
\bauthor{\bsnm{Donoho}, \binits{D.L.}},
\bauthor{\bsnm{Saunders}, \binits{M.A.}}:
\batitle{Atomic decomposition by basis pursuit}.
\bjtitle{SIAM review}
\bvolume{43}(\bissue{1}),
\bfpage{129}--\blpage{159}
(\byear{2001})
\end{barticle}
\endbibitem

\bibitem[\protect\citeauthoryear{Candes and Tao}{2006}]{candes2006near}
\begin{barticle}
\bauthor{\bsnm{Candes}, \binits{E.J.}},
\bauthor{\bsnm{Tao}, \binits{T.}}:
\batitle{Near-optimal signal recovery from random projections: Universal
  encoding strategies?}
\bjtitle{IEEE transactions on information theory}
\bvolume{52}(\bissue{12}),
\bfpage{5406}--\blpage{5425}
(\byear{2006})
\end{barticle}
\endbibitem

\bibitem[\protect\citeauthoryear{Gill et~al.}{2011}]{gill2011crowd}
\begin{barticle}
\bauthor{\bsnm{Gill}, \binits{P.R.}},
\bauthor{\bsnm{Wang}, \binits{A.}},
\bauthor{\bsnm{Molnar}, \binits{A.}}:
\batitle{The in-crowd algorithm for fast basis pursuit denoising}.
\bjtitle{IEEE Transactions on Signal Processing}
\bvolume{59}(\bissue{10}),
\bfpage{4595}--\blpage{4605}
(\byear{2011})
\end{barticle}
\endbibitem

\bibitem[\protect\citeauthoryear{Tibshirani}{1996}]{tibshirani1996regression}
\begin{barticle}
\bauthor{\bsnm{Tibshirani}, \binits{R.}}:
\batitle{Regression shrinkage and selection via the lasso}.
\bjtitle{Journal of the Royal Statistical Society: Series B (Methodological)}
\bvolume{58}(\bissue{1}),
\bfpage{267}--\blpage{288}
(\byear{1996})
\end{barticle}
\endbibitem

\bibitem[\protect\citeauthoryear{Bertsimas and
  Copenhaver}{2018}]{bertsimas2018characterization}
\begin{barticle}
\bauthor{\bsnm{Bertsimas}, \binits{D.}},
\bauthor{\bsnm{Copenhaver}, \binits{M.S.}}:
\batitle{Characterization of the equivalence of robustification and
  regularization in linear and matrix regression}.
\bjtitle{European Journal of Operational Research}
\bvolume{270}(\bissue{3}),
\bfpage{931}--\blpage{942}
(\byear{2018})
\end{barticle}
\endbibitem

\bibitem[\protect\citeauthoryear{Elad and
  Bruckstein}{2002}]{elad2002generalized}
\begin{barticle}
\bauthor{\bsnm{Elad}, \binits{M.}},
\bauthor{\bsnm{Bruckstein}, \binits{A.M.}}:
\batitle{A generalized uncertainty principle and sparse representation in pairs
  of bases}.
\bjtitle{IEEE Transactions on Information Theory}
\bvolume{48}(\bissue{9}),
\bfpage{2558}--\blpage{2567}
(\byear{2002})
\end{barticle}
\endbibitem

\bibitem[\protect\citeauthoryear{Donoho and Elad}{2003}]{donoho2003optimally}
\begin{barticle}
\bauthor{\bsnm{Donoho}, \binits{D.L.}},
\bauthor{\bsnm{Elad}, \binits{M.}}:
\batitle{Optimally sparse representation in general (nonorthogonal)
  dictionaries via l1 minimization}.
\bjtitle{Proceedings of the National Academy of Sciences}
\bvolume{100}(\bissue{5}),
\bfpage{2197}--\blpage{2202}
(\byear{2003})
\end{barticle}
\endbibitem

\bibitem[\protect\citeauthoryear{Gribonval and
  Nielsen}{2003}]{gribonval2003sparse}
\begin{barticle}
\bauthor{\bsnm{Gribonval}, \binits{R.}},
\bauthor{\bsnm{Nielsen}, \binits{M.}}:
\batitle{Sparse representations in unions of bases}.
\bjtitle{IEEE transactions on Information theory}
\bvolume{49}(\bissue{12}),
\bfpage{3320}--\blpage{3325}
(\byear{2003})
\end{barticle}
\endbibitem

\bibitem[\protect\citeauthoryear{Tropp}{2004}]{tropp2004greed}
\begin{barticle}
\bauthor{\bsnm{Tropp}, \binits{J.A.}}:
\batitle{Greed is good: algorithmic results for sparse approximation}.
\bjtitle{IEEE Transactions on Information Theory}
\bvolume{50}(\bissue{10}),
\bfpage{2231}--\blpage{2242}
(\byear{2004})
\end{barticle}
\endbibitem

\bibitem[\protect\citeauthoryear{Candes and Tao}{2005}]{candes2005decoding}
\begin{barticle}
\bauthor{\bsnm{Candes}, \binits{E.J.}},
\bauthor{\bsnm{Tao}, \binits{T.}}:
\batitle{Decoding by linear programming}.
\bjtitle{IEEE transactions on information theory}
\bvolume{51}(\bissue{12}),
\bfpage{4203}--\blpage{4215}
(\byear{2005})
\end{barticle}
\endbibitem

\bibitem[\protect\citeauthoryear{Baraniuk et~al.}{2008}]{baraniuk2008simple}
\begin{barticle}
\bauthor{\bsnm{Baraniuk}, \binits{R.}},
\bauthor{\bsnm{Davenport}, \binits{M.}},
\bauthor{\bsnm{DeVore}, \binits{R.}},
\bauthor{\bsnm{Wakin}, \binits{M.}}:
\batitle{A simple proof of the restricted isometry property for random
  matrices}.
\bjtitle{Constructive Approximation}
\bvolume{28}(\bissue{3}),
\bfpage{253}--\blpage{263}
(\byear{2008})
\end{barticle}
\endbibitem

\bibitem[\protect\citeauthoryear{Gu{\'e}don
  et~al.}{2014}]{guedon2014restricted}
\begin{barticle}
\bauthor{\bsnm{Gu{\'e}don}, \binits{O.}},
\bauthor{\bsnm{Litvak}, \binits{A.E.}},
\bauthor{\bsnm{Pajor}, \binits{A.}},
\bauthor{\bsnm{Tomczak-Jaegermann}, \binits{N.}}:
\batitle{Restricted isometry property for random matrices with heavy-tailed
  columns}.
\bjtitle{Comptes Rendus Mathematique}
\bvolume{352}(\bissue{5}),
\bfpage{431}--\blpage{434}
(\byear{2014})
\end{barticle}
\endbibitem

\bibitem[\protect\citeauthoryear{Candes et~al.}{2008}]{candes2008enhancing}
\begin{barticle}
\bauthor{\bsnm{Candes}, \binits{E.J.}},
\bauthor{\bsnm{Wakin}, \binits{M.B.}},
\bauthor{\bsnm{Boyd}, \binits{S.P.}}:
\batitle{Enhancing sparsity by reweighted l1 minimization}.
\bjtitle{Journal of Fourier analysis and applications}
\bvolume{14}(\bissue{5}),
\bfpage{877}--\blpage{905}
(\byear{2008})
\end{barticle}
\endbibitem

\bibitem[\protect\citeauthoryear{Needell}{2009}]{needell2009noisy}
\begin{bchapter}
\bauthor{\bsnm{Needell}, \binits{D.}}:
\bctitle{Noisy signal recovery via iterative reweighted l1-minimization}.
In: \bbtitle{2009 Conference Record of the Forty-Third Asilomar Conference on
  Signals, Systems and Computers},
pp. \bfpage{113}--\blpage{117}
(\byear{2009}).
\bcomment{IEEE}
\end{bchapter}
\endbibitem

\bibitem[\protect\citeauthoryear{Asif and Romberg}{2013}]{asif2013fast}
\begin{barticle}
\bauthor{\bsnm{Asif}, \binits{M.S.}},
\bauthor{\bsnm{Romberg}, \binits{J.}}:
\batitle{Fast and accurate algorithms for re-weighted l1-norm minimization}.
\bjtitle{IEEE Transactions on Signal Processing}
\bvolume{61}(\bissue{23}),
\bfpage{5905}--\blpage{5916}
(\byear{2013})
\end{barticle}
\endbibitem

\bibitem[\protect\citeauthoryear{Chen and Zhou}{2010}]{chen2010convergence}
\begin{botherref}
\oauthor{\bsnm{Chen}, \binits{X.}},
\oauthor{\bsnm{Zhou}, \binits{W.}}:
Convergence of reweighted l1 minimization algorithms and unique solution of
  truncated lp minimization.
Department of Applied Mathematics, The Hong Kong Polytechnic University
(2010)
\end{botherref}
\endbibitem

\bibitem[\protect\citeauthoryear{Wang et~al.}{2021}]{wang2021nonconvex}
\begin{barticle}
\bauthor{\bsnm{Wang}, \binits{H.}},
\bauthor{\bsnm{Zhang}, \binits{F.}},
\bauthor{\bsnm{Shi}, \binits{Y.}},
\bauthor{\bsnm{Hu}, \binits{Y.}}:
\batitle{Nonconvex and nonsmooth sparse optimization via adaptively iterative
  reweighted methods}.
\bjtitle{Journal of Global Optimization}
\bvolume{81}(\bissue{3}),
\bfpage{717}--\blpage{748}
(\byear{2021})
\end{barticle}
\endbibitem

\bibitem[\protect\citeauthoryear{Wang et~al.}{2022}]{wang2022extrapolated}
\begin{botherref}
\oauthor{\bsnm{Wang}, \binits{H.}},
\oauthor{\bsnm{Zeng}, \binits{H.}},
\oauthor{\bsnm{Wang}, \binits{J.}}:
An extrapolated iteratively reweighted l1 method with complexity analysis.
Computational Optimization and Applications,
1--31
(2022)
\end{botherref}
\endbibitem

\bibitem[\protect\citeauthoryear{Pati et~al.}{1993}]{pati1993orthogonal}
\begin{bchapter}
\bauthor{\bsnm{Pati}, \binits{Y.C.}},
\bauthor{\bsnm{Rezaiifar}, \binits{R.}},
\bauthor{\bsnm{Krishnaprasad}, \binits{P.S.}}:
\bctitle{Orthogonal matching pursuit: Recursive function approximation with
  applications to wavelet decomposition}.
In: \bbtitle{Proceedings of 27th Asilomar Conference on Signals, Systems and
  Computers},
pp. \bfpage{40}--\blpage{44}
(\byear{1993}).
\bcomment{IEEE}
\end{bchapter}
\endbibitem

\bibitem[\protect\citeauthoryear{Mallat and Zhang}{1993}]{mallat1993matching}
\begin{barticle}
\bauthor{\bsnm{Mallat}, \binits{S.G.}},
\bauthor{\bsnm{Zhang}, \binits{Z.}}:
\batitle{Matching pursuits with time-frequency dictionaries}.
\bjtitle{IEEE Transactions on signal processing}
\bvolume{41}(\bissue{12}),
\bfpage{3397}--\blpage{3415}
(\byear{1993})
\end{barticle}
\endbibitem

\bibitem[\protect\citeauthoryear{Cai and Wang}{2011}]{cai2011orthogonal}
\begin{barticle}
\bauthor{\bsnm{Cai}, \binits{T.T.}},
\bauthor{\bsnm{Wang}, \binits{L.}}:
\batitle{Orthogonal matching pursuit for sparse signal recovery with noise}.
\bjtitle{IEEE Transactions on Information theory}
\bvolume{57}(\bissue{7}),
\bfpage{4680}--\blpage{4688}
(\byear{2011})
\end{barticle}
\endbibitem

\bibitem[\protect\citeauthoryear{Wang}{2015}]{wang2015support}
\begin{barticle}
\bauthor{\bsnm{Wang}, \binits{J.}}:
\batitle{Support recovery with orthogonal matching pursuit in the presence of
  noise}.
\bjtitle{IEEE Transactions on Signal processing}
\bvolume{63}(\bissue{21}),
\bfpage{5868}--\blpage{5877}
(\byear{2015})
\end{barticle}
\endbibitem

\bibitem[\protect\citeauthoryear{Dai and Milenkovic}{2009}]{dai2009subspace}
\begin{barticle}
\bauthor{\bsnm{Dai}, \binits{W.}},
\bauthor{\bsnm{Milenkovic}, \binits{O.}}:
\batitle{Subspace pursuit for compressive sensing signal reconstruction}.
\bjtitle{IEEE transactions on Information Theory}
\bvolume{55}(\bissue{5}),
\bfpage{2230}--\blpage{2249}
(\byear{2009})
\end{barticle}
\endbibitem

\bibitem[\protect\citeauthoryear{G{\"u}nl{\"u}k and
  Linderoth}{2012}]{gunluk2012perspective}
\begin{bchapter}
\bauthor{\bsnm{G{\"u}nl{\"u}k}, \binits{O.}},
\bauthor{\bsnm{Linderoth}, \binits{J.}}:
\bctitle{Perspective reformulation and applications}.
In: \bbtitle{Mixed Integer Nonlinear Programming},
pp. \bfpage{61}--\blpage{89}.
\bpublisher{Springer},
\blocation{USA}
(\byear{2012})
\end{bchapter}
\endbibitem

\bibitem[\protect\citeauthoryear{Lasserre}{2001a}]{lasserre2001explicit}
\begin{bchapter}
\bauthor{\bsnm{Lasserre}, \binits{J.B.}}:
\bctitle{An explicit exact sdp relaxation for nonlinear 0-1 programs}.
In: \bbtitle{Integer Programming and Combinatorial Optimization: 8th
  International IPCO Conference Utrecht, The Netherlands, June 13--15, 2001
  Proceedings 8},
pp. \bfpage{293}--\blpage{303}
(\byear{2001}).
\bcomment{Springer}
\end{bchapter}
\endbibitem

\bibitem[\protect\citeauthoryear{Lasserre}{2001b}]{lasserre2001global}
\begin{barticle}
\bauthor{\bsnm{Lasserre}, \binits{J.B.}}:
\batitle{Global optimization with polynomials and the problem of moments}.
\bjtitle{SIAM Journal on optimization}
\bvolume{11}(\bissue{3}),
\bfpage{796}--\blpage{817}
(\byear{2001})
\end{barticle}
\endbibitem

\bibitem[\protect\citeauthoryear{Lasserre}{2009}]{lasserre2009moments}
\begin{bbook}
\bauthor{\bsnm{Lasserre}, \binits{J.B.}}:
\bbtitle{Moments, Positive Polynomials and Their Applications}
vol. \bseriesno{1}.
\bpublisher{World Scientific},
\blocation{France}
(\byear{2009})
\end{bbook}
\endbibitem

\bibitem[\protect\citeauthoryear{Boyd et~al.}{2004}]{boyd2004convex}
\begin{bbook}
\bauthor{\bsnm{Boyd}, \binits{S.}},
\bauthor{\bsnm{Boyd}, \binits{S.P.}},
\bauthor{\bsnm{Vandenberghe}, \binits{L.}}:
\bbtitle{Convex Optimization}.
\bpublisher{Cambridge university press},
\blocation{USA}
(\byear{2004})
\end{bbook}
\endbibitem

\bibitem[\protect\citeauthoryear{Land and Doig}{2010}]{Land2010}
\begin{bbook}
\bauthor{\bsnm{Land}, \binits{A.H.}},
\bauthor{\bsnm{Doig}, \binits{A.G.}}:
\bbtitle{An Automatic Method for Solving Discrete Programming Problems},
pp. \bfpage{105}--\blpage{132}.
\bpublisher{Springer},
\blocation{Berlin, Heidelberg}
(\byear{2010})
\end{bbook}
\endbibitem

\bibitem[\protect\citeauthoryear{Little}{1966}]{little1966branch}
\begin{botherref}
\oauthor{\bsnm{Little}, \binits{J.D.}}:
Branch and bound methods for combinatorial problems.
(1966)
\end{botherref}
\endbibitem

\bibitem[\protect\citeauthoryear{Petersen et~al.}{2008}]{petersen2008matrix}
\begin{barticle}
\bauthor{\bsnm{Petersen}, \binits{K.B.}},
\bauthor{\bsnm{Pedersen}, \binits{M.S.}}, \betal:
\batitle{The matrix cookbook}.
\bjtitle{Technical University of Denmark}
\bvolume{7}(\bissue{15}),
\bfpage{510}
(\byear{2008})
\end{barticle}
\endbibitem

\bibitem[\protect\citeauthoryear{Bertsimas et~al.}{2023}]{bertsimas2023SLR}
\begin{barticle}
\bauthor{\bsnm{Bertsimas}, \binits{D.}},
\bauthor{\bsnm{Cory-Wright}, \binits{R.}},
\bauthor{\bsnm{Johnson}, \binits{N.A.}}:
\batitle{Sparse plus low rank matrix decomposition: A discrete optimization
  approach}.
\bjtitle{The Journal of Machine Learning Research}
\bvolume{24}(\bissue{1}),
\bfpage{12478}--\blpage{12528}
(\byear{2023})
\end{barticle}
\endbibitem

\bibitem[\protect\citeauthoryear{Morrison et~al.}{2016}]{MORRISON201679}
\begin{barticle}
\bauthor{\bsnm{Morrison}, \binits{D.R.}},
\bauthor{\bsnm{Jacobson}, \binits{S.H.}},
\bauthor{\bsnm{Sauppe}, \binits{J.J.}},
\bauthor{\bsnm{Sewell}, \binits{E.C.}}:
\batitle{Branch-and-bound algorithms: A survey of recent advances in searching,
  branching, and pruning}.
\bjtitle{Discrete Optimization}
\bvolume{19},
\bfpage{79}--\blpage{102}
(\byear{2016})
\end{barticle}
\endbibitem

\bibitem[\protect\citeauthoryear{Bertsimas and
  Digalakis~Jr}{2022}]{bertsimas2022backbone}
\begin{barticle}
\bauthor{\bsnm{Bertsimas}, \binits{D.}},
\bauthor{\bsnm{Digalakis~Jr}, \binits{V.}}:
\batitle{The backbone method for ultra-high dimensional sparse machine
  learning}.
\bjtitle{Machine Learning}
\bvolume{111}(\bissue{6}),
\bfpage{2161}--\blpage{2212}
(\byear{2022})
\end{barticle}
\endbibitem

\bibitem[\protect\citeauthoryear{Goodfellow et~al.}{2016}]{GoodBengCour16}
\begin{bbook}
\bauthor{\bsnm{Goodfellow}, \binits{I.J.}},
\bauthor{\bsnm{Bengio}, \binits{Y.}},
\bauthor{\bsnm{Courville}, \binits{A.}}:
\bbtitle{Deep Learning}.
\bpublisher{MIT Press},
\blocation{Cambridge, MA, USA}
(\byear{2016})
\end{bbook}
\endbibitem

\bibitem[\protect\citeauthoryear{Reuther et~al.}{2018}]{reuther2018interactive}
\begin{bchapter}
\bauthor{\bsnm{Reuther}, \binits{A.}},
\bauthor{\bsnm{Kepner}, \binits{J.}},
\bauthor{\bsnm{Byun}, \binits{C.}},
\bauthor{\bsnm{Samsi}, \binits{S.}},
\bauthor{\bsnm{Arcand}, \binits{W.}},
\bauthor{\bsnm{Bestor}, \binits{D.}},
\bauthor{\bsnm{Bergeron}, \binits{B.}},
\bauthor{\bsnm{Gadepally}, \binits{V.}},
\bauthor{\bsnm{Houle}, \binits{M.}},
\bauthor{\bsnm{Hubbell}, \binits{M.}},
\bauthor{\bsnm{Jones}, \binits{M.}},
\bauthor{\bsnm{Klein}, \binits{A.}},
\bauthor{\bsnm{Milechin}, \binits{L.}},
\bauthor{\bsnm{Mullen}, \binits{J.}},
\bauthor{\bsnm{Prout}, \binits{A.}},
\bauthor{\bsnm{Rosa}, \binits{A.}},
\bauthor{\bsnm{Yee}, \binits{C.}},
\bauthor{\bsnm{Michaleas}, \binits{P.}}:
\bctitle{Interactive supercomputing on 40,000 cores for machine learning and
  data analysis}.
In: \bbtitle{2018 IEEE High Performance Extreme Computing Conference (HPEC)},
pp. \bfpage{1}--\blpage{6}
(\byear{2018}).
\bcomment{IEEE}
\end{bchapter}
\endbibitem

\bibitem[\protect\citeauthoryear{Aharon et~al.}{2006}]{aharon2006k}
\begin{barticle}
\bauthor{\bsnm{Aharon}, \binits{M.}},
\bauthor{\bsnm{Elad}, \binits{M.}},
\bauthor{\bsnm{Bruckstein}, \binits{A.}}:
\batitle{K-svd: An algorithm for designing overcomplete dictionaries for sparse
  representation}.
\bjtitle{IEEE Transactions on signal processing}
\bvolume{54}(\bissue{11}),
\bfpage{4311}--\blpage{4322}
(\byear{2006})
\end{barticle}
\endbibitem

\bibitem[\protect\citeauthoryear{Kapoor et~al.}{2012}]{kapoor2012multilabel}
\begin{botherref}
\oauthor{\bsnm{Kapoor}, \binits{A.}},
\oauthor{\bsnm{Viswanathan}, \binits{R.}},
\oauthor{\bsnm{Jain}, \binits{P.}}:
Multilabel classification using bayesian compressed sensing.
Advances in neural information processing systems
\textbf{25}
(2012)
\end{botherref}
\endbibitem

\bibitem[\protect\citeauthoryear{Som}{2016}]{som2016learning}
\begin{bchapter}
\bauthor{\bsnm{Som}, \binits{S.}}:
\bctitle{Learning label structure for compressed sensing based multilabel
  classification}.
In: \bbtitle{2016 SAI Computing Conference (SAI)},
pp. \bfpage{54}--\blpage{60}
(\byear{2016}).
\bcomment{IEEE}
\end{bchapter}
\endbibitem

\bibitem[\protect\citeauthoryear{Kai et~al.}{2017}]{kai2017compressed}
\begin{bchapter}
\bauthor{\bsnm{Kai}, \binits{W.}},
\bauthor{\bsnm{Jingzhi}, \binits{L.}},
\bauthor{\bsnm{Yanjun}, \binits{C.}}:
\bctitle{Compressed sensing based multi-label classification without label
  sparsity level prior}.
In: \bbtitle{Proceedings of the 2017 International Conference on Deep Learning
  Technologies},
pp. \bfpage{66}--\blpage{69}
(\byear{2017})
\end{bchapter}
\endbibitem

\bibitem[\protect\citeauthoryear{Tsoumakas
  et~al.}{2008}]{tsoumakas2008effective}
\begin{bchapter}
\bauthor{\bsnm{Tsoumakas}, \binits{G.}},
\bauthor{\bsnm{Katakis}, \binits{I.}},
\bauthor{\bsnm{Vlahavas}, \binits{I.}}:
\bctitle{Effective and efficient multilabel classification in domains with
  large number of labels}.
In: \bbtitle{Proc. ECML/PKDD 2008 Workshop on Mining Multidimensional Data
  (MMD’08)},
vol. \bseriesno{21},
pp. \bfpage{53}--\blpage{59}
(\byear{2008})
\end{bchapter}
\endbibitem

\end{thebibliography}

\end{document}